\RequirePackage[2020-02-02]{latexrelease}

\documentclass[british]{rQUF2e}

\usepackage{epstopdf}% To incorporate .eps illustrations using PDFLaTeX, etc.
\usepackage{subfigure}% Support for small, `sub' figures and tables

\usepackage{accents, bbm,enumitem, color,csquotes,  mathtools, footnote, xcolor,url}
\usepackage{marginnote}
\usepackage{float}
\usepackage{babel}

\ifpdf
\DeclareGraphicsExtensions{.eps,.pdf,.png,.jpg}
\else
\DeclareGraphicsExtensions{.eps}
\fi

\newcommand{\Ap}{A_{\pi}}
\newcommand{\Bp}{B_{\pi}}
\newcommand{\Jp}{J^{\pi}}
\newcommand{\E}{\mathbbm{E}}
\newcommand{\N}{\mathbbm{N}}
\newcommand{\R}{\mathbbm{R}}

\newcommand{\1}{\mathbbm{1}}

\newcommand{\la}{\lambda^{\ast}}
\newcommand{\Rex}{\bar{\R}}

\newcommand{\diag}{\text{diag}}

\theoremstyle{plain}
\newtheorem{theorem}{Theorem}[section]
\newtheorem{corollary}[theorem]{Corollary}
\newtheorem{lemma}[theorem]{Lemma}

\newtheorem{assumption}[theorem]{Assumption}

\theoremstyle{definition}

\theoremstyle{remark}
\newtheorem{remark}{Remark}

\allowdisplaybreaks

\begin{document}

%\jvol{00} \jnum{00} \jyear{2014} \jmonth{October}

\title{Mind the Cap! - Constrained Portfolio Optimisation in Heston's Stochastic Volatility Model}

\author{M. ESCOBAR-ANEL$\dag$, M. KSCHONNEK$^{\ast}$$\ddag$\thanks{$^\ast$Corresponding author.
Email: michel.kschonnek@tum.de} and R. ZAGST$\ddag$\\
\affil{$\dag$ Department of Statistical and Actuarial Sciences, Western University, London, ON, Canada \\
$\ddag$ Chair of Mathematical Finance, Technical University of Munich, D-80333 Munich, Germany \received{\today} }}

\maketitle

\begin{abstract}
We consider a portfolio optimisation problem for a utility-maximising investor who faces convex constraints on his portfolio allocation in Heston's stochastic volatility model. We apply the duality methods developed in \citet{EKZ2023} to obtain a closed-form expression for the optimal portfolio allocation. In doing so, we observe that allocation constraints impact the optimal constrained portfolio allocation in a fundamentally different way in Heston's stochastic volatility model than in the Black Scholes model. In particular, the optimal constrained portfolio may be different from the naive \enquote{capped} portfolio, which caps off the optimal unconstrained portfolio at the boundaries of the constraints. Despite this difference, we illustrate by way of a numerical analysis that in most realistic scenarios the capped portfolio leads to slim annual wealth equivalent losses compared to the optimal constrained portfolio. During a financial crisis, however, a capped solution might lead to compelling annual wealth equivalent losses.
\end{abstract}


\begin{classcode}
	G11, C61
\end{classcode}

\section{Introduction}\label{sec: Introduction}

Despite its widespread use in both academia and the financial industry, it has been well-documented in the mathematical finance literature that a variety of properties of financial time series, so-called stylised facts, are not captured by the Black-Scholes model (see e.g. \citet{Cont2000}, \citet{Lux2008}, \citet{Weatherall2018}). A major point of criticism is the constant  volatility of modelled log returns, whereas empirical returns appear to be time-dependent and random (\citet{Taylor1994}). The stochastic volatility model proposed by \citet{Heston1993} aims to avoid this gap by modelling the volatility of log returns as a Cox-Ingersoll-Ross process (CIR process). While the Heston model was originally proposed in the context of option pricing, its analytical tractability has led to insightful applications in continuous-time portfolio optimisation (see e.g. \citet{Kraft2005}), which is the subject of this article. \\
Specifically, we consider the portfolio optimisation problem of an investor trading continuously in a financial market consisting of one risk-free asset and one risky asset with stochastic Heston volatility. The investor seeks to maximise his expected utility derived from terminal wealth at a finite time point $T>0$ under the condition that his portfolio allocation $\pi$ abides by given convex allocation constraints $K.$ The considered optimisation problem involves two major facets, which differ from the original portfolio optimisation setting of \citet{merton1971}: stochasticity of the volatility of risky asset log returns (\enquote{stochastic volatility}) and the presence of convex allocation constraints. The following paragraphs present a brief overview of the relevant literature, with respect to both of these facets.\\

{\emph{(i) Portfolio Optimisation in Financial Markets with Stochastic Volatility.}}\\

Our continuous time set-up can be traced back to the seminal work of \citet{merton1971}, who used dynamic programming methods to derive explicit solutions to the unconstrained dynamic portfolio optimisation problem for an investor with HARA utility function in a Black-Scholes model with constant volatility. Generalisations for different utility functions and more complex financial markets were achieved by employing martingale methods (\citet{Pliska1986} and \citet{karatzas1987}), which rely heavily on the assumption that all contingent claims in the financial market are replicable. However, this assumption is not generally satisfied for financial markets with stochastic volatility, which means that martingale methods are not directly applicable, unless the financial market is artificially completed by the addition of fictitious, volatility-dependent assets (see e.g. \citet{Liu2003}, \citet{Branger2008}, \citet{Egloff2010}, \citet{Rubtsov2017}, and \citet{Chen2021}). Without such completion, the solvability of the portfolio optimisation problem in financial markets with stochastic volatility is often directly linked to the solvability of the associated HJB PDE.\footnote{Note that obtaining and formally verifying the optimality of a candidate portfolio process requires more than just a solution to the associated HJB PDE, as pointed out by \citet{Korn2004}.} Solutions to such portfolio optimisation problems were first characterised in terms of viscosity solutions to the associated HJB PDE for multi-factor models in \citet{Zariphopoulou2001}, and explicit closed-form solutions for Heston's stochastic volatility model were first derived and formally verified in \citet{Kraft2005}. Subsequently, \citet{Liu2006} derived explicit solution formulae for an optimal consumption and portfolio allocation problem in a general multi-factor model, where the factor is a quadratic diffusion. \citet{Kallsen2010a} used the notion of an opportunity process and semi-martingale characteristics to develop an approach which leads to closed-form solutions for the optimal portfolio process in a range of exponentially affine stochastic volatility models, including the Heston model and the jump model of \citet{Carr2003}. The PCSV model and the Wishart process were considered as multi-dimensional extensions of the Heston model for multidimensional asset universes with stochastic correlation in \citet{Rubtsov2017} and \citet{baeuerle2013}. An extensive overview of related papers in the field of dynamic portfolio optimisation in stochastic factor models can be found in \citet{Zariphopoulou2009}. \\

{\emph{(ii) Portfolio Optimisation in Financial Markets with Convex Allocation Constraints.}}\\

Convex allocation constraints were first considered in the context of continuous-time dynamic portfolio optimisation in \citet{karatzas1991} and \citet{cvitanic1992}. The authors showed that a solution to the primal constrained portfolio optimisation problem could be derived by considering a family of \enquote{auxiliary} unconstrained portfolio optimisation problems, determining the optimal portfolios for these problems, and selecting the least favourable portfolio among them. For CRRA-utility functions, \citet{cvitanic1992} determined the optimal constrained portfolio process in closed-form up to a deterministic minimiser of a real convex optimisation problem. \citet{Zariph1994} considered a similar setting in a one-dimensional Black Scholes market, but did not employ any duality techniques. Rather, the author characterised the value function as the unique viscosity solution of the associated HJB PDE and gave a semi-explicit expression of the optimal portfolio allocation in terms of the value function. From both \citet{cvitanic1992} and \citet{Zariph1994}, one can easily see that the optimal constrained portfolio allocation for an investor with a CRRA utility function in the Black-Scholes model is equal to the unconstrained optimal solution if the constraints are satisfied, and is otherwise capped at the boundary of the constraint. Ever since, allocation constraints have been integrated into portfolio optimisation problems in a myriad of ways (see e.g. \citet{Cuoco1997}, \citet{Pham2002}, \citet{Lindberg2006a}, \citet{Mnif2007}, \citet{Zheng2011}, \citet{Nutz2012} and \citet{Dong2020}). Regardless, explicit solutions for the optimal portfolio allocation were obtained only on rare occasions. The availability of explicit expressions appears to strongly depend on the interplay between the chosen financial model, utility function and constraint. In this spirit, \citet{EKZ2023} recently developed a duality approach, which enabled them to state a condition under which the exponentially affine separability structure of the value function (as discussed in \citet{Liu2006}) is retained under the addition of convex allocation constraints. If this condition is satisfied, then a candidate for the optimal constrained portfolio allocation is known up to the solution of a Riccati ODE and a deterministic minimiser of a convex optimisation problem. Further, the authors showed that this condition is satisfied for Heston's stochastic volatility model, but failed to determine explicit solutions to the associated ODEs and did not formally verify the optimality of the candidate portfolio allocation.\\

In this paper, we make a threefold contribution to this literature:
\begin{itemize}
	\item We complement the work of \citet{EKZ2023} by deriving the first explicit, closed-form expression for the optimal portfolio allocation $\pi^{\ast}$ in Heston's stochastic volatility model under the presence of convex allocation constraints. This optimality is verified formally in a verification theorem.
	\item We show that the optimal portfolio allocation $\pi^{\ast}$ may be different from the capped optimal unconstrained portfolio allocation $\pi_u$ for Heston's stochastic volatility model. In particular, we prove an equivalent characterisation which describes when these two portfolio are different. We conduct a numerical study to show that this difference is slim for calm market scenarios, but can lead to significant annual wealth equivalent losses during turbulent market scenarios.
	\item We extend these results from the one-dimensional Heston's stochastic volatility model to a multi-dimensional PCSV model with constraints on the exposures to individual stochastic market factors and to generalised financial markets with inverse volatility constraints.
\end{itemize}

The remainder of this paper is structured as follows: In Section \ref{sec: Heston's Stochastic Volatility Model}, we introduce and solve the constrained portfolio optimisation problem $\mathbf{(P)}$ in Heston's stochastic volatility model. Specifically, we derive a solution to the HJB PDE associated with $\mathbf{(P)}$ and the candidate optimal portfolio $\pi^{\ast}$ in Section \ref{subsec: Solution to HJB PDE}, verify its optimality formally by proving a verification theorem in Section \ref{subsec: Verification Theorem for Heston's Model} and discuss its relation to the optimal unconstrained portfolio $\pi_u$ in Section \ref{subsec: Comparison to Unconstrained Portfolio}. In Section \ref{sec: Implications for Related Models}, we consider a generalised financial market model which depends on a multi-dimensional CIR process and derive the optimal portfolio in the PCSV model under exposure constraints (Section \ref{subsec: PCSV Model}) and in general financial markets with inverse volatility constraints (Section \ref{subsec: Inverse Volatility Constraints}). In Section \ref{sec: Numerical Studies}, we illustrate our theoretical results for Heston's stochastic volatility model in a numerical analysis, where we analyse the wealth equivalent loss of the optimal constrained portfolio for a Black Scholes model (Section \ref{subsec: Optimal Constrained Merton}) and the capped optimal unconstrained portfolio for Heston's stochastic volatility model (Section \ref{subsec: Capped Optimal Unconstrained Heston Portfolio}). Section \ref{sec: Conclusion} concludes the paper. All proofs relating to this paper can be found in Appendix \ref{sec: App. Proofs}.

\section{Heston's Stochastic Volatility Model}\label{sec: Heston's Stochastic Volatility Model}
We consider a finite time horizon $T>0$ and a complete, filtered probability space $(\Omega, \mathcal{F}_T,\mathbbm{F} = (\mathcal{F}_t)_{t \in [0,T]}, Q)$, in which the filtration $\mathbbm{F}$ is generated by two independent Wiener processes $\big(\hat{W},W^z) = \big(\hat{W}(t),W^z(t)\big)_{t \in [0,T]}.$ We define a financial market $\mathcal{M}_H$, consisting of one risk-free asset $P_0$ and one risky asset $P_1$, where the risky asset's instantaneous variance is driven by a CIR process $z$. %\footnote{This volatility model was first introduced in an option pricing context in \citet{Heston1993} and is thus commonly referred to as Heston's stochastic volatility model.} 
Specifically, we assume a market consisting of one risk-free asset $P_0$ and one risky asset $P_1$, which satisfy $P_0(0)=P_1(0)=1$ and follow the dynamics
%%%%&& 1-D Heston Volatility Market Coeffs %%%%%%%%
\begin{align*}
	dP_0(t) &= P_0(t) r dt,\\
	dP_1(t) &= P_1(t) \big(r + \eta \cdot z(t)\big)dt +  \sqrt{z(t)}\underbrace{\left(\rho dW^z(t)+\sqrt{1-\rho^2} d\hat{W}(t)\right)}_{=:dW(t)},
\end{align*}
where $z(0)=z_0>0$ and 
\begin{align*}
	dz(t) &= \kappa(\theta-z(t))dt+\sigma\sqrt{z(t)}dW^z(t).
\end{align*}

The market coefficients $r, \eta,\kappa, \theta, \sigma,z_0$ are assumed to be positive constants, and $\rho \in (-1,1).$ Lastly, it is assumed that Feller's condition holds for the parameters of $z$, i.e. $2\kappa \theta > \sigma^2,$ and, therefore, $z(t)$ is guaranteed to take only positive values with probability 1 (see \citet{Gikhman2011}). The wealth process $V^{v_0,\pi}$ of an investor trading in $\mathcal{M}_H$ according to a relative portfolio process $\pi$ and initial wealth $v_0>0$ now satisfies the usual SDE
\begin{align}\label{eq: SDE V original market}
	dV^{v_0, \pi}(t) = V^{v_0, \pi}(t) \Big( [r + (\eta z(t)\pi(t)]dt + \pi(t)\cdot \sqrt{z(t)}dW(t)\Big).
\end{align}

In this context, the portfolio process $\pi= \big(\pi(t)\big)_{t\in [0,T]}$ is one-dimensional, and represents the fraction of wealth invested in the risky asset $P_1$ at time $t$. The remaining fraction $1-\pi(t)$ is invested in the risk-free asset $P_0$. We restrict our analysis to the portfolio processes $\pi$, which guarantees that a unique solution to (\ref{eq: SDE V original market}) exists, i.e. to $\pi$ in 
\begin{align}\label{eq: def. admissible strategy Heston}
	\Lambda = \left \{ \pi = \big(\pi(t)\big)_{t \in [0,T]} \ \text{progr. measurable} \  \Big | \  \int_0^T z(t) \big(\pi(t) \big)^2 dt < \infty \ Q-a.s.\right \} 
\end{align}
If $\pi \in \Lambda$, it is straightforward to show that the unique solution $V^{v_0,\pi}(t)$ to (\ref{eq: def. admissible strategy Heston}) can be expressed in closed-form as 
\begin{align}\label{eq: wealth process explicitly original market}
	V^{v_0, \pi}(t) = v_0 \exp \Big( \int_0^t r+ \eta z(s) \pi(s)- \frac{1}{2}z(s)\pi(s)^2 ds + \int_0^t \pi(s)\sqrt{z(s)} dW(s)\Big). \
\end{align}

For a closed convex set $K\subset \R\cup\{\infty, -\infty\}=:\Rex$ with non-empty interior and CRRA utility function $U(v)= \frac{1}{b}v^b$ with $b<1$ and $b \neq 0,$ we consider the portfolio optimisation problem
\begin{equation*}
	\mathbf{(P)}
	\begin{cases}
		\Phi(v_0) &= \underset{\pi \in \Lambda_K}{\sup} \mathbbm{E}\big[U(V^{v_0, \pi}(T)) \big] \\
		\Lambda_K &= \left \{ \pi \in \Lambda \ \big | \ \pi(t) \in K \ \mathcal{L}[0,T]\otimes Q-\text{a.e.} \right \}.
	\end{cases}
\end{equation*}
As the considered financial market contains only one risky asset, the set of allocation constraints $K$ is a subset of the extended real numbers $\Rex$. However, in this one-dimensional setting, any such closed convex set $K\subset \Rex$ with non-empty interior can be expressed as an interval of the form\footnote{As any $\pi \in \Lambda$ can only take finite values $\mathcal{L}[0,T]\otimes Q$-a.s., we do not need to distinguish between $(-\infty,\beta]$ and $[-\infty,\beta]$ or $[\alpha,\infty)$ and $[\alpha,\infty]$ for any $-\infty \leq \alpha, \beta \leq \infty$.}
\begin{align}\label{eq: 1-d convex constraints as interval}
	K = [\alpha,\beta], \quad \text{with} \quad -\infty \leq \alpha < \beta \leq \infty. 
\end{align}
We make substantial use of this in the subsequent analysis.

\subsection{Solution to HJB PDE}\label{subsec: Solution to HJB PDE}
We approach $(\mathbf{P})$ using classic stochastic optimal control methods. For this purpose, let us introduce the generalised primal portfolio optimisation problem $\mathbf{(P^{(t,v,z)})}$ as 
\begin{equation*}
	\mathbf{(P^{(t,v,z)})}
	\begin{cases}
		\Phi(t,v,z) &= \underset{\pi \in \Lambda_K(t)}{\sup} \mathbbm{E}\big[U(V^{v_0, \pi}(T))\  \big | \ V^{v_0, \pi}(t) = v, \ z(t) = z \big] \\
		\Lambda_K(t) &= \left \{ \left(\pi(s)\right)_{s\in [t,T]} \ \big | \ \pi \in \Lambda_K \right \}.
	\end{cases}
\end{equation*} 

Then, the HJB equation associated with $\mathbf{(P^{(t,v,z)})}$ is given by
\begin{alignat}{2}\label{eq: constr. HJB PDE Heston}
	0&= \sup_{\pi \in K}\Big( && G_t + v(r+\eta  z  \pi)G_v +  \kappa (\theta-z) G_{z} + \frac{1}{2}v^2z \pi^2 G_{vv} +\rho v z \pi  \sigma G_{vz}+\frac{1}{2}\sigma^2 z G_{z z} \Big)  \\
	G(T,v,z) &= U(v), \label{eq: terminal condition constr. HJB equation}
\end{alignat}
Any (sufficiently regular) solution $G$ to (\ref{eq: constr. HJB PDE Heston}) naturally yields a candidate optimal portfolio $\pi^{\ast}$ to $\mathbf{(P^{(t,v,z)})}$ (and therefore $\mathbf{(P)}$) as the maximising argument of (\ref{eq: constr. HJB PDE Heston}). In their work on dynamic stochastic-factor models, \citet{EKZ2023} characterised $G$ as an exponentially affine function, whose exponents satisfy certain Riccati ODEs. To recall their results, we need to introduce the support function $\delta_K:\R\rightarrow \Rex$ of $K$ as 
\begin{align*}
	\delta_K(x) = -\inf_{y \in K}\left(x \cdot y\right) \overset{K = [\alpha, \beta]}{=} - \alpha x \1_{\{x>0\}} - \beta x \1_{\{x<0\}}.
\end{align*}
\begin{lemma}[]\label{lem: dual ODEs Heston Model MK23}
	Let $A$ and $B$ be solutions to the system of ODEs
	\begin{align}
		A'(\tau) &= br +\kappa \theta  B(\tau), \label{eq:  ODE A constr. Heston} \\
		B'(\tau) &= -\kappa  B(\tau) + \frac{1}{2}\sigma^2  \left(B(\tau)\right)^2  + \frac{1}{2}\frac{b}{1-b}\inf_{\lambda\in \R}\left(2(1-b)\delta_K(\lambda) + \left(\eta + \lambda + \sigma \rho  B(\tau) \right)^2  \right), \label{eq:  ODE B constr. Heston} 
	\end{align}
	with initial condition $A(0)=B(0) = 0.$ Then, $G(t,v,z) = \frac{1}{b}v^b\exp \left(A(T-t) + B(T-t)z\right)$ is a solution to (\ref{eq: constr. HJB PDE Heston}).
\end{lemma}
Given the minimising argument and the solution $B$ from (\ref{eq:  ODE B constr. Heston}), a candidate optimal portfolio $\pi^{\ast}$ for $(\mathbf{P})$ is known. If we define the sequence of stopping times $\tau_{n,t}$ as $\tau_{n,t}=\min(T,\hat{\tau}_{n,t})$, with
\begin{align*}
	\hat{\tau}_{n,t} = \inf \Big \{t\leq u \leq T \ \Big | \ &\int_t^u \left(b \cdot  \sqrt{z(s)}\cdot \pi(s)\cdot G(s,V^{v_0, \pi^{\ast}}(s),z(s))\right)^2ds \geq n, \\ 
	&\int_t^u  \left(\sigma \sqrt{z(s)} \cdot B(T-s)\cdot G(s,V^{v_0, \pi^{\ast}}(s),z(s))\right)^2ds \geq n \Big \},
\end{align*}
then we can give a uniform integrability condition which guarantees that the candidate optimal portfolio $\pi^{\ast}$ is indeed optimal for $(\mathbf{P}).$
\begin{lemma}\label{lem: optimal portfolio Heston MK23}
	Consider $A, B$ and $G$ from Lemma \ref{lem: dual ODEs Heston Model MK23}. Moreover, define
	\begin{align}\label{eq: optimal lambda Heston}
		\la(B) = \underset{\lambda \in \R}{\text{argmin}}\left \{ 2(1-b)\delta_K(\lambda) + \left(\eta + \lambda + \sigma \rho B \right)^2  \right\}
	\end{align}
	\begin{align}\label{eq: optimal portfolio Heston implicit}
		\pi^{\ast}(t) &= \frac{1}{1-b}\Big(\eta + \la(B(T-t)) + \sigma \rho B(T-t)\Big).
	\end{align}
	If $\left(G\left(\tau_{n,t}, V^{v_0, \pi^{\ast}}(\tau_{n,t}), z(\tau_{n,t})\right) \right)_{n \in \N}$
	is uniformly integrable for every $t \in [0,T],$ then $\pi^{\ast}$ is optimal for $\mathbf{(P)}$ and $\Phi(t,z,v) = G(t,v,z)$ for all $(t,v,z)\in [0,T]\times (0,\infty)\times(0,\infty).$
\end{lemma} 
Lemma \ref{lem: dual ODEs Heston Model MK23} and Lemma \ref{lem: optimal portfolio Heston MK23} naturally lead to a three-step procedure for finding the optimal portfolio $\pi^{\ast}$ for $\mathbf{(P)}$:
\begin{itemize}
	\item[(i)] Determine the minimising argument $\la$ in (\ref{eq: optimal lambda Heston}).
	\item[(ii)] Determine the solution $B$ to ODE (\ref{eq:  ODE B constr. Heston})  and thereby the candidate optimal portfolio $\pi^{\ast}$ from (\ref{eq: optimal portfolio Heston implicit}).
	\item[(iii)] Verify that $\pi^{\ast}$ satisfies the uniform integrability condition from Lemma \ref{lem: optimal portfolio Heston MK23}.
\end{itemize}
In the following, we complete these steps consecutively and complete the results of \citet{EKZ2023} by providing a fully closed-form solution for the optimal allocation-constrained portfolio in Heston's stochastic volatility model via steps (i) and (ii) and formally verifying its optimality in step (iii). \\

As $K$ is an interval, as specified in (\ref{eq: 1-d convex constraints as interval}), the ODE (\ref{eq:  ODE B constr. Heston}) can be written as a composition of three Riccati ODEs - each with constant coefficients. 

\begin{lemma}\label{lem: ODE B for Box constraints}
	Define $B_{-} = \frac{(1-b)\alpha-\eta}{\sigma}$ and $B_{+} = \frac{(1-b)\beta-\eta}{\sigma}$. Then, the minimising argument $\la,$ as in (\ref{eq: optimal lambda Heston}), is given as
	\begin{align}\label{eq: optimal lambda Heston explicit}
		\la(B) =& \big[(1-b)\alpha - \left( \eta+\sigma \rho B\right) \big]\1_{\{\rho B < B_- \}}+\big[(1-b)\beta - \left( \eta+\sigma \rho B\right) \big]\1_{\{\rho B > B_+ \}}.
	\end{align}
	Moreover, $B(\tau)$ is a solution to $(\ref{eq:  ODE B constr. Heston})$ if and only if  $B(0)= 0$ and
	\begin{align}\label{eq: ODE zones}
			B'(\tau) = \quad &\Big(- \underbrace{\frac{1}{2}b\alpha\big((1-b)\alpha-2 \eta\big)}_{=:r_0^-} + \underbrace{\big(b\sigma \rho \alpha-\kappa\big)}_{=:r_1^-}B(\tau) + \frac{1}{2}\underbrace{\sigma^2}_{=:r_2^-}\big(B(\tau)\big)^2\Big)\1_{\{ \rho B(\tau) < B_{-} \}}  \nonumber \\
			\quad  +& \Big( -\underbrace{\frac{-b}{2(1-b)}\eta^2}_{=:r_0} + \underbrace{\big(\frac{b}{1-b}\eta \sigma \rho-\kappa\big)}_{=:r_1}B(\tau) + \frac{1}{2}\underbrace{\sigma^2\big(1+\frac{b}{1-b}\rho^2\big)}_{=:r_2}\big(B(\tau)\big)^2\Big)\1_{\{ B_{-} \leq \rho B(\tau)\leq B_{+}\}}  \nonumber  \\
			\quad  +&\Big(- \underbrace{ \frac{1}{2}b\beta\big((1-b)\beta-2 \eta\big)}_{=:r_0^+} + \underbrace{\big(b\sigma \rho \beta-\kappa\big)}_{=:r_1^+}B(\tau) + \frac{1}{2}\underbrace{\sigma^2}_{=:r_2^+}\big(B(\tau)\big)^2\Big)\1_{\{B_{+} < \rho B(\tau) \}}  \\
			= \quad &\Big(-r_0^- + r_1^- B(\tau) + \frac{1}{2}r_2^-+\big(B(\tau)\big)^2\Big)\1_{\{ \rho B(\tau) < B_{-} \}}  \nonumber  \\
			\quad + &\Big( -r_0 + r_1 B(\tau) + \frac{1}{2}r_2\big(B(\tau)\big)^2\Big)\1_{\{ B_{-} \leq \rho B(\tau)\leq B_{+}\}}  \nonumber \\
			\quad + &\Big(-r_0^+ +r_1^+ B(\tau) + \frac{1}{2}r_2^+\big(B(\tau)\big)^2 \Big)\1_{\{B_{+} < \rho B(\tau) \}}. \nonumber 
	\end{align}
\end{lemma}

	\begin{remark}\label{rem: monotonicty of B}
		By restricting the minimisation in ODE (\ref{eq:  ODE B constr. Heston}) from $\lambda \in \R$ to one of the three optimal values $\lambda \in \{ (1-b)\alpha - \left( \eta+\sigma \rho B\right), (1-b)\beta - \left( \eta+\sigma \rho B\right), 0 \}$ (cf.\ (\ref{eq: optimal lambda Heston explicit})), we may use (\ref{eq: ODE zones}) to write
		\begin{align*}
			&B'(\tau) = -\kappa  B(\tau) + \frac{1}{2}\sigma^2  \left(B(\tau)\right)^2  + \frac{1}{2}\frac{b}{1-b}\inf_{\lambda\in \R}\left(2(1-b)\delta_K(\lambda) + \left(\eta + \lambda + \sigma \rho  B(\tau) \right)^2  \right) \\
			&= \min \left(-r_0^- + r_1^- B(\tau) + \frac{1}{2}r_2^-\big(B(\tau)\big)^2,  -r_0 + r_1 B(\tau) + \frac{1}{2}r_2\big(B(\tau)\big)^2, -r_0^+ +r_1^+ B(\tau) + \frac{1}{2}r_2^+\big(B(\tau)\big)^2 \right)\\
			&=: f(B(\tau)).
		\end{align*}
		The coefficients $r_2^-,$ $r_2$ and $r_2^+$ are non-negative, and therefore $f$ is the minimum of three convex functions. As real convex functions are locally Lipschitz continuous and Lipschitz continuity is preserved when taking the minimum over a finite number of functions, $f$ is locally Lipschitz continuous too. Hence, by the existence and uniqueness theorem of Picard-Lindel\"{o}f, there exists a unique solution $B$ to (\ref{eq:  ODE B constr. Heston}) for small $\tau >0.$ Moreover, as $f$ does not depend on $\tau,$ the ODE for $B$ is autonomous and its solution $B$ is either constant (if $f(0)=0$) or strictly monotone in $\tau$ (if $f(0)\neq 0$). Analogous arguments can be used to conclude the (strict) monotonicity of $B_u(\tau)$ from Corollary \ref{cor: unconstrained Heston} in Section \ref{subsec: Comparison to Unconstrained Portfolio}.
	\end{remark}

\begin{remark}\label{rem: optimal constrained pi is capped on ODE solution}
	Note that if $B$ is a solution to (\ref{eq:  ODE B constr. Heston}), then Lemma \ref{lem: optimal portfolio Heston MK23} and Lemma \ref{lem: ODE B for Box constraints} imply
	\begin{align*}
		\pi^{\ast}(t) = \frac{\eta + \la(B(T-t)) + \sigma \rho B(T-t)}{1-b} = 
		\begin{cases}
			\alpha, &\quad \rho B(T-t) < B_{-} \\
			\frac{\eta + \sigma \rho B(T-t)}{1-b}, &\quad B_{-} \leq \rho B(T-t) \leq B_+ \\
			\beta,  &\quad B_+ < \rho B(T-t). \\
		\end{cases}
	\end{align*}
	Therefore, the zones $Z_{-}=(-\infty, B_{-})$, $Z_{0}=[B_{-}, B_{+}]$ and $Z_{+}=(B_{+},\infty)$ determine whether the allocation constraint $K=[\alpha,\beta]$ is enforced for the candidate optimal portfolio process $\pi^{\ast}$ from Lemma \ref{lem: optimal portfolio Heston MK23}. Moreover, we may define $\hat{\pi}^{\ast}(t):= \frac{1}{1-b}\left(\eta + \sigma \rho B(T-t)\right)$ and express $\pi^{\ast}$ as a capped version of $\hat{\pi}^{\ast}$, i.e.
	$$
	\pi^{\ast}(t)= \text{Cap}(\hat{\pi}^{\ast}(t),\alpha, \beta):= \begin{cases}
		\alpha, &\quad \hat{\pi}^{\ast}(t) < \alpha \\
		\hat{\pi}^{\ast}(t), & \quad \alpha \leq \hat{\pi}^{\ast}(t) \leq \beta \\
		\beta, & \quad \beta < \hat{\pi}^{\ast}(t).
	\end{cases}
	$$
\end{remark}

In a true constrained context, i.e.\ $K\neq \R$, we may either determine an approximation of $B$ by using  a suitable numerical ODE solver (such as an Euler method) to solve the ODE (\ref{eq: ODE zones}) or directly derive an explicit expression for $B$ by individually solving each of the three Riccati ODEs in (\ref{eq: ODE zones}) and merging the solutions at the transition points between the zones $Z_{-}$, $Z_0$ and $Z_{+}$. To ensure that such solutions exist and do not explode before time $T$, we need to make the following assumption on $\mathcal{M}_H$ and the constraints $K = [\alpha, \beta].$
\newpage 
\begin{assumption}\label{ass: condition box constraints}
	\textcolor{white}{1}
	\begin{itemize}
		\item[(i)] \underline{Existence of Solution}:
		\begin{align*}
			\max \left\{\begin{array}{lr}
				\frac{b}{1-b}\eta \Big( \frac{\kappa \rho}{\sigma} + \frac{\eta}{2} \Big), \\
				b\alpha \left(\eta - \frac{1}{2}\alpha +\frac{\kappa \rho}{\sigma} + \frac{1}{2}\alpha b(1-\rho^2)\right), \\
				b\beta \left(\eta - \frac{1}{2}\beta +\frac{\kappa \rho}{\sigma} + \frac{1}{2}\beta b(1-\rho^2)\right)
			\end{array}\right\} < \frac{\kappa^2}{2\sigma^2}
		\end{align*}
		\item[(ii)] \underline{No Blow-Up}: \vspace*{0.25cm}\\
		The coefficients of each of the three Riccati ODEs satisfy $t_{+}(B_0) > T$ (cf.\ Lemma \ref{lem: Riccati ODE Filipovic}, (ii)) for each initial value
		$$
		B_0 \in \left \{\left(\frac{B_{-}}{\rho}\right)\1_{\{\rho \neq 0\}},\left(\frac{B_{+}}{\rho}\right)\1_{\{\rho \neq 0\}},0 \right \}.
		$$
	\end{itemize}
\end{assumption}

Provided that Assumption \ref{ass: condition box constraints} holds, the coefficients
\begin{align}\label{eq: definition r3 for three ODEs}
	r_3^{-} = \sqrt{(r_1^-)^2+2r_0^-r_2^-}, \quad  r_3 = \sqrt{(r_1)^2+2r_0r_2}, \quad r_3^{+} = \sqrt{(r_1^+)^2+2r_0^+r_2^+}
\end{align}
are well-defined and the solutions to each of the Riccati ODEs (\ref{eq: ODE zones}) do not blow up before time $T$ when started at any of the transition points between the zones $Z_{-}$, $Z_{0}$ and $Z_{+}$.\footnote{Technically, one can formulate this assumption less restrictively by expressing \enquote{No Blow-Up} in terms of the time spent in each of the zones $Z_{-}$, $Z_{0}$ and $Z_{+}$. However, as this would significantly complicate the presentation without adding major additional insights, it is omitted here.} For this reason, we define the following auxiliary functions: 

\begin{itemize}
	\item Let $\hat{B}^+,$ $\hat{B},$ and $\hat{B}^-$ be the solution to Riccati ODE (\ref{eq: Riccati Filipovic ODE}) with initial value $0$ as well as coefficients $r_0^+,$ $r_1^+,$ $r_2^+$, $r_0,$ $r_1,$ $r_2$ and $r_0^-,$ $r_1^-,$ $r_2^-,$ respectively.
	\item If $\rho \neq 0$, let $\hat{B}_{+}^+,$ $\hat{B}_{-}^-$ be the solution to Riccati ODE (\ref{eq: Riccati Filipovic ODE}) with initial value $\frac{B_{+}}{\rho}$, $\frac{B_{-}}{\rho}$ and coefficients $r_0^+,$ $r_1^+,$ $r_2^+$, $r_0^-,$ $r_1^-,$ $r_2^-,$ respectively.
	\item If $\rho \neq 0$, let $\hat{B}_{+},$ $\hat{B}_{-}$ be the solution to Riccati ODE (\ref{eq: Riccati Filipovic ODE}) with initial value $\frac{B_{+}}{\rho}$, $\frac{B_{-}}{\rho}$, respectively, and coefficients $r_0,$ $r_1,$ $r_2$.
\end{itemize}
 Moreover, if $\rho \neq 0$, we define the  transition times\footnote{If $\rho = 0$ all of these transition times will be infinite.}
\begin{align*}
	\tau_1^{+} &= \inf \left \{ \tau \ | \ \hat{B}^+(\tau) = \frac{B^+}{\rho}\right \}, \quad \tau_2^{+} = \inf \left \{ \tau \ | \ \hat{B}_{+}(\tau) = \frac{B^-}{\rho}\right \}, 
	\quad \tau_1 = \inf \left \{ \tau \ | \ \hat{B}(\tau) \in \left \{ \frac{B^-}{\rho}, \frac{B^+}{\rho} \right \} \right \}\\
	\tau_1^{-} &= \inf \left \{ \tau \ | \ \hat{B}^-(\tau) = \frac{B^-}{\rho}\right \} \quad \text{and} \quad \tau_2^{-} = \inf \left \{ \tau \ | \ \hat{B}_{-}(\tau) = \frac{B^+}{\rho}\right \}.
\end{align*}

Note that each of the above functions and transition times, if finite, admit a closed-form expression, which can be obtained via Lemma \ref{lem: Riccati ODE Filipovic} and Corollary \ref{cor: transition times Filipovic} in the supplementary material. Having introduced these auxiliary functions and transition times, we can finally express a closed-form solution for $B$ in terms of these processes.

\begin{theorem}\label{thm: solution constrained B in 1-D}
	Let Assumption \ref{ass: condition box constraints} hold. Then, 
	\begin{align*}
		B(\tau) = 
		\begin{cases}
			\hat{B}^-(\tau)\1_{\{\tau \leq \tau^-_1\}} + \hat{B}_{-}(\tau-\tau^-_1)\1_{\{\tau^-_1 < \tau \leq \tau^-_1 + \tau^-_2\}} + 	\hat{B}_+^+(\tau-(\tau^-_1 + \tau^-_2))\1_{\{\tau^-_1 + \tau^-_2 < \tau\}} , & \ \text{if} \  0 \in Z_{-} \\
			\hat{B}(\tau)\1_{\{\tau \leq \tau_1\}} + \hat{B}_-^-(\tau-\tau_1) \1_{\{ \tau > \tau_1, \rho 	\hat{B}(\tau_1) = B_{-} \}} + \hat{B}_+^+(\tau-\tau_1) \1_{\{ \tau > \tau_1, \rho 	\hat{B}(\tau_1) = B_{+} \}} , & \ \text{if} \  0 \in Z_{0} \\
			\hat{B}^+(\tau)\1_{\{\tau \leq \tau^+_1\}} + \hat{B}_{+}(\tau-\tau^+_1)\1_{\{\tau^+_1 < \tau \leq \tau^+_1 + \tau^+_2\}} + 	\hat{B}_-^-(\tau-(\tau^+_1 + \tau^+_2))\1_{\{\tau^+_1 + \tau^+_2 < \tau \}} , & \ \text{if} \  0 \in Z_{+} \\
		\end{cases}	
	\end{align*}
	satisfies ODE (\ref{eq:  ODE B constr. Heston}) for $0\leq \tau \leq T.$\footnote{Using a similar separation with respect to the zones $Z_{-},$ $Z_0,$ and $Z_{+}$ and equation (\ref{eq: integrated Riccati Filipovic}), it is also possible to determine a closed-form expression for $A$ from Lemma \ref{lem: dual ODEs Heston Model MK23}.}
\end{theorem}

\subsection{Verification Theorem}\label{subsec: Verification Theorem for Heston's Model}
Combining Remark \ref{rem: optimal constrained pi is capped on ODE solution} with Theorem \ref{thm: solution constrained B in 1-D} immediately yields a closed-form expression for the candidate optimal portfolio process $\pi^{\ast}$. It now just remains to prove a verification theorem which verifies that this candidate is indeed the optimal portfolio process corresponding to $\mathbf{(P)}$. This proof requires an additional assumption on the constraints $K = [\alpha, \beta],$ which ensures a certain boundedness of $\pi^{\ast}(t)$ for $t$ close to maturity $T$ as well as two auxiliary lemmas.
\begin{assumption}\label{ass: boundedness near maturity}
	\begin{align}\label{eq: boundedness assumption on constraints at t=T}
		\max \left \{ \frac{b\rho}{\kappa}\alpha, \  \frac{b\rho}{\kappa}\beta \right \} \leq \frac{\kappa}{\sigma^2}, 
	\end{align}
\end{assumption}

\begin{lemma}\label{lem: A1}
	Let Assumptions \ref{ass: condition box constraints} and \ref{ass: boundedness near maturity} hold and let $B$ be given as in Theorem \ref{thm: solution constrained B in 1-D}. Then, the following inequality holds for all $t\in [0,T]$
	\begin{align*}
		\frac{b\rho}{\sigma}\pi^{\ast}(t) + B(T-t) \leq \frac{\kappa}{\sigma^2}.
	\end{align*}
\end{lemma}

\begin{lemma}\label{lem: A2}
	Let Assumptions \ref{ass: condition box constraints} and \ref{ass: boundedness near maturity} hold and let $B$ be given as in Theorem \ref{thm: solution constrained B in 1-D}. Then the following inequality holds for all $t\in [0,T]$
	\begin{align*}
		\frac{1}{2}\frac{b}{1-b}\eta^2-\frac{1}{2}\frac{b}{1-b}\left(\la\left(B(T-t)\right) + \sigma \rho B(T-t) \right)^2 - \frac{1}{2}b^2\rho^2\left(\pi^{\ast}(s)\right)^2 & \ \nonumber \\
		+ b \frac{\rho\kappa}{\sigma} \pi^{\ast}(t) + \frac{b}{1-b}\frac{\rho}{\sigma}\left[\left(\la\right)'\left(B(T-t)\right)+\sigma \rho\right] B'(T-t) \quad & < \frac{1}{2}\frac{\kappa^2}{\sigma^2}.
	\end{align*}
\end{lemma}

\begin{theorem}[Verification Theorem in $\mathcal{M}_H$]\label{thm: strong verification in Heston market}
	\textcolor{white}{1}\\
	Consider the financial market $\mathcal{M}_H$, let Assumptions \ref{ass: condition box constraints} and \ref{ass: boundedness near maturity} hold and let $B$ be given as in Theorem \ref{thm: solution constrained B in 1-D}. Then,
	\begin{align}
		\pi^{\ast}(t) =
		\begin{cases}
			\alpha, &\quad \rho B(T-t) < B_{-} \\
			\frac{\eta + \sigma \rho B(T-t)}{1-b}, &\quad B_{-} \leq \rho B(T-t) \leq B_+ \\
			\beta,  &\quad B_+ < \rho B(T-t). \\
		\end{cases}
	\end{align}
	is optimal for $\mathbf{(P)}$.
\end{theorem}

\subsection{Comparison to Unconstrained Portfolio}\label{subsec: Comparison to Unconstrained Portfolio}
Unsurprisingly, we can immediately recover the solution to the unconstrained optimisation problem, as discussed in \citet{Kraft2005}, from Lemma \ref{lem: optimal portfolio Heston MK23} and Lemma \ref{lem: ODE B for Box constraints}.

\begin{corollary}\label{cor: unconstrained Heston}[Closed-form Unconstrained Optimal Portfolio as in \citet{Kraft2005}]
	\textcolor{white}{1}\\
	Let $K= \R$ (i.e. $\alpha = -\infty,$ $\beta = \infty)$ and $B_u:[0,T]\rightarrow \R$ with $B_u(0)=0$ satisfy
	\begin{align}\label{eq: Unconstr ODE B}
		B_u'(\tau) = -r_0 + r_1 B_u(\tau) + \frac{1}{2}r_2B_u(\tau)^2 \qquad \forall \tau \in [0,T].
	\end{align}
	Then, $\la(B) = 0$ $\forall B\in \R$ and the candidate optimal portfolio $\pi^{\ast}$ is given by
	$$
	\pi_{u}(t):= \pi^{\ast}(t) = \frac{1}{1-b}\left(\eta +  \sigma \rho  B_u(T-t)\right).
	$$
\end{corollary}

\begin{remark}\label{rem: closed-form unconstrained portfolio}
	If the market parameters satisfy (cf.\ Assumption \ref{ass: condition box constraints})
	\begin{align}\label{eq: inequality ensuring solution to Riccati ODE}
		\frac{b}{1-b}\eta \Big( \frac{\kappa \rho}{\sigma} + \frac{\eta}{2} \Big) <  \frac{\kappa^2}{2\sigma^2},
	\end{align}
	then
	\begin{align}\label{eq: unconstr. optimal B}
		B_u(\tau) = \frac{2r_0(e^{ r_3\tau}-1)}{(r_1- r_3)(e^{ r_3\tau}-1)-2 r_3}
	\end{align}
	and the optimality of $\pi_u$ for the unconstrained portfolio optimisation problem can be verified formally (see e.g. Theorem 5.3 in \citet{Kraft2005}).\footnote{Equation (\ref{eq: inequality ensuring solution to Riccati ODE}) corresponds to part (i) of Assumption \ref{ass: condition box constraints}. In the setting of \cite{Kraft2005}, part (ii) of Assumption \ref{ass: condition box constraints} is also implied by (\ref{eq: inequality ensuring solution to Riccati ODE}) and so does not have to be mentioned explicitly.}
\end{remark}

On an abstract level, when adding (allocation) constraints $K=[\alpha, \beta]$ to a portfolio optimisation problem, the optimal constrained portfolio $\pi^{\ast}$ for $(\mathbf{P})$ will be given by a projection $\mathcal{P}_K: \Lambda \rightarrow \Lambda_K$ which maps the optimal unconstrained portfolio $\pi_u$ onto $\Lambda_K,$ i.e.
$$
\pi^{\ast} = \pi_u + \left(\pi^{\ast}-\pi_u\right) =: \mathcal{P}_K\left(\pi_u\right).
$$
In a Black-Scholes financial market $\mathcal{M}_{BS}$ with constant market coefficients (i.e. $\mathcal{M}_H$ with $\sigma=\kappa= \rho = 0.$), the optimal unconstrained portfolio is a constant-mix strategy $\pi_u(t):= \pi_M = \frac{1}{1-b}\eta,$ the so-called \enquote{Merton portfolio}. Setting $\sigma = \rho = 0$ and $B\equiv 0$ in Remark \ref{rem: optimal constrained pi is capped on ODE solution}, one can easily see that the projection $\mathcal{P}=\mathcal{P}^{BS}$ in the Black-Scholes market simply caps off $\pi_M$ at the boundaries if $\pi_M \notin K = [\alpha, \beta],$ i.e.%\footnote{This can be shown more formally by explicitly calculating the minimiser $\la$ and the corresponding optimal portfolio process $\pi^{\ast}$ in Lemma 4.2 in \citet{EKZ2023}.}$
$$
\mathcal{P}^{BS}_K(\pi_M) =  \text{Cap}\left(\pi_M,\alpha, \beta\right)= \begin{cases}
	\alpha, &\quad \pi_M < \alpha \\
	\pi_M, &\quad \alpha \leq \pi_M \leq \beta \\
	\beta, &\quad \beta < \pi_M.
\end{cases}
$$
Given a solution $B$ to (\ref{eq:  ODE B constr. Heston}) and considering Remark \ref{rem: optimal constrained pi is capped on ODE solution}, it initially appears that the optimal constrained portfolio $\pi^{\ast}$ in $\mathcal{M}_H$ can be obtained from the same projection. However, if $K \neq \R,$ then $B_u$ as in Corollary \ref{cor: unconstrained Heston} and $B$ as in Theorem \ref{thm: solution constrained B in 1-D} are solutions to possibly different ODEs. In particular, this implies that the portfolios $\pi_u$ and $\hat{\pi}^{\ast}$ may not be identical, in which case the projection $\mathcal{P}^H_K$ for the Heston market does not necessarily coincide with the projection $\mathcal{P}^{BS}_K$ for the Black-Scholes market either. In other words, in a financial market with Heston stochastic volatility we in general have
\begin{align*}
	\pi^{\ast} = \mathcal{P}^H_K\left(\pi_u\right) = \text{Cap}\big(\pi_u + \underbrace{\left(\hat{\pi}^{\ast}-\pi_u\right)}_{\neq 0},\alpha,\beta\big) \neq \text{Cap}\big(\pi_u,\alpha,\beta\big) = \mathcal{P}^{BS}_K(\pi_u).
\end{align*}
In the following, we render this observation more precise by providing both conditions under which $\mathcal{P}^H_K = \mathcal{P}^{BS}_K$ and conditions under which $\mathcal{P}^H_K \neq \mathcal{P}^{BS}_K.$ The former case is true, whenever either $\rho = 0$ or $\pi_M \in K.$
\begin{lemma}\label{lem: constrained portfolio equal capped portfolio}
	Let $\pi^{\ast}$ be as in Lemma \ref{lem: optimal portfolio Heston MK23}, $\hat{\pi}^{\ast}$ be as in Remark \ref{rem: optimal constrained pi is capped on ODE solution} and $\pi_u$ as in Corollary \ref{cor: unconstrained Heston}. If either
	$$
	\rho = 0 \qquad \text{or} \qquad \pi_M \in K,
	$$
	then
	$$
	\pi^{\ast}=\mathcal{P}^H_K\left(\pi_u\right) = \text{Cap}\left(\pi_u,\alpha, \beta\right) =  \mathcal{P}^{BS}_K\left(\pi_u\right).
	$$
\end{lemma}

	If $\rho=0$, the stochasticity of the volatility is completely unhedgeable in $\mathcal{M}_H$. As a consequence, the optimal unconstrained portfolio processes coincide in $\mathcal{M}_{BS}$ and in $\mathcal{M}_H.$ Thus, the projections $\mathcal{P}^{BS}_K$ and $\mathcal{P}^H_K$ are identical if $\rho = 0$ too. In contrast, if $\rho \neq 0,$ then the projections can only be different if the underlying ODE solutions $B_u$ and $B$ are different, specifically when $\pi_u$ and $\hat{\pi}^{\ast}$ begin taking values inside $K.$ This is the case if and only if $\pi_u$ and $\hat{\pi}^{\ast}$ begin taking values inside $K$ at different time points. This observation  leads to an equivalent characterisation of when the projections $\mathcal{P}^{BS}_K$ and $\mathcal{P}^H_K$ are different.
	
	\begin{lemma}\label{lem: equivalence gap}
		Let $\pi^{\ast}$ be as in Lemma \ref{lem: optimal portfolio Heston MK23}, $\hat{\pi}^{\ast}$ be as in Remark \ref{rem: optimal constrained pi is capped on ODE solution} and $\pi_u$ as in Corollary \ref{cor: unconstrained Heston}. The following statements are equivalent:
		\begin{itemize}
			\item[(i)] 
			$$
			\pi^{\ast} = \mathcal{P}^H_K\left(\pi_u,\alpha, \beta\right) \neq \mathcal{P}^{BS}_K\left(\pi_u,\alpha, \beta\right) = \text{Cap}(\pi_u,\alpha,\beta)
			$$
			\item[(ii)]
			$$
			 \pi_M \notin [\alpha, \beta] \quad \text{and} \quad \exists t \in (0,T): \left | \Big \{\hat{\pi}^{\ast}(t), \pi_u(t) \Big \} \cap  (\alpha,\beta)\right | = 1
			$$
		\end{itemize}
	\end{lemma}
	
	We can construct an extreme case which satisfies the requirements of Lemma \ref{lem: equivalence gap} by choosing $\alpha$ such that $B(\tau)$ is constant and choosing the market parameters such that $\pi_u$ changes sufficiently during the investment horizon to ensure that $\pi_u(t^{\ast})\in (\alpha, \beta)$ for some $t^{\ast}\in [0,T].$
	
	\begin{corollary}\label{cor: stationary case}
		Let $\pi^{\ast}$ be as in Lemma \ref{lem: optimal portfolio Heston MK23}, $\hat{\pi}^{\ast}$ be as in Remark \ref{rem: optimal constrained pi is capped on ODE solution} and $\pi_u$ as in Corollary \ref{cor: unconstrained Heston}. Let $\text{sign}(x)\in \{-1,0,1\}$ denote the sign of $x\in \R.$
		\begin{itemize}
			\item[(i)] If 
			$$
			0< \pi_M = \frac{\alpha}{2}< \alpha \quad \text{and} \quad \alpha < \pi_u(t^{\ast})< \beta \quad \text{for some} \ t^{\ast}\in [0,T],
			$$
			then $B(\tau) = \alpha$ for all $\tau \in [0,T],$ $\pi^{\ast}(t) = \alpha$ for all $t\in[0,T]$ and $\pi^{\ast} = \mathcal{P}^H_K\left(\pi_u,\alpha, \beta\right) \neq \mathcal{P}^{BS}_K\left(\pi_u,\alpha, \beta\right) = \text{Cap}\left(\pi_u,\alpha, \beta\right).$
			\item[(ii)] If $\pi_M > \beta >0,$ then
			$$
			\text{sign}\left(\frac{\partial}{\partial t} \hat{\pi}^{\ast}(t)\right) = \text{sign}\left(\frac{\partial}{\partial t} \pi_u(t)\right) = - \text{sign}(\rho b) \quad \forall t\in [0,T].
			$$
			Hence, if in addition $b<0$ and $\rho <0,$ then $\mathcal{P}^H_K = \mathcal{P}^{BS}_K.$
		\end{itemize}
	\end{corollary}

Clearly, the requirements on $\alpha$ in Corollary \ref{cor: stationary case}, (i) are quite restrictive, but they still provide a valuable insight into when we can expect to see a large difference between the projections $\mathcal{P}^{BS}_K$ and $\mathcal{P}^H_K.$ Namely, if 
\begin{itemize}
	\item the optimal unconstrained portfolio $\pi_u$ violates the constraint at maturity (i.e. $\pi_u(T)=\pi_M \notin K$) and there is sufficient change in $\pi_u(t)$ during the investment period such that $\pi_u(t^{\ast})\in K$ for some $t^{\ast} \in K.$
	\item the derivatives of $B(\tau)$ and $B_u(\tau)$ are considerably different while $\pi_u \notin K$ (constant $B$ being the extreme case). 
\end{itemize}
As a matter of fact, we will later see in the numerical experiments in Section \ref{sec: Numerical Studies} that it is sufficient if $\alpha \approx 2\pi_M$ (i.e. $B(\tau)$ is nearly constant) to cause a considerable difference between the two projections. \\

As evidenced by the majority of empirical calibrations of Heston's stochastic volatility model to financial time series (see e.g. \citet{Escobar-Anel2016} for an overview), the parameter $\rho$ is negative for most realistic applications. In the empirical study on risk preferences of mutual fund managers by \citet{Koijen2014}, it is reported that the risk aversion parameter $b$ has a median of $b = -1.43$ and a mean of $b = -4.8.$ For more risk averse investors, such as insurance companies, reinsurance companies or pension funds, one can thus realistically assume negative values for $b.$ Thus, for most realistic parameter configurations of $\mathcal{M}_H$ with $\pi_M> \beta,$ the projections $\mathcal{P}^{H}_K$ and $\mathcal{P}^{BS}_K$ coincide for investors with a high degree of risk aversion (i.e. for a low value of $b$).

\section{Implications for Related Models}\label{sec: Implications for Related Models}
	We now consider a generalised version of the financial market $\mathcal{M}_H$ with $d\in N$ risky assets, $d$ independent CIR processes as risk drivers and a generalised dependence of market price of risk and risky asset volatility on these risk drivers. From now on, let $W^z$ and $\hat{W}$ denote independent $d$-dimensional Wiener processes and consider parameters $\kappa,$ $\theta,$ $\sigma,$ $z_0 \in (0,\infty)^d$ such that their components satisfy $2\kappa_i \theta_i > \sigma_i^2$ for $i=1,...,d.$
	Then, we define the $d-$dimensional CIR process $z=\left(z_1,...,z_d\right)'$ through the dynamics
	\begin{align*}
		dz_i(t) = \kappa \left(\theta_i-z_i(t)\right)dt + \sigma_i\sqrt{z_i(t)}dW^z_i(t).
	\end{align*}
	
	In the following, let $\1\in \R^d$ be the $d$-dimensional vector of ones, let $x\odot y$ denote the element-wise multiplication of $x,y\in \R^d$ and let $\sqrt{x}$ denote the element-wise square root of $x\in \R^d.$ Consider a given correlation vector $\rho \in (-1,1)^d,$ market price of risk $\gamma:(0,\infty)^d\rightarrow \R^d$ and volatility $\Sigma:(0,\infty)^d\rightarrow \R^{d\times d}$ such that $\Sigma(z)$ is non-singular for all $z\in (0,\infty)^d.$ Then, we define the financial market $\mathcal{M}_{H}^{\gamma, \Sigma},$ consisting of one risk-free asset $P_0$ and $d$ risky assets $P=\left(P_1,\hdots, P_d\right)'$ with dynamics  
	\begin{align*}
		dP_0(t) &= P_0(t) r dt\\
		dP(t) &= P(t)\odot \big[\underbrace{\left(r\1 + \Sigma(z(t))\gamma(z(t))\right)}_{=:\mu(z(t))}dt + \Sigma(z(t))\underbrace{\left(\rho\cdot dW^z(t)+\sqrt{\1-\rho \odot \rho }d\hat{W}(t)\right)}_{=:dW(t)} \big].
	\end{align*}
	Clearly, we can recover the financial market $\mathcal{M}_H,$ as considered in Section \ref{sec: Heston's Stochastic Volatility Model}, by assuming $d=1$ and choosing $\gamma(z) = \eta \sqrt{z}$ and $\Sigma(z) = \sqrt{z}$ and the Black Scholes model $\mathcal{M}_{BS}$ if both $\gamma$ and $\Sigma$ are constants. Similar, but slightly more general financial market models than $\mathcal{M}_{H}^{\gamma, \Sigma}$ have been considered in \citet{Liu2006} or \citet{EKZ2023}, for example. In $\mathcal{M}_{H}^{\gamma, \Sigma},$ the wealth process $V^{v_0, \pi}$ of an investor with initial wealth $v_0$ who trades continuously in time with $\R^d$-valued relative portfolio process $\pi$, satisfies the SDE
	\begin{align*}
		dV^{v_0,\pi}(t) &= V^{v_0,\pi}(t)\left[\left(r+\gamma(z(t))'\Sigma(z(t))'\pi(t)\right)dt + \pi(t)'\Sigma(z(t))dW(t)\right] \\
		&= V^{v_0,\pi}(t)\left[\left(r+\left(\mu(z(t))-r\1\right)'\pi(t)\right)dt + \pi(t)'\Sigma(z(t))dW(t)\right].
	\end{align*}
	In $\mathcal{M}_{H}^{\gamma, \Sigma},$ the set of admissible portfolio processes naturally generalises to
	\begin{align*}
		\Lambda^{\gamma, \Sigma} = \left \{ \pi = \big(\pi(t)\big)_{t \in [0,T]} \ \text{progr. measurable} \  \Big | \  \int_0^T \Vert \Sigma(z(t))'\pi(t)\Vert^2 dt < \infty \ Q-a.s.\right \} 
	\end{align*}
	For a closed convex set with non-empty interior $K\subset \Rex^d,$ the portfolio optimisation problem $\mathbf{(P)}$ in $\mathcal{M}_{H}^{\gamma, \Sigma}$ is then defined as
	\begin{equation*}
		\mathbf{(P)}
		\begin{cases}
			\Phi(v_0) &= \underset{\pi \in \Lambda_K}{\sup} \mathbbm{E}\big[U(V^{v_0, \pi}(T)) \big] \\
			\Lambda_K &= \big \{\pi \in \Lambda^{\gamma, \Sigma} \ \big | \ \pi(t) \in K \ \mathcal{L}[0,T]\otimes Q-\text{a.e.} \big \}.
		\end{cases}
	\end{equation*}
	In the following two sections, we investigate the solvability of $\mathbf{(P)}$ for given choices of $\gamma,$ $\Sigma,$ and $K.$ In Section \ref{subsec: PCSV Model}, we consider the PCSV Model, as discussed in \citet{Rubtsov2017}, and in Section \ref{subsec: Inverse Volatility Constraints}, we consider inverse volatility constraints $K,$ which impose stronger restrictions on an investor's portfolio during periods of high volatility.
	
	\subsection{PCSV Model}\label{subsec: PCSV Model}
	
	We recover the PCSV (\enquote{Principal Component Stochastic Volatility}) model $\mathcal{M}^{PCSV},$ as proposed in \citet{Escobar2010}, from $\mathcal{M}_{H}^{\gamma, \Sigma}$ by considering an orthogonal matrix $A=\big(a_1,...,a_d\big)\in \R^{d\times d}$ and defining market price of risk and volatility as
	\begin{align}\label{eq: Definition PCSV Market}
		\gamma(z) = \diag(\sqrt{z})A'\eta = \Sigma(z)'\eta, \quad \Sigma(z)= A\diag(\sqrt{z}) \quad \forall z\in (0,\infty)^d,
	\end{align}
	where $\diag(x)\in \R^{d\times d}$ denotes the diagonal matrix with entries $x\in \R^d,$ and $\eta \in \R^d$ is a constant. \\ 
	If $A=I$, then $\mathcal{M}^{PCSV}$ can be regarded as a canonical generalisation of Heston's stochastic volatility model $\mathcal{M}_H$ for a $d$-dimensional asset universe, where each risky asset's volatility is determined as the square root of one of the $d$ independent CIR processes $z_i.$ However, in its general form, the independent components of the $d$-dimensional CIR process $z$ are not directly regarded as volatilities. Instead, the instantaneous covariance matrix of risky asset returns
	\begin{align*}
		\Sigma(z(t))\Sigma(z(t))' = A\diag(z(t))A'
	\end{align*}
	is decomposed into its principal components, i.e. the columns $a_i$ of the matrix $A$ represent its eigenvectors, and the independent CIR processes $z_i$ represent their (stochastic) eigenvalues. This approach not only enables the modelling of stochastic covariances of asset returns because of additional degrees of freedom in $A,$ but also allows for an interpretation of $z$ as hidden risk factors determining the volatility level in the financial market. Moreover, \citet{Rubtsov2017} demonstrated that several stylised facts are captured by the PCSV, such as stochasticity of volatilities and correlation of risky asset returns, volatility and correlation leverage effect, volatility spillovers, and increasing correlation in periods of high market volatility. \\
	Let $\Lambda^{PCSV}$ be the set of admissible portfolios in $\mathcal{M}^{PCSV}.$ Then, for any $\pi\in\Lambda^{PCSV},$ the wealth process $V^{v_0, \pi}$ satisfies the SDE
	\begin{align*}
		dV^{v_0,\pi}(t) = V^{v_0,\pi}(t)\left[\left(r+\eta'A\diag(z(t))A'\pi(t)\right)dt + \pi(t)'A\diag(\sqrt{z(t)})dW(t)\right],
	\end{align*}
	for $t\in[0,T].$ The instantaneous variance of $V^{v_0,\pi}(t)$ can therefore be decomposed into a weighted sum of the risk factors $z,$ since
	\begin{align*}
		\Vert \diag(\sqrt{z(t)})A'\pi(t) \Vert = \left \Vert \begin{pmatrix}
			a_1'\pi(t) \sqrt{z_1(t)} \\
			\vdots \\
			a_d'\pi(t)\sqrt{z_d(t)}
		\end{pmatrix}\right \Vert^2 = \sum_{i=1}^d \left(a_i'\pi(t)\right)^2z_i(t).
	\end{align*}
	In this sense, the portfolio weights determine a risk exposure $\left(a_i'\pi(t)\right)^2$ to the risk factor $z_i.$ Hence, it is very natural to impose risk limits on these exposures, i.e. for given upper bounds $\beta_1,...,\beta_d>0$ we require that
	\begin{align*}
		\left(a_i'\pi(t)\right)^2 \leq \beta_i \quad \forall i = 1,...,d \quad \Leftrightarrow \quad A'\pi(t)\in \bigtimes_{i=1}^d [0, \sqrt{\beta_i}] \quad \Leftrightarrow \quad \pi(t)\in \underbrace{A\cdot \left( \bigtimes_{i=1}^d [0, \sqrt{\beta_i}]\right)}_{=:K_{PCSV}}. 
	\end{align*}
	We can reuse the ideas and results from Section \ref{sec: Heston's Stochastic Volatility Model} to obtain the optimal portfolio to the portfolio optimisation problem $\mathbf{(P)}$ in $\mathcal{M}_{H}^{\gamma, \Sigma} = \mathcal{M}_{PCSV}$ with constraints $K= K_{PCSV}.$ 
	
	\begin{theorem}\label{thm: strong verification PCSV}
		Consider the portfolio optimisation problem $\mathbf{(P)}$ in $\mathcal{M}_{PCSV}$ with constraints $K= K_{PCSV}$, let the parameters $b,$ $\left(\eta_A\right)_i:=\left(A'\eta\right)_i$, $\kappa_i,$ $\theta_i,$ $\sigma_i,$ $\alpha_i = 0,$ $\sqrt{\beta_i}$ satisfy Assumptions \ref{ass: condition box constraints} and \ref{ass: boundedness near maturity} and $B_i$ be defined as in Theorem \ref{thm: solution constrained B in 1-D}. Define the portfolio $\pi^{\ast}_A(t)$ via
		\begin{align*}
			\left(\pi^{\ast}_A(t)\right)_i= \text{Cap}\left(\frac{1}{1-b}\left(\eta_i + \sigma\rho B_i(T-t)\right),0,\sqrt{\beta_i}\right).
		\end{align*}
		Then the portfolio $\pi^{\ast}(t) = A\cdot \pi^{\ast}_A(t)$ is optimal for $\mathbf{(P)}.$
	\end{theorem}
	
	The key argument in the proof of Theorem \ref{thm: strong verification PCSV} lies in a change of control, which transforms the portfolio optimisation problem $\mathbf{(P)}$ into an equivalent portfolio optimisation problem $\mathbf{(P_A)}$ in a financial market, which consists of $d$ risky assets with independent Heston volatilities and interval constraints. Thanks to this familiar structure and the independence of the risky asset volatilities, we can extend the ideas from Section \ref{sec: Heston's Stochastic Volatility Model} to solve $\mathbf{(P_A)}$ and invert the change of control to obtain a solution to $\mathbf{(P)}.$
	
	\subsection{Inverse Volatility Constraints}\label{subsec: Inverse Volatility Constraints}
	
	We now discuss a related problem, which we approach with a similar methodology to the one in Section \ref{subsec: PCSV Model}. Consider again the one-dimensional setting with one risky asset, i.e. $d=1.$ In this section, we no longer assume that the convex constraints $K\subset \Rex$ are static, but allow them to depend on the stochastic factor $z.$ More specifically, we consider volatility-dependent constraints of the form
	\begin{align*}
		\pi(t) \in K(z(t)) \quad \mathcal{L}[0,T]\otimes Q-a.e.,
	\end{align*}
	where $K:(0,\infty)\rightarrow \mathcal{B}(\R)$ is a set-valued function, taking only closed-convex values in the Borel set $\mathcal{B}(\R).$ The motivation for such constraints is quite clear: Depending on the current state of the financial market $\mathcal{M}_{H}^{\gamma, \Sigma},$ in particular the level of risky asset volatility $\Sigma(z(t))$ and the market price of risk $\gamma(z(t)),$ investors may face different constraints on their portfolio, such as more relaxed bounds in periods of low volatility or stricter bounds in periods of high volatility. Further, in the spirit of mean-variance optimisation, we can think of an investor seeking an optimal portfolio allocation subject to constraints on his instantaneous portfolio volatility
	\begin{align*}
		0 \leq \pi(t)\Sigma(z(t)) \leq \beta_{z} \quad \mathcal{L}[0,T]\otimes Q-a.e. \quad \Leftrightarrow \quad 0 \leq \pi(t) \leq \frac{\beta_{z}}{\Sigma(z(t))} \quad \mathcal{L}[0,T]\otimes Q-a.e.,
	\end{align*}
	for a given volatility level $\beta_{z}>0.$\footnote{Note that this is different from classic mean-variance optimisation, where the variance of the terminal portfolio wealth $V^{v_0,\pi}(T)$ is constrained.}
	
	Keeping this motivation in mind, we thus define the portfolio optimisation problem with volatility-dependent constraints $\mathbf{(P^z)}$ as
	
	\begin{align*}
		\mathbf{(P^z)} = \begin{cases}
			\Phi^{z}(v_0) &= \underset{\pi \in \Lambda_{K(\cdot)}}{\sup} \mathbbm{E}\big[U(V^{v_0, \pi}(T)) \big] \\
			\Lambda_{K(\cdot)} &= \left \{ \pi \in \Lambda^{\gamma, \Sigma} \ \big | \ \pi(t) \in K(z(t)) \ \mathcal{L}[0,T]\otimes Q-\text{a.e.} \right \}.
		\end{cases}
	\end{align*}
	
	In its most general form, the portfolio optimisation $\mathbf{(P^z)}$ is highly non-trivial, since closed-form solutions for its optimal portfolio process $\pi^{\ast}_z$ can rarely be determined for general $\gamma$ and $\Sigma,$ even in the absence of (stochastic) allocation constraints. In particular, the portfolio optimisation $\mathbf{(P)}$ is included as a special case in the definition of $\mathbf{(P^z).}$ However, due to the results of \citet{cvitanic1992} for the Black-Scholes model $\mathcal{M}^{BS}$ with constant volatility, as well as the results from Section \ref{sec: Heston's Stochastic Volatility Model} for Heston's stochastic volatility model $\mathcal{M}^H,$ we know of at least two different models in which solutions to $\mathbf{(P)}$ (respectively $\mathbf{(P^z)}$) with static constraints can be obtained in closed form. Using another change of control argument, we can therefore derive conditions on the market parameters $\gamma,$ $\Sigma$ and the constraints $K,$ under which we can transform $\mathbf{(P^z)}$ into an equivalent, solvable portfolio optimisation problem $\mathbf{(P)}$ in either $\mathcal{M}^{BS}$ or $\mathcal{M}^{H}.$

\begin{theorem}\label{thm: optimal portfolio vola-inverse constraints}
	Consider the financial market $\mathcal{M}_{H}^{\gamma, \Sigma}$ and the portfolio optimisation problem $\mathbf{(P^{z})}.$ \linebreak Consider constants $-\infty \leq \alpha < \beta \leq \infty.$
	\begin{itemize}
		\item[(i)] If $\gamma(z)=\eta$ for some $\eta>0$ and $K(z) = \frac{1}{\Sigma(z)}[\alpha_z,\beta_z],$ then the portfolio process
		\begin{align*}
			\pi^{\ast}_{z}(t) = \frac{1}{\Sigma \left(z(t)\right)}\text{Cap}(\pi_M,\alpha,\beta)
		\end{align*}
		is optimal for $\mathbf{(P^{z})}.$
		\item[(ii)] If $\gamma(z) = \eta\sqrt{z(t)}$ for some constant $\eta > 0,$ $K(z) = \frac{\sqrt{z}}{\Sigma(z)}[\alpha,\beta],$ and Assumptions \ref{ass: condition box constraints} and \ref{ass: boundedness near maturity} are satisfied, then for $\pi^{\ast}$ defined as in Theorem \ref{thm: strong verification in Heston market}, the portfolio process
		\begin{align*}
			\pi^{\ast}_{z}(t) = \frac{\sqrt{z(t)}}{\Sigma\left(z(t)\right)}\pi^{\ast}(t)
		\end{align*}
		is optimal for $\mathbf{(P^{z})}.$
	\end{itemize}
\end{theorem}

	The statements of Theorem \ref{thm: optimal portfolio vola-inverse constraints} can easily be easily generalised to financial markets with $d>1$ risky assets by an analogous change of control argument. Using the results for constant volatility markets from Example 15.2 in \citet{cvitanic1992}, one can prove a multi-dimensional analogue to statement (i) and using the results for the PCSV model from Section \ref{subsec: PCSV Model}, one can prove a multi-dimensional analogue to statement (ii). For ease of presentation, we refrain from a detailed discussion of this generalisation.

\section{Numerical Studies}\label{sec: Numerical Studies}
	In this section, we illustrate the properties of the optimal portfolio $\pi^{\ast}$ for $\mathbf{(P)}$ in Heston's stochastic volatility model $\mathcal{M}_H,$ using a numerical example. In particular, we analyse the difference between $\pi^{\ast}$ and two suboptimal naive portfolio processes $\pi,$ which either directly follow the optimal portfolio process in $\mathcal{M}_{BS}$ (i.e. $\pi = \text{Cap}(\pi_M, \alpha,\beta)$) or apply the projection $\mathcal{P}^{BS}_K$ from $\mathcal{M}_{BS}$ to the optimal unconstrained portfolio $\pi_u$ in $\mathcal{M}_H$ (i.e. $\pi = \text{Cap}(\pi_u, \alpha,\beta)$). The suboptimality of these portfolios will be quantified using the concept of wealth-equivalent loss. \\
	
	Such an analysis is only meaningful if the differences between a financial market with stochastic (Heston) volatility $\mathcal{M}_H$ and a financial market with constant volatility $\mathcal{M}_{BS}$ are already reflected in the optimal unconstrained portfolios $\pi_u$ for $\mathcal{M}_H$ and $\pi_M$ for $\mathcal{M}_{BS}.$ Since allocation constraints further restrict the set of admissible portfolio allocations, any existing differences between $\pi_u$ and $\pi_M$ tend to be diminished further when adding allocation constraints. From an investor's perspective, the distinction between $\mathcal{M}_{BS}$ and $\mathcal{M}_H$ is only relevant if the volatility of risky asset log returns $\sqrt{z(t)}$ changes significantly and these changes are partially hedgeable through trading in the risky asset. This is the case if the volatility of the volatility ($\sigma$) is large, the mean reversion speed ($\kappa$) is small, and the correlation between risky asset and volatility diffusion ($\rho$) is close to either $1$ or $-1$ (i.e. $|\rho|$ is large). \\
	Based on these requirements, we choose the market parameters (see Table \ref{tab: base market parameters}) for our numerical example such that the resulting market dynamics resemble a financial crisis.
	\begin{table}[h]
		\centering
		\begin{tabular}[H]{l c|| c |  l}
			\textbf{Parameter} & & \textbf{Value} & \textbf{Explanation} \\ 
			\hline
			End of Investment-Horizon &$T$ & $1$ & Limited duration of financial crises \\
			Risk Aversion Parameter & $b$ & $-2.5$ & Within ranges estimated in Table 1, \citet{Koijen2014}\\
			Initial Wealth & $v_0$ & $1$ & For convenience \\
			Risk-Free Interest Rate & $r$ &$0$ & For convenience \\
			Market Price of Risk Driver & $\eta$ & $3.0071$ & Table 2, \citet{Escobar2021}\\
			Mean Reversion Speed &$\kappa$  & $3.15$ & Table 3 \enquote{\%MSE}, \citet{Moyaert2011}\\
			Volatility of Volatility & $\sigma$ & $0.76$ & Table 3 \enquote{\%MSE}, \citet{Moyaert2011}\\
			Correlation & $\rho$ & $-0.81$ & Table 3 \enquote{\%MSE}, \citet{Moyaert2011}\\
			Long-Term Mean & $\theta$ &$0.35$ & Feller's Condition\\
			Initial Variance & $z_0$ & $0.35$ & Chosen equal to $\theta$ \\
		\end{tabular}
		\caption{Base parameters for the financial market $\mathcal{M}_H.$}\label{tab: base market parameters}
	\end{table}
	The only volatility parameters which influence the optimal portfolio allocation are $\sigma,$ $\kappa,$ and $\rho.$ Our choices for these parameters in Table \ref{tab: base market parameters}, follow the calibration results of \citet{Moyaert2011}, who calibrated Heston's stochastic volatility model using option prices on the Eurostoxx 50 during the 2008 financial crisis. The relatively short investment horizon of $T=1$ year is chosen to reflect the limited duration of most financial crises.\footnote{\citet{Forbes2022} reported that the average length of an S\&P500 bear market (defined as a period with drawdown in excess of $20\%$) was $289$ days.} \\
	We quantify the sub-optimality of both naive portfolio processes in comparison to the optimal constrained portfolio process using the concept of wealth-equivalent loss (\enquote{WEL}). For an arbitrary portfolio process $\pi \in \Lambda_K,$ we define the expected utility functional $\Jp:[0,T]\times(0,\infty)\times(0,\infty) \rightarrow \R$ as
	\begin{align}\label{eq: expected utility functional}
		\Jp(t,v,z) = \E \left[ U\left(V^{v_0,\pi}(T)\right) \ | \ V^{v_0,\pi}(t)=v, \ z(t)=z \right].
	\end{align}
	When considering the optimal portfolio process $\pi^{\ast},$ the expected utility functional coincides with the value function of $\mathbf{(P)},$ i.e. $J^{\pi^{\ast}}(t,v,z) = \Phi(t,v,z)$ for all $(t,v,z) \in [0,T]\times(0,\infty)\times (0,\infty).$ The WEL $L^{\pi}=L^{\pi}(t,z)$ of $\pi$ is then defined as the solution to the equation\footnote{Since we exclusively work with power utility functions in this paper, we may without loss of generality assume that the WEL is independent of wealth.}
	\begin{align*}
		\Phi(t,v(1-L^{\pi}(t,z)),z) = 	\Jp(t,v,z).
	\end{align*}
	An investor following the optimal portfolio allocation $\pi^{\ast}$ only needs $(1-L^{\pi}(t,z))$ as much capital to achieve the same average utility as an investor following the sub-optimal strategy $\pi.$ In this sense, $L^{\pi}(t,z)$ can be interpreted as a relative loss incurred for investing sub-optimally.\footnote{If $\pi$ is deterministic and $J^{\pi}$ is the unique solution to the Feynman-Kac PDE, one can use an exponentially affine ansatz to characterise $J^{\pi}$ in terms of the solutions to a system of ODEs. If the ODE solutions are given, then the WEL $L^{\pi}(0,z_0)$ is known in closed form. We provide a description of this approach in Lemma \ref{lem: utility sub-optimal strategy} and Corollary \ref{cor: Wealth-Equivalent Loss Deterministic Strategy} in the supplementary material. In our studies, we approximated the corresponding ODE solutions by an Euler method.}
	
	\subsection{Optimal Constrained Merton Portfolio $\pi = \text{Cap}\left(\pi_M,\alpha,\beta\right)$}\label{subsec: Optimal Constrained Merton}
	In this subsection, we compare $\pi^{\ast}$ with the first naive portfolio process $\pi= \text{Cap}(\pi_M, \alpha, \beta).$ Although $\pi$ is static, it is at least known that $\pi$ is optimal for an allocation constrained portfolio optimisation problem in the Black-Scholes market $\mathcal{M}_{BS},$ whereas no theoretical guarantees were available for the corresponding optimisation problem in $\mathcal{M}_H$ prior to this work. During this analysis, we consider the allocation constraint $K=[\alpha,\beta] = [0,1],$ which corresponds to a no-borrowing constraint that prevents short-selling in the risk-free and risky asset.

	\begin{figure}[H]
		\centering
		\includegraphics[width=.45\linewidth]{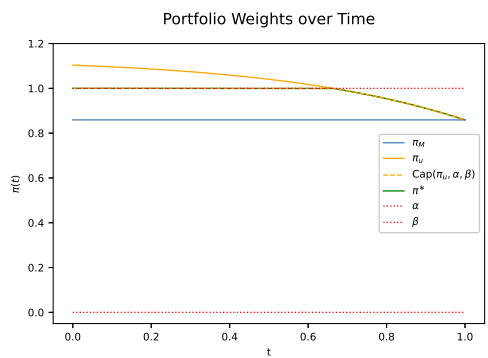}
		\caption{Portfolio weights $\pi(t)$ for $t\in [0,T]$, lower bound $\alpha = 0,$ upper bound $\beta = 1,$ and parameters as in Table \ref{tab: base market parameters}.}\label{fig: portfolio weights over time upper bound}
	\end{figure}
	
	Assuming the parameters in Table \ref{tab: base market parameters}, the Merton portfolio satisfies the constraints, i.e. $\pi_M \in [\alpha,\beta]$ and thus $\text{Cap}(\pi_M,\alpha,\beta) = \pi_M.$ In contrast to this constant allocation in the interior of $[\alpha,\beta],$ $\pi^{\ast}$ initially takes a constant value at the upper bound $\beta$ and later decreases towards $\pi_M$ at the end of the investment horizon. Therefore, other than $\text{Cap}(\pi_M,\alpha,\beta),$ $\pi^{\ast}$ is able to realise and benefit from a higher allocation to the risky asset. Note that $\pi_M \in [\alpha,\beta]$ implies $\pi^{\ast} = \text{Cap}(\pi_u,\alpha,\beta),$ as shown in Lemma \ref{lem: constrained portfolio equal capped portfolio}. \\

	In the following, we quantify the impact of the suboptimal allocations $\pi = \text{Cap}(\pi_M,\alpha,\beta)$ by computing the annual WEL $L^{\pi}(0,z)$ at the beginning of the investment horizon and analyse its sensitivity with respect to the risk-aversion parameter $b$ and the volatility drivers $\sigma,$ $\kappa$ and $\rho.$ The ranges of the volatility parameter are chosen to be within the minimum and maximum parameter values obtained in individual calibrations in Table 5 of \citet{Moyaert2011}.
	
	\begin{figure}[H]
		\begin{center}
			\subfigure[]{\label{fig: WEL sensitivity b}\includegraphics[width=.45\linewidth]{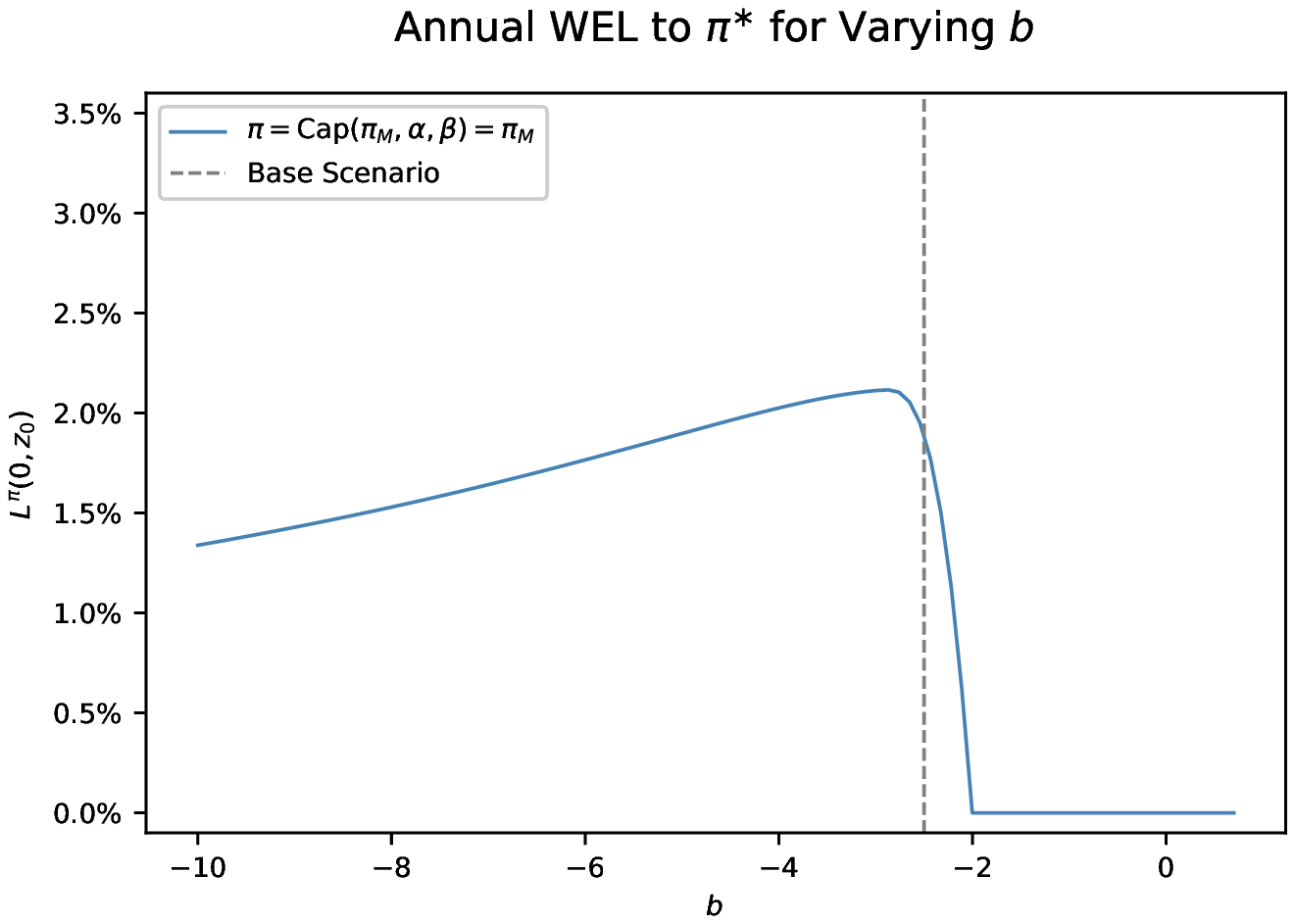}}
			\hspace{.05\linewidth}
			\subfigure[]{\label{fig: WEL sensitivity sigma}\includegraphics[width=.45\linewidth]{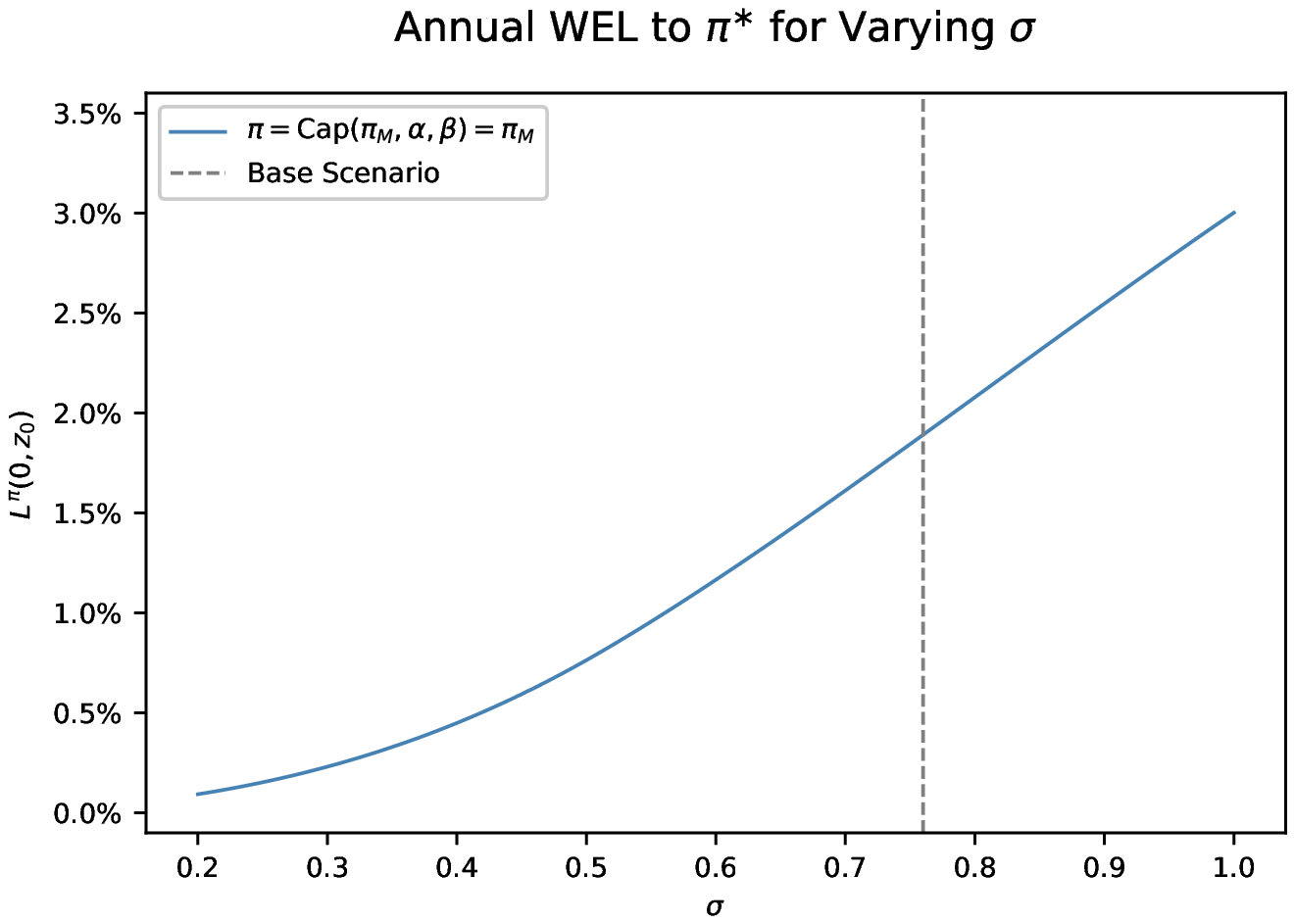}}
			\hspace{.05\linewidth}
			\subfigure[]{\label{fig: WEL sensitivity kappa}\includegraphics[width=.45\linewidth]{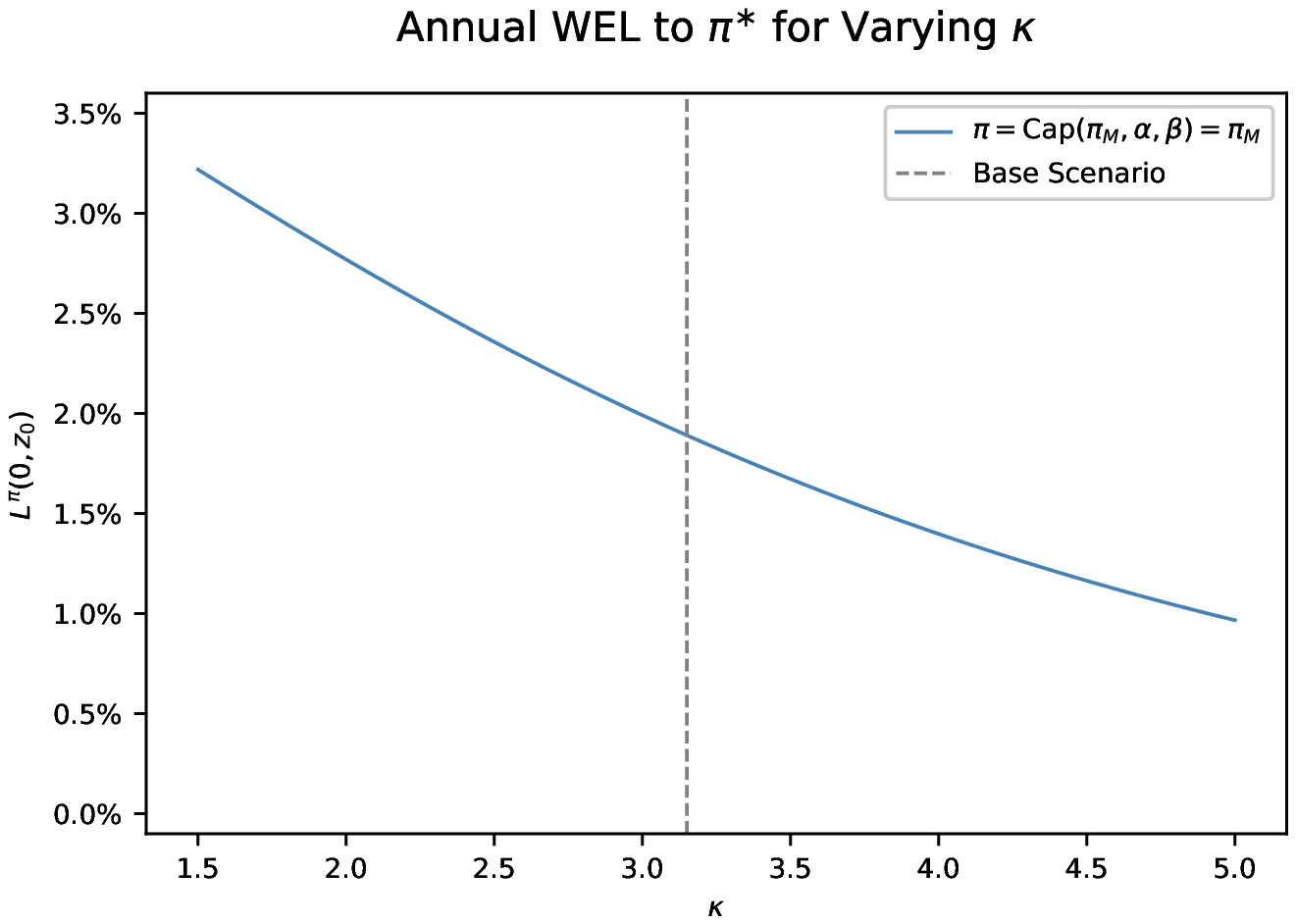}}
			\hspace{.05\linewidth}
			\subfigure[]{\label{fig: WEL sensitivity rho}\includegraphics[width=.45\linewidth]{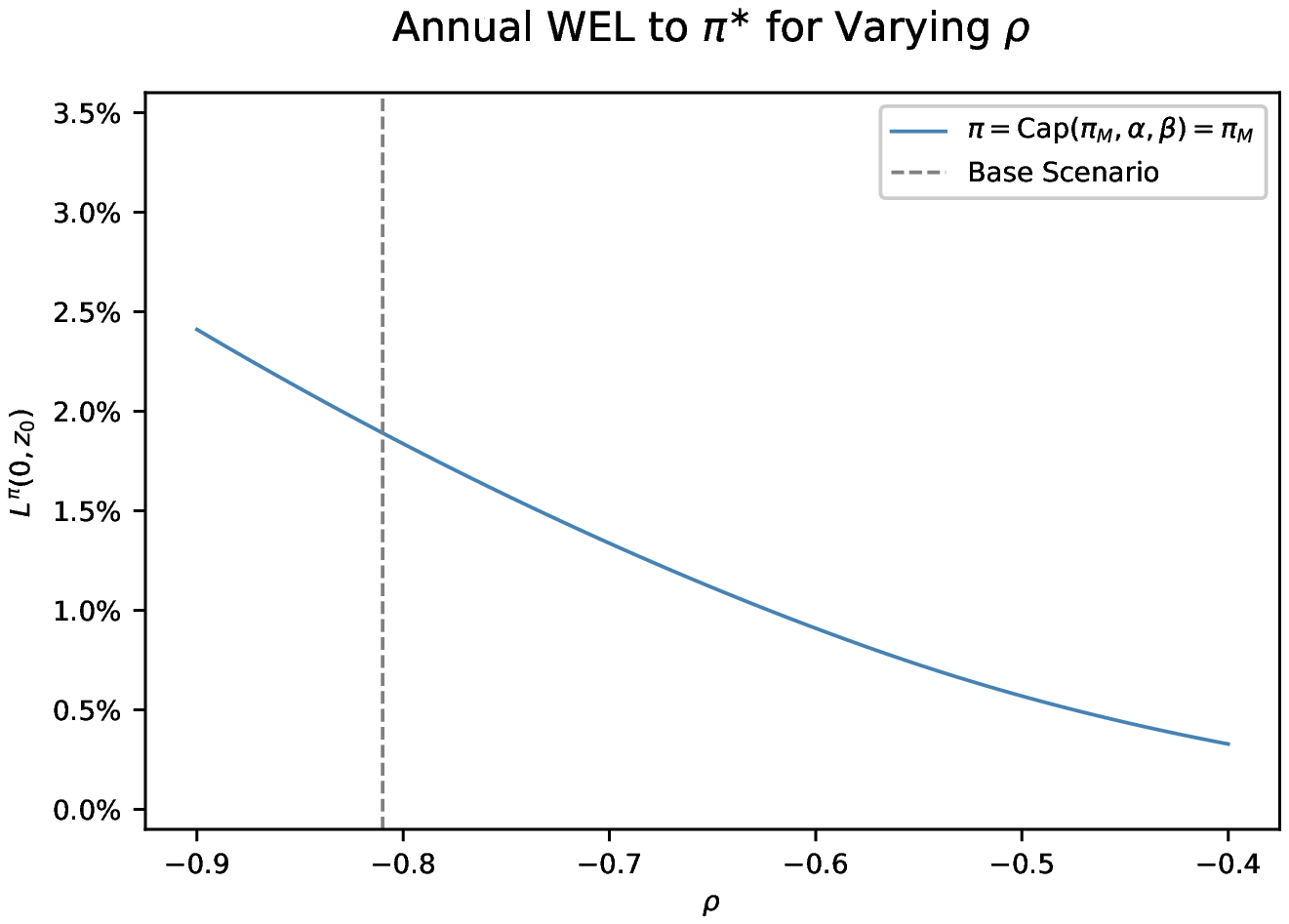}}
			\caption{Dependence of the annual wealth-equivalent loss $L^{\pi}(0,z_0)$ on individual parameters for $\pi = \text{Cap}(\pi_M,\alpha,\beta) = \pi_M,$ lower bound $\alpha = 0,$ upper bound $\beta = 1,$ and parameters as in Table \ref{tab: base market parameters}. Figure \ref{fig: WEL sensitivity b} illustrates the dependence on $b\in[-10,0.7]\backslash\{0\},$ whereas Figures \ref{fig: WEL sensitivity sigma}, \ref{fig: WEL sensitivity kappa} and \ref{fig: WEL sensitivity rho} display the dependence on $\sigma \in[0.2,1.0],$ $\kappa \in[1.5,5.0]$ and $\rho \in [-0.9,-0.4],$ respectively.}\label{fig: WEL Merton}
		\end{center}
	\end{figure}

	For small values of $b,$ where the allocation constraint $K = [0,1]$ is largely satisfied by the unconstrained portfolios $\pi_M$ and $\pi_u,$ the WELs displayed in Figure \ref{fig: WEL sensitivity b} are increasing in $b.$ However, as $b$ increases past an inflection point of approximately $b = -3,$ the allocation constraint $K$ becomes active. From this point onwards, $K$ forces $\pi^{\ast}$ and $\text{Cap}(\pi_M,\alpha,\beta)$ closer towards each other for increasing $b$ and therefore leads to decreasing WELs. Ultimately, for $b \geq -2,$ we have $\pi^{\ast}(t) = \text{Cap}(\pi_M,\alpha,\beta) = \beta = 1$ for all $t \in [0,T]$ and thus the WEL is zero.
	Figures \ref{fig: WEL sensitivity sigma}, \ref{fig: WEL sensitivity kappa} and \ref{fig: WEL sensitivity rho} display WELs which are increasing in $\sigma$ as well as decreasing in $\kappa$ and $\rho.$ This confirms the intuition voiced at the beginning of Section \ref{sec: Numerical Studies}, in which we argued that $\mathcal{M}_H$ is \enquote{more different} to $\mathcal{M}_{BS}$ under these circumstances. Within the chosen parameter ranges, we observe the largest annual WEL of $3.2\%$ for small $\kappa,$ whereas increasing $\sigma$ and decreasing $\rho$ leads to WELs of $3.0\%$ and $2.5\%,$ respectively. Changing any one of the values of the volatility parameters $\sigma,$ $\kappa$ or $\rho$ to less extreme levels, which are obtained during calibrations on long-term data sets\footnote{For an overview, consider e.g. Table 4 in \cite{Escobar-Anel2016}. Here, the authors consider values of $\kappa = 3.5,$ $\sigma = 0.3,$ $\rho = -0.4$ as an \enquote{Average Case} for their reviewed literature.} leads to significant decreases in annual WELs. Even if only $\sigma<0.5,$ while $\kappa$ and $\rho$ remain at crisis level, the annual WEL still drops to values of $0.75\%$ or lower. Note that Feller's condition is satisfied for all parameter values that were considered in our analysis.
	
	\subsection{Capped Optimal Unconstrained Heston Portfolio $\pi = \text{Cap}\left(\pi_u,\alpha,\beta\right)$}\label{subsec: Capped Optimal Unconstrained Heston Portfolio}
	In this subsection, we compare $\pi^{\ast}$ to the second naive portfolio process, the capped optimal unconstrained Heston portfolio $\text{Cap}\left(\pi_u,\alpha,\beta\right).$ In particular, we would like to illustrate the phenomenon described in Section \ref{subsec: Comparison to Unconstrained Portfolio}. Despite having a theoretical guarantee that $\text{Cap}\left(\pi_u,\alpha,\beta\right)$ is indeed different from $\pi^{\ast}$ for certain parameter settings, these differences appear to be mostly meaningful in terms of WEL for extreme market scenarios in combination with specific, large lower bounds $\alpha.$ According to Lemma \ref{lem: equivalence gap} and Corollary \ref{cor: stationary case}, we know that $\text{Cap}\left(\pi_u,\alpha,\beta\right)$ and $\pi^{\ast}$ are identical in the parameter setting of Table \ref{tab: base market parameters}, unless $\pi_M < \alpha.$ Further, Corollary \ref{cor: stationary case} suggests that we should consider market parameters which ensure $\pi_u(t)>2\pi_M$ for some $t\in [0,T].$ \\
	For these reasons, we adjust the previously considered parameter setting. In the following, we choose the most extreme volatility parameters from the sensitivity analysis in Figure \ref{fig: WEL Merton}, i.e. we set $\sigma = 1.0,$ $\kappa = 1.5$ and $\rho = -0.9,$ and increase the risk aversion coefficient to $b=-15$ to obtain realistic portfolio allocations.
	
	\begin{figure}[H]
		\begin{center}
			\subfigure[]{\label{fig: Portfolio Gap Heston Scale}\includegraphics[width=.45\linewidth]{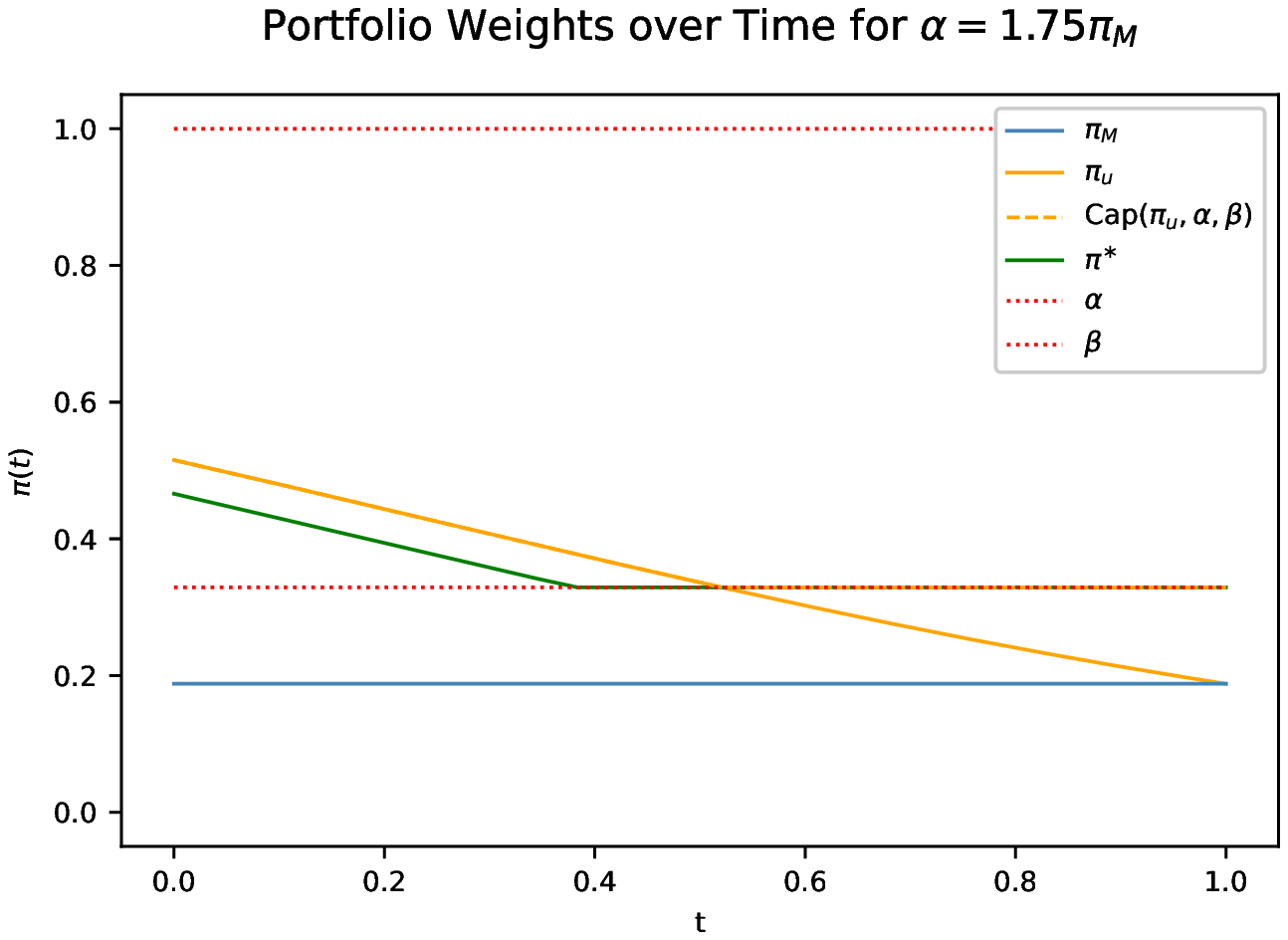}}
			\hspace{.05\linewidth}
			\subfigure[]{\label{fig: Portfolio Gap Heston Stationary}\includegraphics[width=.45\linewidth]{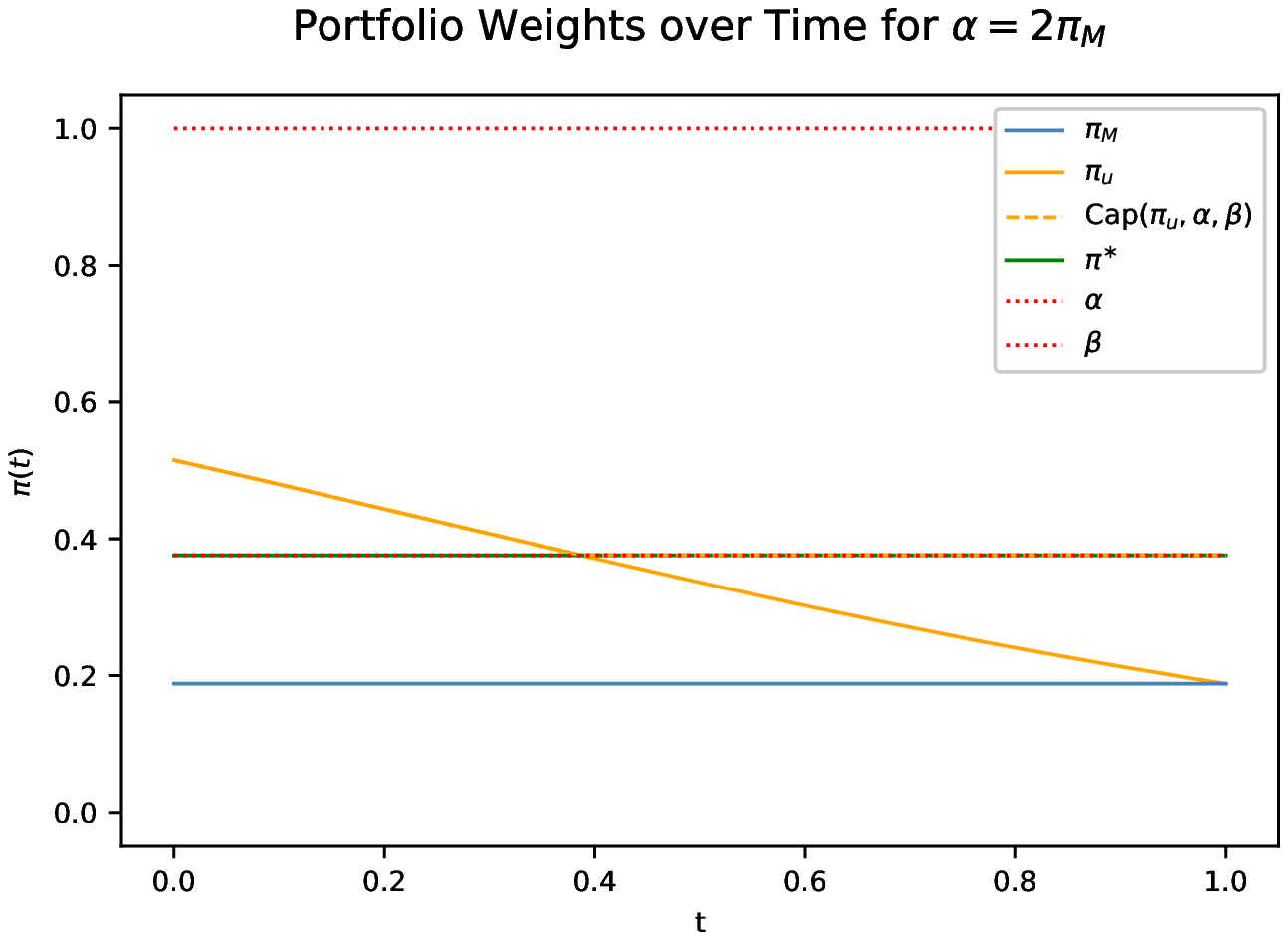}}
			\caption{Portfolio weights $\pi(t)$ for $t\in [0,T],$ upper bound $\beta = 1,$ and parameters as in Table \ref{tab: base market parameters}, except for $b=15,$ $\sigma = 1.0,$ $\kappa = 1.5$ and $\rho = -0.9.$ Figure \ref{fig: Portfolio Gap Heston Scale} considers a lower bound $\alpha = 1.75\pi_M,$ and Figure \ref{fig: Portfolio Gap Heston Stationary} considers a lower bound $\alpha = 2\pi_M.$}\label{fig: Portfolio Gap Heston}
		\end{center}
	\end{figure}
	
	Figure \ref{fig: Portfolio Gap Heston} compares the portfolio weights of $\pi^{\ast}(t)$ and $\text{Cap}(\pi_u(t),\alpha,\beta)$ for lower bounds $\alpha \in \{1.75\pi_M, 2\pi_M\}$ such that the corresponding ODE solution $B$ is (nearly) constant, as described in Corollary \ref{cor: stationary case}. If $\alpha = 1.75\pi_M$ then $\pi^{\ast}$ is initially larger than the lower bound $\alpha,$ then decreases until $\alpha$ is reached, while $\pi^{\ast}$ is constant for $\alpha = 2 \pi_M.$ Additionally, we observe that
	$$
	\pi^{\ast}(t) \leq  \text{Cap}(\pi_u(t),\alpha,\beta) \quad  \forall t\in [0,T],
	$$
	with equality only if $\pi_u(t)\leq \alpha.$  In both cases illustrated in Figure \ref{fig: Portfolio Gap Heston}, $\pi^{\ast}$ lowers the portfolio allocation early throughout the investment horizon, thus accounting for the fact that the lower bound forces the portfolio allocation to be larger than the optimal unconstrained allocation later during the investment horizon.

	\begin{figure}[H]
		\begin{center}
			\subfigure[]{\label{fig: Portfolio Gap Heston Varying alpha}\includegraphics[width=.45\linewidth]{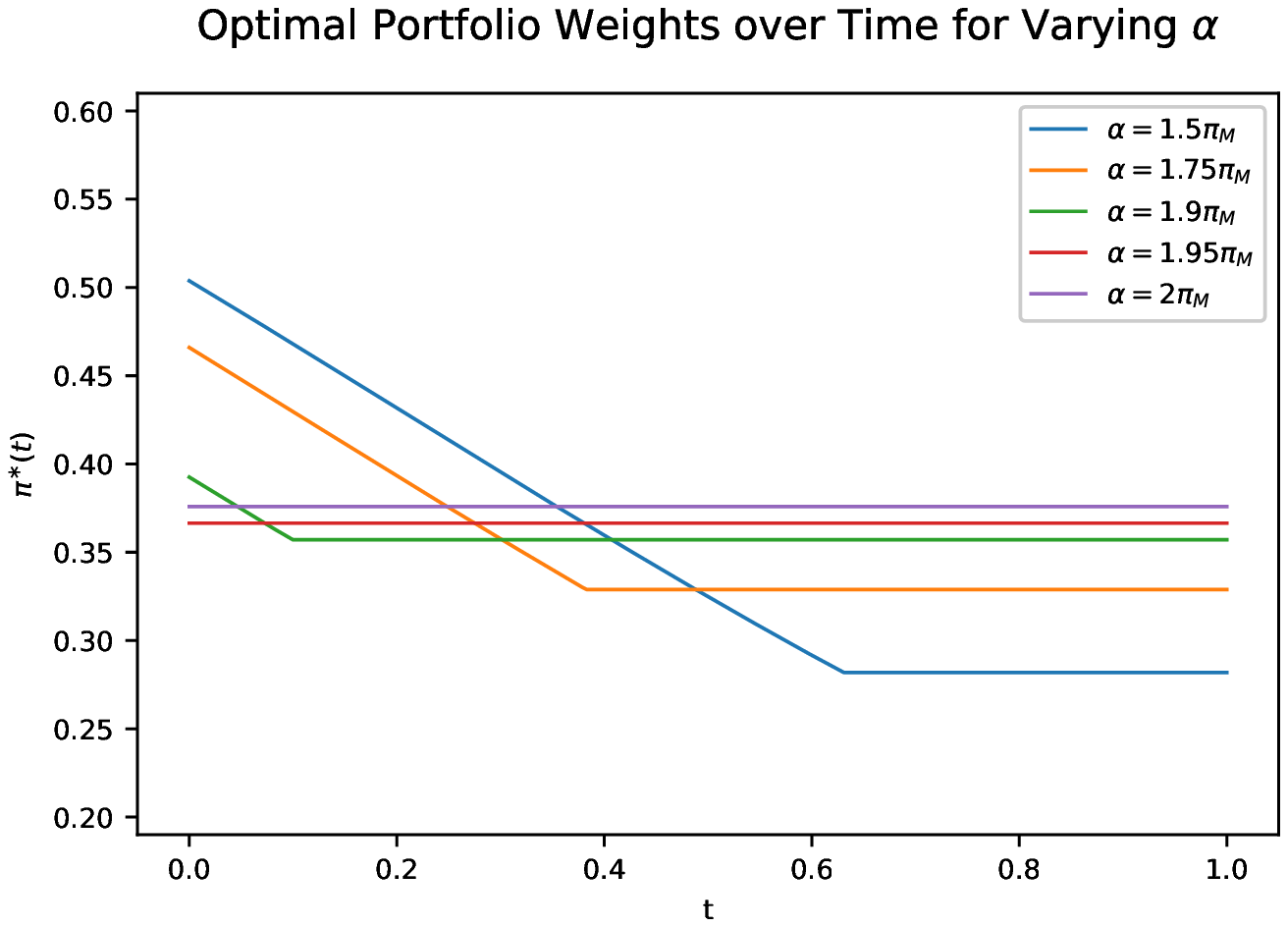}}
			\hspace{.05\linewidth}
			\subfigure[]{\label{fig: WEL Heston Varying alpha}\includegraphics[width=.45\linewidth]{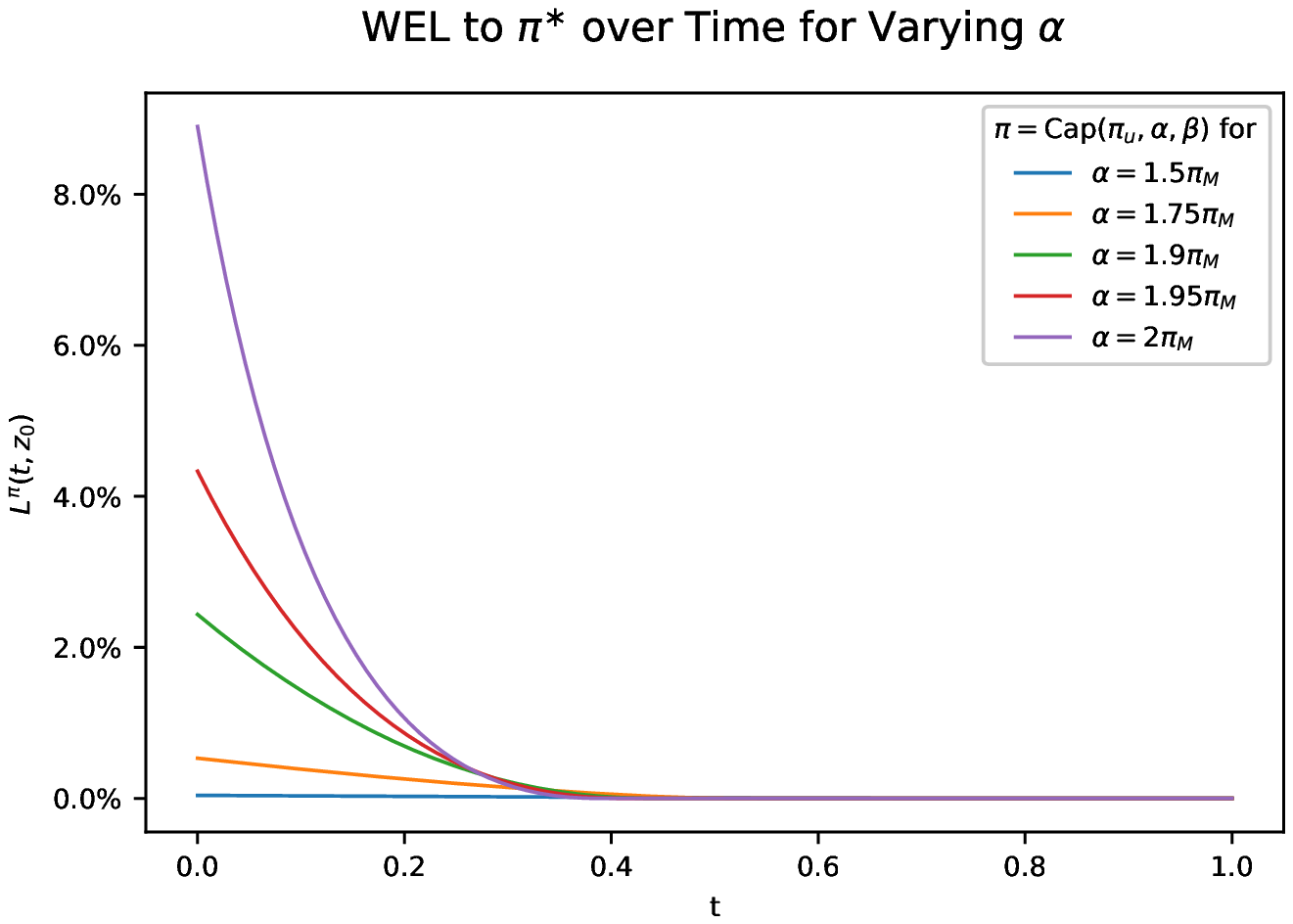}}
			\caption{Lower bounds $\alpha \in \{1.5\pi_M, 1.75\pi_M, 1.9 \pi_M, 1.95\pi_M, 2\pi_M\},$ upper bound $\beta = 1,$ $t\in [0,T],$ and parameters as in Table \ref{tab: base market parameters}, except for $b=15,$ $\sigma = 1.0,$ $\kappa = 1.5$ and $\rho = -0.9.$ Figure \ref{fig: Portfolio Gap Heston Varying alpha} displays the portfolio weights $\pi^{\ast}(t),$ and Figure \ref{fig: WEL Heston Varying alpha} displays the WEL $\L^{\pi}(t,z_0)$ for $\pi = \text{Cap}\left(\pi_u,\alpha,\beta\right).$}\label{fig: Illustration Portfolio and WEL varying alpha}
		\end{center}
		\
	\end{figure}
	Figure \ref{fig: Portfolio Gap Heston Varying alpha} illustrates the behaviour of $\pi^{\ast}$ for varying $\alpha.$ When increasing $\alpha,$ we observe that $\pi^{\ast}$ decreases to the lower bound at earlier time points $t.$ Despite $\pi^{\ast}$ being constant for $\alpha = 1.95\pi_M,$ the solution $B$ to the ODE (\ref{eq: ODE zones}) for $B$ is not stationary, but $\rho B(\tau)$ does not leave zone $Z_{-}$ for $\tau \leq T.$ Thus, we expect $\pi^{\ast}$ to be constant for all $\alpha \in [1.95\pi_M,2\pi_M]$ in our parameter setting. Figure \ref{fig: Portfolio Gap Heston Varying alpha} displays WELs $L^{\pi}(t,z_0)$ of $\pi = \text{Cap}(\pi_u,\alpha,\beta),$ which are increasing in $\alpha$ at $t=0.$ However, this monotonicity does not hold throughout the entire investment horizon, as increasing the lower bound $\alpha$ implies that $\pi^{\ast}$ and $\text{Cap}(\pi_u,\alpha,\beta)$ coincide for longer parts of the investment horizon. \\
	Clearly, Figures \ref{fig: Portfolio Gap Heston} and \ref{fig: Illustration Portfolio and WEL varying alpha} suggest a strong link between the value of $\alpha$ and the difference between $\pi^{\ast}$ and $\text{Cap}(\pi_u,\alpha,\beta).$ Therefore, we quantify this difference not only using WELs, but additionally define the maximum absolute weight difference between $\pi^{\ast}$ and a portfolio $\pi$ as
	\begin{align}\label{eq: max difference}
		\Delta^{\pi}_{\max}:= \max_{t\in [0,T]}\Big |\pi(t)-\pi^{\ast}(t)\Big |.
	\end{align}
	The relationship between the lower bound $\alpha$ and the maximum absolute difference $\Delta^{\pi}_{\max}$ and the annual WEL $L^{\pi}(0,z_0)$ are analysed in Figure \ref{fig: max Gap and WEL Heston}.
	
	\begin{figure}[H]
		\begin{center}
			\subfigure[]{\label{fig: max Gap Heston}\includegraphics[width=.45\linewidth]{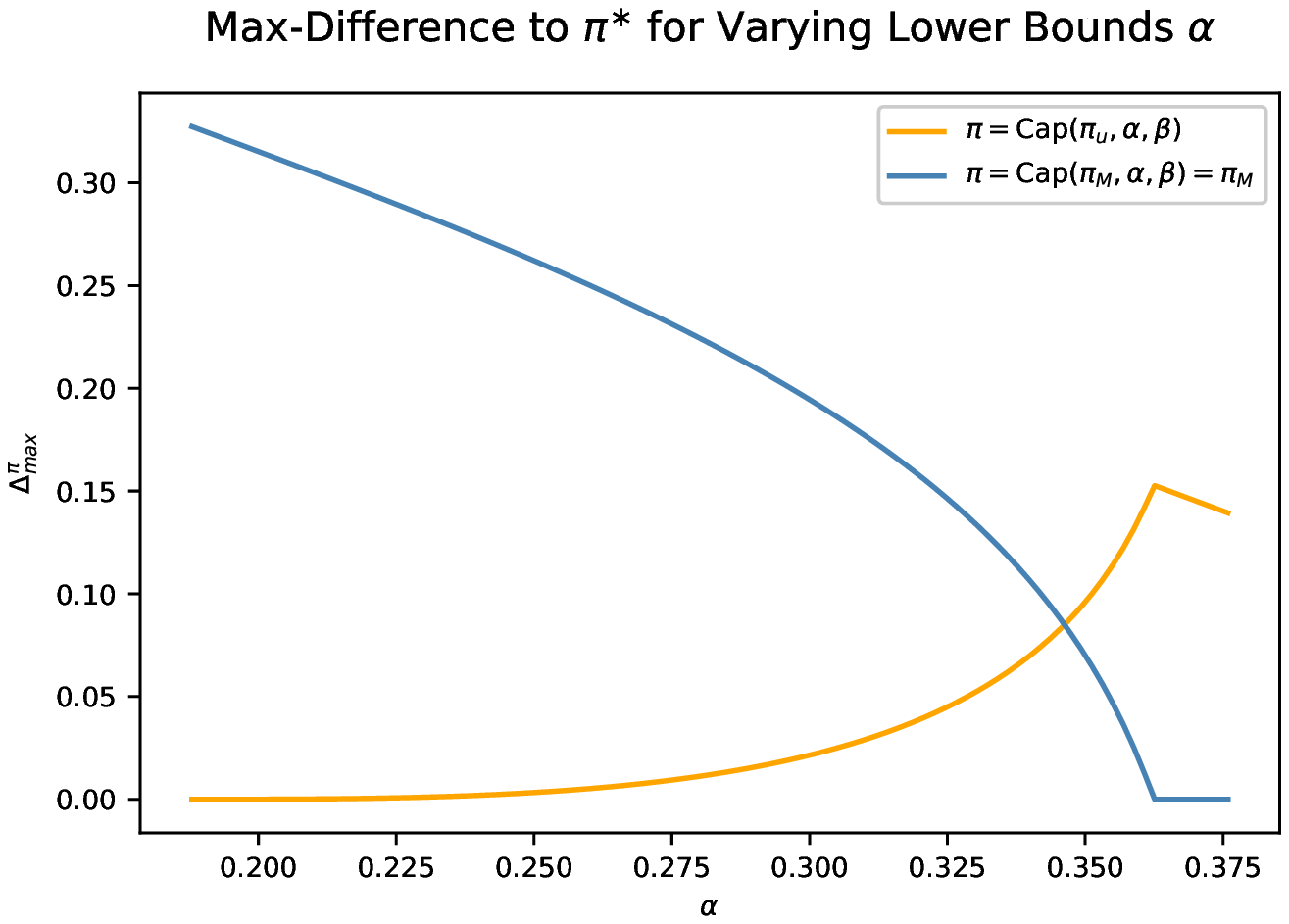}}
			\hspace{.05\linewidth}
			\subfigure[]{\label{fig: max WEL Heston}\includegraphics[width=.45\linewidth]{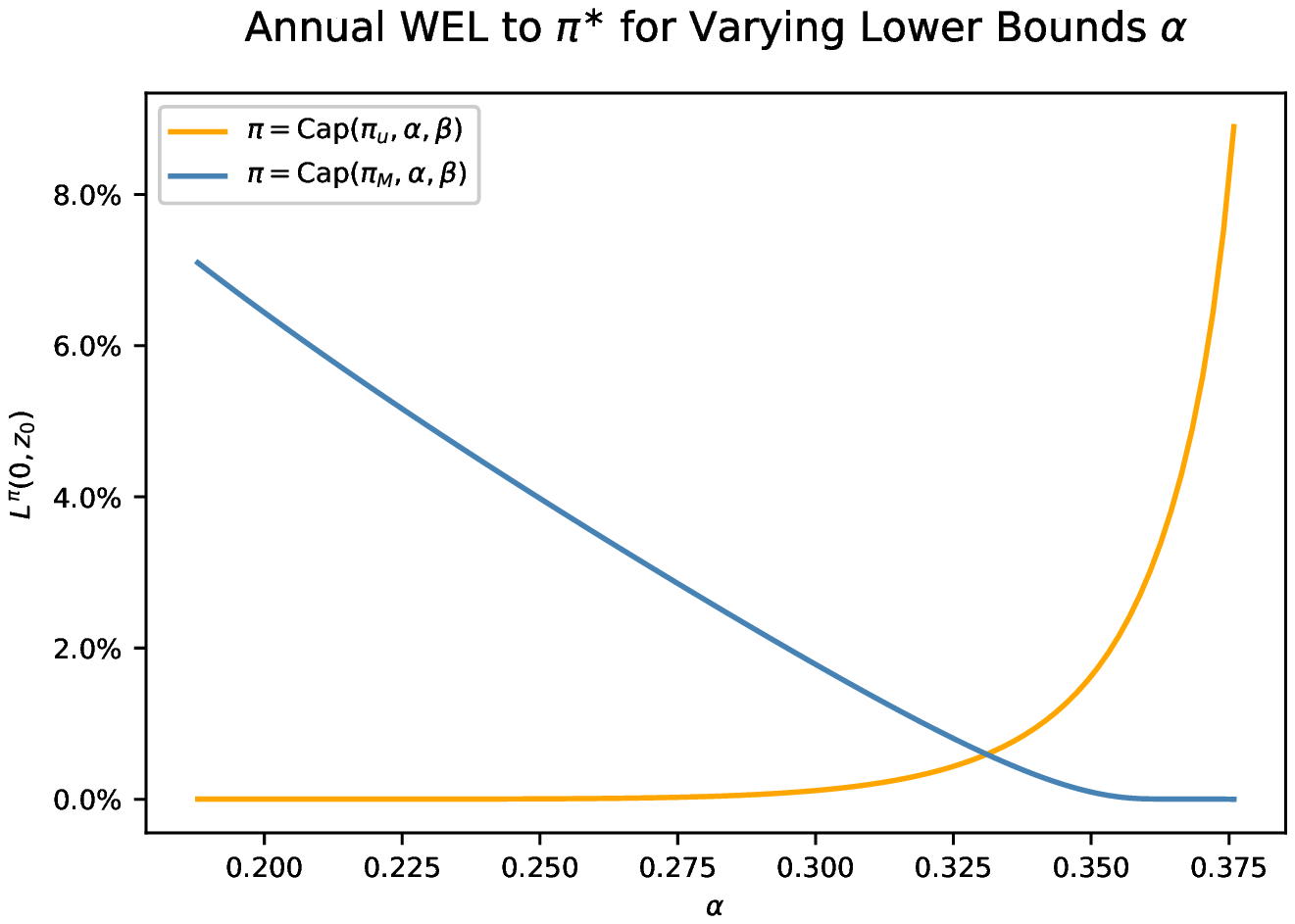}}
			\caption{Lower bounds $\alpha \in [\pi_M, 2\pi_M],$ upper bound $\beta = 1,$ $\pi= \text{Cap}\left(\pi_u, \alpha,\beta\right),$ and parameters as in Table \ref{tab: base market parameters}, except for $b=15,$ $\sigma = 1.0,$ $\kappa = 1.5$ and $\rho = -0.9.$ For $\pi = \text{Cap}\left(\pi_u,\alpha,\beta\right),$ Figure \ref{fig: max Gap Heston} displays the maximum absolute difference $\Delta^{\pi}$ to the portfolio weights $\pi^{\ast}(t),$ and Figure \ref{fig: max WEL Heston} displays the WEL $\L^{\pi}(0,z_0)$ to $\pi^{\ast}.$}\label{fig: max Gap and WEL Heston}
		\end{center}
	\end{figure}
	
	For $\pi = \text{Cap}(\pi_u,\alpha,\beta),$ the maximum absolute difference $\Delta^{\pi}_{\max}$ is generally increasing with $\alpha$ except for larger lower bounds $\alpha.$ For large $\alpha,$ the optimal constrained portfolio $\pi^{\ast}$ is constant throughout the investment horizon, as illustrated in Figure \ref{fig: Portfolio Gap Heston Varying alpha}. Then, the monotonicity of $\pi_u$ (see Remark \ref{rem: monotonicty of B}) and $\pi^{\ast}(T) = \alpha = \text{Cap}(\pi_u(T),\alpha,\beta)$ ensure that the maximum in (\ref{eq: max difference}) is attained at $t=0.$ Further, $\pi_u(0)>2\pi_M\geq \alpha$ does not depend on $\alpha$.  Therefore, if $\pi^{\ast}(t) = \alpha$ for all $t\in[0,T],$ then the difference 
	$$
	\Delta^{\text{Cap}(\pi_u,\alpha,\beta)}_{\max} = \max_{t\in [0,T]}\Big |\text{Cap}(\pi_u(t),\alpha,\beta)-\pi^{\ast}(t)\Big | = \Big | \pi_u(0)-\alpha \Big | = \pi_u(0)-\alpha
	$$
	decreases linearly in $\alpha,$ which causes a slight kink in Figure \ref{fig: max Gap Heston} for large lower bounds $\alpha.$ Irrespectively, the annual WEL for $\pi = \text{Cap}(\pi_u,\alpha,\beta)$ is increasing with $\alpha.$ However, note that for all but very large lower bounds (e.g. $\alpha \geq 1.75\pi_M$), the annual WEL is still negligible.

\section{Conclusion}\label{sec: Conclusion}
In this paper, we considered a portfolio optimisation problem with allocation constraints in Heston's stochastic volatility model. We derived an explicit expression for the optimal portfolio and analysed its properties. Surprisingly, this portfolio can be different from the naive constrained portfolio which caps off the optimal unconstrained portfolio at the boundaries of the constraint. In light of this fact, we have shown that the addition of allocation constraints can have a fundamentally different impact on the optimal portfolio in markets with stochastic volatility as compared to a Black-Scholes market with constant volatility - even in financial markets with only one risky asset. Irrespective of these theoretical certainties, we observed in a numerical study that the annual wealth equivalent loss incurred due to trading according to this naive portfolio is relatively small for the majority of realistic scenarios. In this sense, the naive \enquote{capped} portfolio is nearly optimal for most applications. However, in turbulent financial markets, such as the financial crisis of 2008, investors with a high degree of risk aversion and lower bound on their portfolio allocation can suffer high wealth equivalent losses. For such scenarios, investors should be mindful of the optimal constrained portfolio and the naive capped portfolio.

\bibliographystyle{rQUF}
\bibliography{references}

\begin{thebibliography}{43}
\providecommand{\natexlab}[1]{#1}
\providecommand{\noopsort}[1]{}
\providecommand{\printfirst}[2]{#1}
\providecommand{\singleletter}[1]{#1}
\providecommand{\switchargs}[2]{#2#1}

\bibitem[\protect\citeauthoryear{B\"{a}uerle and Li}{2013}]{baeuerle2013}
B\"{a}uerle, N. and Li, Z., Optimal Portfolios for Financial Markets with
  Wishart Volatility. {\itshape Journal of Applied Probability}, 2013,
  \textbf{50}, 1025–1043.

\bibitem[\protect\citeauthoryear{Bian {\itshape{et~al.}}}{2011}]{Zheng2011}
Bian, B., Miao, S. and Zheng, H., Smooth Value Functions for a Class of
  Nonsmooth Utility Maximization Problems. {\itshape SIAM Journal on Financial
  Mathematics}, 2011, \textbf{2}, 727--747.

\bibitem[\protect\citeauthoryear{Branger
  {\itshape{et~al.}}}{2008}]{Branger2008}
Branger, N., Schlag, C. and Schneider, E., Optimal Portfolios When Volatility
  Can Jump. {\itshape Journal of Banking \& Finance}, 2008, \textbf{32}, 1087
  -- 1097.

\bibitem[\protect\citeauthoryear{Chen {\itshape{et~al.}}}{2021}]{Chen2021}
Chen, A., Nguyen, T. and Rach, M., {A Collective Investment Problem in a
  Stochastic Volatility Environment: The Impact of Sharing Rules}. {\itshape
  Annals of Operations Research}, 2021, \textbf{302}, 85--109.

\bibitem[\protect\citeauthoryear{Cheng and Escobar-Anel}{2021}]{Escobar2021}
Cheng, Y. and Escobar-Anel, M., Optimal Investment Strategy in the Family of
  4/2 Stochastic Volatility Models. {\itshape Quantitative Finance}, 2021,
  \textbf{21}, 1723--1751.

\bibitem[\protect\citeauthoryear{Cont}{2001}]{Cont2000}
Cont, R., Empirical Properties of Asset Returns: Stylized Facts and Statistical
  Issues. {\itshape Quantitative Finance}, 2001, \textbf{1}, 223--236.

\bibitem[\protect\citeauthoryear{Cuoco}{1997}]{Cuoco1997}
Cuoco, D., Optimal Consumption and Equilibrium Prices with Portfolio
  Constraints and Stochastic Income. {\itshape Journal of Economic Theory},
  1997, \textbf{72}, 33 -- 73.

\bibitem[\protect\citeauthoryear{Cvitanic and Karatzas}{1992}]{cvitanic1992}
Cvitanic, J. and Karatzas, I., Convex Duality in Constrained Portfolio
  Optimization. {\itshape Annals of Applied Probability}, 1992, \textbf{2},
  767--818.

\bibitem[\protect\citeauthoryear{Dong and Zheng}{2020}]{Dong2020}
Dong, Y. and Zheng, H., Optimal Investment with S-shaped Utility and Trading
  and Value at Risk Constraints: An Application to Defined Contribution Pension
  Plan. {\itshape European Journal of Operational Research}, 2020,
  \textbf{281}, 341 -- 356.

\bibitem[\protect\citeauthoryear{Durrett}{2019}]{Durrett2019}
Durrett, R., {\itshape Probability: Theory and Examples. 5th Edition}, 2019
  (Cambridge University Press: Princeton, NJ).

\bibitem[\protect\citeauthoryear{Egloff {\itshape{et~al.}}}{2010}]{Egloff2010}
Egloff, D., Leippold, M. and Wu, L., The Term Structure of Variance Swap Rates
  and Optimal Variance Swap Investments. {\itshape The Journal of Financial and
  Quantitative Analysis}, 2010, \textbf{45}, 1279--1310.

\bibitem[\protect\citeauthoryear{Escobar-Anel
  {\itshape{et~al.}}}{2017}]{Rubtsov2017}
Escobar-Anel, M., Ferrando, S. and Rubtsov, A., Optimal Investment under
  Multi-Factor Stochastic Volatility. {\itshape Quantitative Finance}, 2017,
  \textbf{17}, 241--260.

\bibitem[\protect\citeauthoryear{Escobar-Anel and
  Gschnaidtner}{2016}]{Escobar-Anel2016}
Escobar-Anel, M. and Gschnaidtner, C., Parameters Recovery via Calibration in
  the Heston Model: A Comprehensive Review. {\itshape Wilmott}, 2016,
  \textbf{2016}, 60--81.

\bibitem[\protect\citeauthoryear{Escobar-Anel
  {\itshape{et~al.}}}{2010}]{Escobar2010}
Escobar-Anel, M., Götz, B., Seco, L. and Zagst, R., Pricing a CDO on
  stochastically correlated underlyings. {\itshape Quantitative Finance}, 2010,
  \textbf{10}, 265--277.

\bibitem[\protect\citeauthoryear{{Escobar-Anel}
  {\itshape{et~al.}}}{2023}]{EKZ2023}
{Escobar-Anel}, M., {Kschonnek}, M. and {Zagst}, R., {Portfolio Optimization
  with Allocation Constraints and Stochastic Factor Market Dynamics}. {\itshape
  arXiv e-prints}, 2023, arXiv:2303.09835.

\bibitem[\protect\citeauthoryear{Filipovic}{2009}]{Filipovic2009}
Filipovic, D., {\itshape Term-Structure Models}, 2009, Springer Berlin,
  Heidelberg.

\bibitem[\protect\citeauthoryear{Geman {\itshape{et~al.}}}{2003}]{Carr2003}
Geman, H., Carr, P., Madan, D. and Yor, M., {Stochastic Volatility for Levy
  Processes}. {\itshape {Mathematical Finance}}, 2003, \textbf{13}, 345--382.

\bibitem[\protect\citeauthoryear{Gikhman}{2011}]{Gikhman2011}
Gikhman, I., A Short Remark on Feller’s Square Root Condition. {\itshape SSRN
  Electronic Journal}, 2011.

\bibitem[\protect\citeauthoryear{Heston}{1993}]{Heston1993}
Heston, S., A Closed-Form Solution for Options with Stochastic Volatility with
  Applications to Bond and Currency Options. {\itshape Review of Financial
  Studies}, 1993, \textbf{6}, 327--343.

\bibitem[\protect\citeauthoryear{Hiriart-Urruty and
  Lemaréchal}{2001}]{Hiriart-Urruty2001}
Hiriart-Urruty, J.B. and Lemaréchal, C., {\itshape Fundamentals of Convex
  Analysis}, 2001, Springer Berlin, Heidelberg.

\bibitem[\protect\citeauthoryear{Kallsen and Muhle-Karbe}{2010}]{Kallsen2010a}
Kallsen, J. and Muhle-Karbe, J., Utility Maximization in Affine Stochastic
  Volatility Models. {\itshape International Journal of Theoretical and Applied
  Finance}, 2010, \textbf{13}, 459--477.

\bibitem[\protect\citeauthoryear{Karatzas
  {\itshape{et~al.}}}{1987}]{karatzas1987}
Karatzas, I., Lehoczky, J.P. and Shreve, S.E., Optimal Portfolio and
  Consumption Decisions for a “Small Investor” on a Finite Horizon.
  {\itshape SIAM Journal on Control and Optimization}, 1987, \textbf{25},
  1557--1586.

\bibitem[\protect\citeauthoryear{Karatzas
  {\itshape{et~al.}}}{1991}]{karatzas1991}
Karatzas, I., Lehoczky, J.P., Shreve, S.E. and Xu, G.L., Martingale and Duality
  Methods for Utility Maximization in an Incomplete Market. {\itshape SIAM
  Journal on Control and Optimization}, 1991, \textbf{29}, 702--730.

\bibitem[\protect\citeauthoryear{Koijen}{2014}]{Koijen2014}
Koijen, R.S., The Cross-Section of Managerial Ability, Incentives, and Risk
  Preferences. {\itshape The Journal of Finance}, 2014, \textbf{69},
  1051--1098.

\bibitem[\protect\citeauthoryear{Korn and Kraft}{2004}]{Korn2004}
Korn, R. and Kraft, H., On the Stability of Continuous-Time Portfolio Problems
  with Stochastic Opportunity Set. {\itshape Mathematical Finance}, 2004,
  \textbf{14}, 403--414.

\bibitem[\protect\citeauthoryear{Kraft}{2005}]{Kraft2005}
Kraft, H., Optimal portfolios and Heston's Stochastic Volatility Model: An
  Explicit Solution for Power Utility. {\itshape Quantitative Finance}, 2005,
  \textbf{5}, 303--313.

\bibitem[\protect\citeauthoryear{Lindberg}{2006}]{Lindberg2006a}
Lindberg, C., News-Generated Dependence and Optimal Portfolios for n Stocks in
  a Market of Barndorff-Nielsen and Shephard Type. {\itshape Mathematical
  Finance}, 2006, \textbf{16}, 549--568.

\bibitem[\protect\citeauthoryear{Liu}{2006}]{Liu2006}
Liu, J., Portfolio Selection in Stochastic Environments. {\itshape The Review
  of Financial Studies}, 2006, \textbf{20}, 1--39.

\bibitem[\protect\citeauthoryear{Liu and Pan}{2003}]{Liu2003}
Liu, J. and Pan, J., Dynamic Derivative Strategies. {\itshape Journal of
  Financial Economics}, 2003, \textbf{69}, 401--430.

\bibitem[\protect\citeauthoryear{Lux}{2009}]{Lux2008}
Lux, T., CHAPTER 3 - Stochastic Behavioral Asset-Pricing Models and the
  Stylized Facts. In {\itshape Handbook of Financial Markets: Dynamics and
  Evolution}, Handbooks in Finance, pp. 161--215, 2009, North-Holland.

\bibitem[\protect\citeauthoryear{Merton}{1971}]{merton1971}
Merton, R.C., Optimum consumption and portfolio rules in a continuous-time
  model. {\itshape Journal of Economic Theory}, 1971, \textbf{4}, 373--413.

\bibitem[\protect\citeauthoryear{Mnif}{2007}]{Mnif2007}
Mnif, M., Portfolio Optimization with Stochastic Volatilities and Constraints:
  An Application in High Dimension. {\itshape Applied Mathematics and
  Optimization}, 2007, \textbf{56}, 243–--264.

\bibitem[\protect\citeauthoryear{Moyaert and Petitjean}{2011}]{Moyaert2011}
Moyaert, T. and Petitjean, M., The Performance of Popular Stochastic Volatility
  Option Pricing Models During the Subprime Crisis. {\itshape Applied Financial
  Economics}, 2011, \textbf{21}, 1059--1068.

\bibitem[\protect\citeauthoryear{Nutz}{2012}]{Nutz2012}
Nutz, M., Power Utility Maximization in Constrained Exponential L\'{e}vy
  Models. {\itshape Mathematical Finance}, 2012, \textbf{22}, 690--709.

\bibitem[\protect\citeauthoryear{Pham}{2002}]{Pham2002}
Pham, H., Smooth Solutions to Optimal Investment Models with Stochastic
  Volatilities and Portfolio Constraints. {\itshape Applied Mathematics \&
  Optimization}, 2002, \textbf{46}, 55--78.

\bibitem[\protect\citeauthoryear{Pliska}{1986}]{Pliska1986}
Pliska, S.R., A Stochastic Calculus Model of Continuous Trading: Optimal
  Portfolios. {\itshape Mathematics of Operations Research}, 1986, \textbf{11},
  371--382.

\bibitem[\protect\citeauthoryear{Q.ai}{2022}]{Forbes2022}
Q.ai, The Average Bear Market Lasts 289 Days. How Long Do We Have Left?.
  {\itshape Forbes}, 2022.

\bibitem[\protect\citeauthoryear{Taylor}{1994}]{Taylor1994}
Taylor, S.J., Modeling Stochastic Volatility: A Review and Comparative Study.
  {\itshape Mathematical Finance}, 1994, \textbf{4}, 183--204.

\bibitem[\protect\citeauthoryear{Walter}{1998}]{Walter1998}
Walter, W., {\itshape Ordinary Differential Equations}, 1998, Springer New
  York, NY.

\bibitem[\protect\citeauthoryear{Weatherall}{2018}]{Weatherall2018}
Weatherall, J.O., The Peculiar Logic of the Black-Scholes Model. {\itshape
  Philosophy of Science}, 2018, \textbf{85}, 1152--1163.

\bibitem[\protect\citeauthoryear{Zariphopoulou}{1994}]{Zariph1994}
Zariphopoulou, T., Consumption-Investment Models with Constraints. {\itshape
  SIAM Journal on Control and Optimization}, 1994, \textbf{32}, 59--85.

\bibitem[\protect\citeauthoryear{Zariphopoulou}{2001}]{Zariphopoulou2001}
Zariphopoulou, T., A Solution Approach to Valuation with Unhedgeable Risks.
  {\itshape Finance and Stochastics}, 2001, \textbf{5}, 61--82.

\bibitem[\protect\citeauthoryear{Zariphopoulou}{2009}]{Zariphopoulou2009}
Zariphopoulou, T.In {\itshape Optimal asset allocation in a stochastic factor
  model – an overview and open problems}, pp. 427--456, 2009, De Gruyter.

\end{thebibliography}

\appendix 

\section{Proofs}\label{sec: App. Proofs}
%%%%%%%%%%%%%%
\vspace{0.25cm}
%%%%%%%%%%%%%%
\begin{proof}[Proof of Lemma \ref{lem: dual ODEs Heston Model MK23}]
	Follows immediately from Lemma 4.5 in \citet{EKZ2023} with $d=m=1$ and $\Sigma_1 = 1.$
\end{proof}
%%%%%%%%%%%%%%
\vspace{0.25cm}
%%%%%%%%%%%%%%
\begin{proof}[Proof of Lemma \ref{lem: optimal portfolio Heston MK23}]
	Noting that $G_v(t,v,z) = \frac{b}{v}G(t,v,z)$ and $G_z(t,v,z) = B(T-t)G(t,v,z)$, this result follows immediately from Theorem 3.12 and Lemma 4.5 in \citet{EKZ2023} with $d=m=1$ and $\Sigma_1 = 1.$
\end{proof}
%%%%%%%%%%%%
\vspace{0.25cm}
%%%%%%%%%%%%
\begin{proof}[Proof of Lemma \ref{lem: ODE B for Box constraints}]
	The main objective of this proof is to determine the minimising argument $\la(B)$ for the minimisation problem
	$$\inf_{\lambda \in \R}\underbrace{\left( 2(1-b)\delta_K(\lambda) + \left(\eta + \lambda + \sigma \rho B \right)^2\right)}_{=: D(\lambda, B) }= \inf_{\lambda \in \R}D(\lambda, B)$$
	for any fixed value $B \in \R$. The remaining statement follows immediately after substituting the obtained minimiser $\la(B) = \la(B(\tau))$ in (\ref{eq:  ODE B constr. Heston}). \\
	
	Given any constraints $K=[\alpha, \beta]$ with $-\infty \leq \alpha \leq \beta \leq \infty$ (see (\ref{eq: 1-d convex constraints as interval})), we have for any $\lambda \in \R$
	$$\delta_K(\lambda) = -\inf_{\alpha \leq x \leq \beta}(x \lambda) = - \alpha \lambda \1_{\{ \lambda > 0 \}} - \beta \lambda \1_{\{ \lambda < 0 \}}.$$
	
	Thus, taking the derivative of $D(\cdot, B)$ on $(-\infty,0)$ and $(0,\infty)$ yields
	$$\frac{\partial}{\partial \lambda}D(\lambda, B) = \begin{cases}
		-2(1-b) \alpha + 2(\eta+\lambda + \sigma \rho B), &\quad \lambda > 0 \\
		-2(1-b) \beta + 2(\eta+\lambda + \sigma \rho B), &\quad \lambda < 0.
	\end{cases}$$
	Then, by the first-order optimality condition and the fact that $D(\cdot, B)$ is quadratic on $[0,\infty),$ 
	\begin{align*}
		\lambda^{-}(B) &= \big( (1-b)\alpha - (\eta+\sigma \rho B ) \big) \1_{\{ (1-b)\alpha - (\eta+\sigma \rho B )  > 0\}} =   \big( (1-b)\alpha - (\eta+\sigma \rho B ) \big) \1_{\{ \rho B <B_{-}  \}}
	\end{align*}
	minimises $D(\cdot, B)$ on $[0,\infty),$ and by the same argument 
	\begin{align*}
		\lambda^{+}(B) &=  \big( (1-b)\beta - (\eta+\sigma \rho B ) \big) \1_{\{  (1-b)\beta - (\eta+\sigma \rho B ) < 0\}} = \big( (1-b)\beta - (\eta+\sigma \rho B ) \big)\1_{ \{ \rho B > B_{+} \}}
	\end{align*}
	minimises $D(\cdot, B)$ on $(-\infty, 0]$. Therefore, noting that by construction 
	\begin{align}\label{eq: minimiser on both zones are better than lambda = 0}
		\max\big(D(\lambda^{-}(B),B), \ D(\lambda^{+}(B), B)\big) & =  \max\Big( \underbrace{\min_{\lambda \in (-\infty, 0]}D(\lambda,B)}_{\leq D(0,B)}, \underbrace{\min_{\lambda \in [0,\infty)}D(\lambda,B)}_{\leq D(0,B)}\Big) \nonumber \\
		& \quad \leq D(0,B) = (\eta+\sigma \rho B)^2 
	\end{align}
	and $B_{-}\leq B_{+}$, we finally derive 
	\begin{align*}
		\inf&_{\lambda \in \R}D(\lambda, B) \\
		&= \min\big(D(\lambda^{-}(B),B), \ D(\lambda^{+}(B), B)\big) \\
		&= \min \Big( \big[ 2(1-b)(-\alpha )\big((1-b)\alpha - \eta - \sigma \rho B \big)+\big((1-b)\alpha\big)^2\big]\1_{\{\rho B < B_{-}\}}+ \big[\eta + \sigma\rho B \big]^2 \1_{\{B_{-} \leq \rho B \}}, \\
		& \hspace{1.5cm} \big[ 2(1-b)(-\beta )\big((1-b)\beta  - \eta - \sigma \rho B \big)+\big((1-b)\beta\big)^2\big]\1_{\{ B_{+}<\rho B\}}+ \big[\eta + \sigma \rho B \big]^2\1_{\{\rho B\leq B_{+}\}} \Big) \\
		&\overset{(\ref{eq: minimiser on both zones are better than lambda = 0})}{=} \big[ 2(1-b)(-\alpha )\big((1-b)\alpha - \eta - \sigma \rho B \big)+\big((1-b)\alpha\big)^2\big] \1_{\{B_{-} \leq \rho B \}} +  \big[\eta + \sigma \rho B \big]^2 \1_{\{B_{-} \leq \rho B \leq B_{+}\}} \\
		&\qquad \quad + \big[ 2(1-b)(-\beta )\big((1-b)\beta  - \eta - \sigma \rho B \big)+\big((1-b)\beta\big)^2\big]\1_{\{ B_{+}<\rho B\}}\\
		&= D(\lambda^{-}(B), B)\1_{\{B_{-} \leq \rho B \}} + D(0, B)\1_{\{B_{-} \leq \rho B \leq B_{+}\}} + D(\lambda^{+}(B), B)\1_{\{ B_{+}<\rho B\}}.
	\end{align*}
	Thus, the minimiser $\la(B)$ is given as
	\begin{align*}
		\la(B) &=\lambda^{-}(B) \1_{\{\rho B < B_{-}\}} + \lambda^{+}(B) \1_{\{ B_{+} <\rho B\}} \\
		&=\big( (1-b)\alpha - (\eta+\sigma \rho B ) \big) \1_{\{\rho B < B_{-}\}} + \big( (1-b)\beta - (\eta+\sigma \rho B ) \big)\lambda^{+}(B) \1_{\{ B_{+} <\rho B\}}.
	\end{align*}
	Substituting $\la(B(\tau))$ in (\ref{eq:  ODE B constr. Heston}) and factoring $B(\tau)$ and $(B(\tau))^2$ concludes the proof.
\end{proof}
%%%%%%%%%%%%
\vspace{0.25cm}
%%%%%%%%%%%%
\begin{proof}[Proof of Theorem \ref{thm: solution constrained B in 1-D}]
	$B$ is specifically constructed in such a way that it is continuous and satisfies the corresponding Riccati ODE of (\ref{eq: ODE zones}), whenever $\rho B(\tau)$ is within each of the zones $Z_{-},$ $Z_0,$ or $Z_{+}$. Moreover, as the right hand side of the ODE (\ref{eq:  ODE B constr. Heston}) is continuous, so are the derivatives of the constructed $B$. Thus, $B$ is a solution to (\ref{eq: ODE zones}) and thereby (\ref{eq:  ODE B constr. Heston}).
\end{proof}
%%%%%%%%%%%%
\vspace{0.25cm}
%%%%%%%%%%%%
\begin{proof}[Proof of Lemma \ref{lem: A1}]
	By rewriting $\pi^{\ast}$ in terms of $\lambda^{\ast}$ and $B$, we see that 
	$$\frac{b \rho}{\sigma}\pi^{\ast}(t) + B(T-t) = \begin{cases}
		\frac{b\rho}{\sigma}\alpha + B(T-t), &\  \text{if } \rho B(T-t)< B_{-} \\
		\frac{b}{1-b}\frac{\rho}{\sigma}\eta +\left(1+\frac{b}{1-b}\rho^2\right)B(T-t),&\  \text{if } B_{-}\leq \rho B(T-t) \leq B_{+}\\
		\frac{b\rho}{\sigma}\beta + B(T-t), &\  \text{if }  B_{+}< \rho B(T-t)\\
	\end{cases}$$ 
	is non-decreasing in $B(T-t).$ Hence, as $B$ is itself monotonous (see Remark \ref{rem: monotonicty of B}), we have 
	\begin{align}\label{eq: two point maximisation}
		\sup_{t\in[0,T]}\left(\frac{b \rho}{\sigma}\pi^{\ast}(t) + B(T-t)\right) = \max \left( \frac{b \rho}{\sigma}\pi^{\ast}(T) + B(0),  \frac{b \rho}{\sigma}\pi^{\ast}(0) + B(T)\right).
	\end{align}
	Hence, when first assuming that the maximum in (\ref{eq: two point maximisation}) is attained at $t = T$, we obtain
	\begin{align*}
		\frac{b \rho}{\sigma}\underbrace{\pi^{\ast}(T)}_{\in [\alpha, \beta]} + \underbrace{B(0)}_{= 0} \leq \max \left(\frac{b \rho}{\sigma} \alpha, \frac{b \rho}{\sigma}\beta\right) \overset{(\ref{eq: boundedness assumption on constraints at t=T})}{\leq} \frac{\kappa}{\sigma^2}.
	\end{align*}
	On the other hand, if we assume that the maximum in (\ref{eq: two point maximisation}) is attained at $t=0$, we need to compare the different terminal values $B(T)$ based on the coefficients of the corresponding Riccati ODE at $T$. Due to the monotonicity of $B$ and the assumption that the maximum in (\ref{eq: two point maximisation}) is attained at $t=0$, it suffices if we can find a bound for $T\rightarrow \infty$. 
	\begin{itemize}
		\item Case $\rho B(T) \in Z_+$: \\
		Then, $\pi^{\ast}(0) = \beta$ and $B(T)$ is given as the solution (\ref{eq: solution Riccati ODE Filipovic}) to the Riccati ODE (\ref{eq: Riccati Filipovic ODE}) with coefficients $r^+_0,$ $r^+_1,$ and $r^+_2$. In particular, we have
		\begin{align*}
			\frac{b \rho}{\sigma}&\underbrace{\pi^{\ast}(0)}_{= \beta} + B(T)  \leq 	\frac{b \rho}{\sigma}\beta + \lim_{T\rightarrow \infty}B(T) \\
			&= \frac{b \rho}{\sigma}\beta + \lim_{T\rightarrow \infty}\left(\frac{2r^+_2r^+_3B_0+\left(e^{r^+_3T}-1\right)(r^+_1+r^+_3)\left(r^+_1+r^+_2B_0-r^+_3\right)}{2r^+_2r^+_3-r^+_2\left(e^{r^+_3T}-1\right)\left(r^+_1+r^+_2B_0-r^+_3\right)}\right) \\
			&= \frac{b \rho}{\sigma}\beta - \frac{r^+_1 + r^+_3}{r^+_2} = \frac{b \rho}{\sigma}\beta - \frac{b\sigma\rho \beta - \kappa}{\sigma^2} - \frac{r^+_3}{\sigma^2} = \frac{\kappa}{\sigma^2} - \underbrace{\frac{r^+_3}{\sigma^2}}_{\geq 0} \leq \frac{\kappa}{\sigma^2}.
		\end{align*}
		\item Case $\rho B(T) \in Z_{-}$: \\
		Analogous arguments as in the case $\rho B(T) \in Z_{+}$ yield that
		\begin{align*}
			\frac{b \rho}{\sigma}\pi^{\ast}(0)+ B(T) \leq \frac{\kappa}{\sigma^2}.
		\end{align*}
		\item Case $\rho B(T) \in Z_{0}$:
		Then, $\pi^{\ast}(0) = \frac{1}{1-b}\left(\eta + \sigma \rho B(T)\right)$ and $B(T)$ is given as the solution (\ref{eq: solution Riccati ODE Filipovic}) to Riccati ODE (\ref{eq: Riccati Filipovic ODE}) with coefficients $r_0,$ $r_1,$ and $r_2$. In particular, we have
		\begin{align*}
			\frac{b \rho}{\sigma}&\pi^{\ast}(0)+ B(T) =  	\frac{\rho}{\sigma}\frac{b}{1-b}\left(\eta + \sigma \rho B(T)\right) + B(T) = \frac{\rho}{\sigma}\frac{b}{1-b}\eta + \left(1+\frac{b}{1-b}\rho^2\right)B(T) \\
			& \leq \frac{\rho}{\sigma}\frac{b}{1-b}\eta + \left(1+\frac{b}{1-b}\rho^2\right)\underbrace{\lim_{T\rightarrow \infty}B(T)}_{=-\frac{r_1+r_3}{r_2}} = \frac{\rho}{\sigma}\frac{b}{1-b}\eta -\left(1+\frac{b}{1-b}\rho^2\right)\frac{r_1+r_3}{r_2}\\
			&= \frac{\rho}{\sigma}\frac{b}{1-b}\eta -\left(1+\frac{b}{1-b}\rho^2\right)\frac{r_1+r_3}{\sigma^2\left(1+\frac{b}{1-b}\rho^2\right)} = \frac{\rho}{\sigma}\frac{b}{1-b}\eta - \frac{r_1+r_3}{\sigma^2} \\
			&= \frac{\rho}{\sigma}\frac{b}{1-b}\eta - \frac{1}{\sigma^2}\left(\frac{b}{1-b}\eta \sigma \rho - \kappa \right) - \frac{r_3}{\sigma^2} = \frac{\kappa}{\sigma^2}- \underbrace{\frac{r_3}{\sigma^2}}_{\geq 0}\leq \frac{\kappa}{\sigma^2}.
		\end{align*}
	\end{itemize}
\end{proof}
%%%%%%%%%%%%
\vspace{0.25cm}
%%%%%%%%%%%%
\begin{proof}[Proof of Lemma \ref{lem: A2}]
	Note that we may express ODE (\ref{eq:  ODE B constr. Heston}) in terms of $\lambda^{\ast}$ and $\pi^{\ast}$ as
	\begin{align*}
		B'(\tau) = - \kappa B(\tau) + \frac{1}{2}\sigma^2B(\tau)^2 + \frac{1}{2}\frac{b}{1-b}\left(2(1-b)\delta_K\left(\lambda^{\ast}\left(B(\tau)\right)\right) + (1-b)^2\left(\pi^{\ast}(T-\tau)\right)^2\right).
	\end{align*}
	
	We distinguish between three cases, depending on whether the allocation constraint is active or not.
	\begin{itemize}
		\item Case $\pi^{\ast}(t)\in (\alpha, \beta)$: \\
		Then, $\la\left(B(T-t)\right)= \left(\la\right)'\left(B(T-t)\right) = \delta_K\left(\la\left(B(T-t)\right)\right) = 0$. Hence,
		\begin{align*}
			&\frac{1}{2}\frac{b}{1-b}\eta^2-\frac{1}{2}\frac{b}{1-b}\left(\la\left(B(T-t)\right) + \sigma \rho B(T-t) \right)^2 - \frac{1}{2}b^2\rho^2\left(\pi^{\ast}(t)\right)^2  \\
			& \quad + b \frac{\rho\kappa}{\sigma} \pi^{\ast}(t) + \frac{b}{1-b}\frac{\rho}{\sigma}\left[\left(\la\right)'\left(B(T-t)\right)+\sigma \rho \right]B'(T-t) \\
			= & \frac{1}{2}\frac{b}{1-b}\eta^2 - \frac{1}{2}\frac{b}{1-b}\left(\sigma \rho B(T-t) \right)^2 - \frac{1}{2}b^2\rho^2\left(\pi^{\ast}(t)\right)^2 \\
			& \quad + b \frac{\rho\kappa}{\sigma} \pi^{\ast}(t) + \frac{b}{1-b} \rho^2\left[-\kappa B(T-t)+ \frac{1}{2}\sigma^2\left(B(T-t)\right)^2 + \frac{1}{2}b(1-b)\left(\pi^{\ast}(t)\right)^2\right] \\
			= & \frac{1}{2}\frac{b}{1-b}\eta^2 -  \frac{1}{2}\frac{b}{1-b}\left(\sigma \rho B(T-t) \right)^2 + \frac{b}{1-b}\frac{\rho \kappa}{\sigma}\left(\eta + \sigma \rho B(T-t)\right) \\
			& \quad + \rho^2\frac{b}{1-b}\left[-\kappa B(T-t) +\frac{1}{2}\sigma^2\left(B(T-t)\right)^2\right] \\
			=& \frac{b}{1-b}\eta\left(\frac{\eta}{2} + \frac{\rho \kappa}{\sigma}\right)
		\end{align*}

		\item Case $\pi^{\ast}(t)= \alpha$: \\
		Then, $\la\left(B(T-t)\right) = (1-b)\alpha -\eta - \sigma \rho B(T-t)$ and
		$$
		\underbrace{\left[\left(\la\right)'\left(B(T-t)\right)+\sigma \rho\right]}_{= 0} B'(T-t) = 0.
		$$
		Hence, 
		\begin{align*}
			&\frac{1}{2}\frac{b}{1-b}\eta^2-\frac{1}{2}\frac{b}{1-b}\overbrace{\left(\la\left(B(T-t)\right) + \sigma \rho B(T-t) \right)^2}^{= \left((1-b)\alpha - \eta\right)^2} - \frac{1}{2}b^2\rho^2\underbrace{\left(\pi^{\ast}(t)\right)^2}_{=\alpha^2}  \\
			& \quad + b \frac{\rho\kappa}{\sigma} \underbrace{\pi^{\ast}(t)}_{= \alpha} + \frac{b}{1-b}\frac{\rho}{\sigma}\underbrace{\left[\left(\la\right)'\left(B(T-t)\right)+\sigma \rho\right]B'(T-t)}_{= 0} \\
			= & \frac{1}{2}\frac{b}{1-b}\eta^2-\frac{1}{2}\frac{b}{1-b}\left((1-b)\alpha - \eta\right)^2 - \frac{1}{2}b^2\rho^2\alpha^2 + b \frac{\rho\kappa}{\sigma}\alpha \\
			= & \frac{1}{2}\frac{b}{1-b}\eta^2-\frac{1}{2}\frac{b}{1-b}\left((1-b)^2\alpha^2-2(1-b)\alpha\eta + \eta^2\right) - \frac{1}{2}b^2\rho^2\alpha^2 + b \frac{\rho\kappa}{\sigma}\alpha \\
			=& -\frac{1}{2}\frac{b}{1-b}\left((1-b)^2\alpha^2-2(1-b)\alpha\eta \right) - \frac{1}{2}b^2\rho^2\alpha^2 + b \frac{\rho\kappa}{\sigma}\alpha \\
			=&-\frac{1}{2}b(1-b)\alpha^2 + b\alpha \eta - \frac{1}{2}b^2\rho^2\alpha^2+b\frac{\rho \kappa}{\sigma} \alpha \\
			=&b\alpha \left[\eta - \frac{1}{2}\alpha + \frac{\rho\kappa}{\sigma}+\frac{1}{2}\alpha b (1-\rho^2)\right].
		\end{align*}
		
		\item Case $\pi^{\ast}(t)= \beta$: \\
		Following the same arguments as in the case $\pi^{\ast}(t)= \alpha$ yields
		\begin{align*}
			&\frac{1}{2}\frac{b}{1-b}\eta^2-\frac{1}{2}\frac{b}{1-b}\left(\la\left(B(T-t)\right) + \sigma \rho B(T-t) \right)^2 - \frac{1}{2}b^2\rho^2\left(\pi^{\ast}(t)\right)^2  \\
			& \quad + b \frac{\rho\kappa}{\sigma} \pi^{\ast}(t) + \frac{b}{1-b}\frac{\rho}{\sigma}\left[\left(\la\right)'\left(B(T-t)\right)+\sigma \rho\right] B'(T-t)\\
			= &b\beta \left[\eta - \frac{1}{2}\beta + \frac{\rho\kappa}{\sigma}+\frac{1}{2}\beta b (1-\rho^2)\right].
		\end{align*}
	\end{itemize}
	Combining all three cases with Assumption \ref{ass: condition box constraints} finally yields the claim.	
\end{proof}
%%%%%%%%%%%%
\vspace{0.25cm}
%%%%%%%%%%%%
\begin{proof}[Proof of Theorem \ref{thm: strong verification in Heston market}]
	Let B be as in Theorem \ref{thm: solution constrained B in 1-D},
	$$
	A(\tau) = rb\tau + \kappa \theta \int_0^{\tau}B(s)ds.
	$$
	and define $G:[0,T]\times (0,\infty)^2\rightarrow \R$ as $G(t,v,z) = \frac{1}{b}v^b \exp(A(T-t)+B(T-t)z)$. Then, $A$ and $B$ are solutions to the ODEs (\ref{eq:  ODE A constr. Heston}) and (\ref{eq:  ODE B constr. Heston}) and thus continuous and bounded on $[0,T].$ Moreover, according to Lemma \ref{lem: dual ODEs Heston Model MK23}, G is a solution to the HJB PDE (\ref{eq: constr. HJB PDE Heston}). Hence, as per Lemma \ref{lem: optimal portfolio Heston MK23}, it only remains to verify that the sequence 
	$$
	\big(G(\tau^{0}_{n,t},V^{v_0, \pi^{\ast}}(\tau^{0}_{n,t}),z(\tau^{0}_{n,t}))\big)_{n\in \N}
	$$
	is uniformly integrable for all $t\in [0,T]$. Let $\mathcal{T}$ denote the set of all $\mathbb{F}$ stopping times taking values in $[0,T]$. We verify the stronger statement that
	$$
	\big(G(\tau,V^{v_0, \pi^{\ast}}(\tau),z(\tau))\big)_{\tau \in \mathcal{T}}
	$$
	is bounded in $L^q$ for some $q>1$ (see e.g. Theorem 4.6.2 in \citet{Durrett2019}). The argument is an adaptation of the proof of Theorem 5.3 in \citet{Kraft2005}. Consider arbitrary $t\in[0,T]$ and $q = 1+\epsilon$ with $\epsilon >0$ and define
	$$
	D(t) = bq\sqrt{1-\rho^2}\int_0^t \pi^{\ast}(s)\sqrt{z(s)}d\hat{W}(s).
	$$ Then, by defining the  deterministic, continuous function 
	$$det_1(t) = q\left(-\ln(|b|)+ b \ln(v_0) + brt  + A(T-t)\right),$$
	we obtain
	\begin{align}
		q \ln &\left( \left | G\left(t,V^{v_0, \pi^{\ast}}(t),z(t)\right)\right | \right) \nonumber\\
		&\hspace{0.25cm}= q\left(-\ln(|b|)+ b \ln\left(V^{v_0, \pi^{\ast}}(t)\right) + A(T-t) + B(T-t)z(t)\right) \nonumber \\
		&\hspace{0.25cm}= q\Bigg(-\ln(|b|)+ b \ln(v_0) + b \int_0^t\big( r+\eta\pi^{\ast}(s)z(s)-\frac{1}{2}\left(\pi^{\ast}(s)\right)^2z(s)\big)ds \nonumber\\
		& \hspace{0.25cm} \qquad  + b\int_0^t\pi^{\ast}(s)\sqrt{z(s)}\underbrace{dW(s)}_{\mathclap{=\rho dW^z(s)+\sqrt{1-\rho^2}d\hat{W}(s)}} + A(T-t) + B(T-t)z(t)\Bigg) \nonumber \\
		&\hspace{0.25cm} = det_1(t) + D(t) - \frac{1}{2}\langle D \rangle_t + \frac{1}{2}q^2b^2(1-\rho^2)\int_0^t \left(\pi^{\ast}(s)\right)^2z(s)ds + qB(T-t)z(t) \nonumber  \\
		& \hspace{0.25cm} \quad + qb\rho \int_0^t \pi^{\ast}(s)\sqrt{z(s)}dW^z(s) + qb \int_0^t\left(\pi^{\ast}(s)\eta-\frac{1}{2}\left(\pi^{\ast}(s)\right)^2\right)z(s)ds \nonumber \\
		&\hspace{0.25cm}= det_1(t) + D(t) - \frac{1}{2}\langle D \rangle_t  + qB(T-t)z(t) + qb\rho \int_0^t \pi^{\ast}(s)\sqrt{z(s)}dW^z(s) \nonumber \\
		& \hspace{0.25cm} \quad + q\int_0^t z(s)\left\{ b\left(\pi^{\ast}(s) \eta - \frac{1}{2}\left(\pi^{\ast}(s)\right)^2\right) + \frac{1}{2}qb^2\left(1-\rho^2\right)\left(\pi^{\ast}(s)\right)^2\right\}ds \nonumber \\
		&\overset{q = 1 +\epsilon}{=} det_1(t) + D(t) - \frac{1}{2}\langle D \rangle_t  + qB(T-t)z(t) + qb\rho \int_0^t \pi^{\ast}(s)\sqrt{z(s)}dW^z(s) \nonumber \\
		& \hspace{0.25cm} \quad + q\int_0^t z(s)\Bigg\{ \frac{1}{2}\frac{b}{1-b}\eta^2-\frac{1}{2}b(1-b)\left(\pi^{\ast}(s)-\frac{\eta}{1-b}\right)^2 \\
		&\hspace{2.75cm}- \frac{1}{2}b^2\rho^2\left(\pi^{\ast}(s)\right)^2 + \frac{1}{2}\epsilon b^2(1-\rho^2)\left(\pi^{\ast}(s)\right)^2\Bigg\}ds \nonumber \\
		&\hspace{0.25cm}\overset{(\ref{eq: optimal portfolio Heston implicit})}{=}  det_1(t) + D(t) - \frac{1}{2}\langle D \rangle_t  + qB(T-t)z(t) + qb\rho \int_0^t \pi^{\ast}(s)\sqrt{z(s)}dW^z(s)\nonumber \\
		&\hspace{0.25cm} \quad + q\int_0^t z(s)\Bigg \{ \frac{1}{2}\frac{b}{1-b}\eta^2-\frac{1}{2}\frac{b}{1-b}\left(\la\left(B(T-s)\right) + \sigma \rho B(T-s) \right)^2 \nonumber \\
		& \hspace{2.75cm} - \frac{1}{2}b^2\rho^2\left(\pi^{\ast}(s)\right)^2 + \frac{1}{2}\epsilon b^2(1-\rho^2)\left(\pi^{\ast}(s)\right)^2\Bigg \}ds \label{eq: 1-D-Verification Theorem eq 1} 
	\end{align} 
	$B$ is continuously differentiable and monotonous, since it is the solution to an autonomous ODE (see Remark \ref{rem: monotonicty of B}). Hence, $\pi^{\ast}$ is monotonous and differentiable and has finite variation. Therefore, It\^{o}'s product rule yields
	\begin{align*}
		d\left(\pi^{\ast}(t)z(t)\right) &= \pi^{\ast}(t)dz(t) + z(t)d\pi^{\ast}(t) \\
		&= \pi^{\ast}(t)dz(t) - \frac{z(t)}{1-b}\left[\left(\la\right)'\left(B(T-t)\right)+\sigma \rho \right]B'(T-t)dt \\
		&= \pi^{\ast}(t)\kappa\left(\theta-z(t)\right)dt + \pi^{\ast}(t)\sigma \sqrt{z(t)}dW^z(t) \\
		& \quad -\frac{z(t)}{1-b}\left[\left(\la\right)'\left(B(T-t)\right)+\sigma \rho \right]B'(T-t)dt.
	\end{align*}
	Hence,
	\begin{align}
		\int_0^t  \pi^{\ast}(s)\sqrt{z(s)}dW^z(s) &= \frac{1}{\sigma}\left( \pi^{\ast}(t)z(t)- \pi^{\ast}(0)z(0)\right)- \frac{\kappa}{\sigma}\int_0^t\pi^{\ast}(s)\left(\theta-z(s)\right)ds  \nonumber  \\
		& \quad + \frac{1}{\sigma}\frac{1}{1-b}\int_0^t  \left[\left(\la\right)'\left(B(T-s)\right)+ \sigma \rho\right]B'(T-s)z(s)ds \label{eq: 1-D-Verification Theorem auxiliary eq 1}
	\end{align}
	Substituting (\ref{eq: 1-D-Verification Theorem auxiliary eq 1}) in (\ref{eq: 1-D-Verification Theorem eq 1}), while defining the deterministic and continuous function
	$$det_2(t)=det_1(t) - \frac{qb\rho}{\sigma}\left[\pi^{\ast}(0)z(0)+\kappa \int_0^t\pi^{\ast}(s)\theta ds\right],$$ 
	yields
	\begin{align}
		q \ln &\left(\left | G\left(t,V^{v_0, \pi^{\ast}}(t),z(t)\right)\right | \right) \nonumber\\	
		&=  det_2(t) + D(t) - \frac{1}{2}\langle D \rangle_t  + qz(t)\left[\frac{b\rho}{\sigma}\pi^{\ast}(t) + B(T-t)\right] \nonumber \\
		&\quad + q\int_0^t z(s)\Bigg \{ \frac{1}{2}\frac{b}{1-b}\eta^2-\frac{1}{2}\frac{b}{1-b}\left(\la\left(B(T-s)\right) + \sigma \rho B(T-s) \right)^2 \nonumber\\
		& \hspace{2.5cm} - \frac{1}{2}b^2\rho^2\left(\pi^{\ast}(s)\right)^2  + \frac{1}{2} \epsilon b^2(1-\rho^2)\left(\pi^{\ast}(s)\right)^2 \nonumber \\
		&\hspace{2.5cm} + b \frac{\rho\kappa}{\sigma} \pi^{\ast}(s) + \frac{b}{1-b}\frac{\rho}{\sigma}\left[\left(\la\right)'\left(B(T-s)\right)+\sigma \rho \right]B'(T-s)\Bigg \}ds \label{eq: 1-D-Verification Theorem eq 2} 
	\end{align}
	Combining Lemma \ref{lem: A1} with (\ref{eq: 1-D-Verification Theorem eq 2}) yields
	\begin{align}
		q \ln &\left( \left | G\left(t,V^{v_0, \pi^{\ast}}(t),z(t)\right)\right | \right) \nonumber \\	
		&\leq  det_2(t) + D(t) - \frac{1}{2}\langle D \rangle_t  + q\frac{\kappa}{\sigma^2}z(t)  \nonumber  \\
		&\quad + q\int_0^t z(s)\Bigg \{ \frac{1}{2}\frac{b}{1-b}\eta^2-\frac{1}{2}\frac{b}{1-b}\left(\la\left(B(T-s)\right) + \sigma \rho B(T-s) \right)^2  \nonumber \\
		& \hspace{2.5cm} - \frac{1}{2}b^2\rho^2\left(\pi^{\ast}(s)\right)^2  + \frac{1}{2}\epsilon b^2(1-\rho^2)\left(\pi^{\ast}(s)\right)^2 \nonumber \\
		&\hspace{2.5cm} + b \frac{\rho\kappa}{\sigma} \pi^{\ast}(s) + \frac{b}{1-b}\frac{\rho}{\sigma}\left[\left(\la\right)'\left(B(T-s)\right)+\sigma \rho \right]B'(T-s)\Bigg \}ds \nonumber  \\
		&= det_2(t) + D(t) - \frac{1}{2}\langle D \rangle_t  + q\frac{\kappa}{\sigma^2}\left(z(0)+\int_0^t\kappa(\theta-z(s))ds + \sigma\int_0^t\sqrt{z}(s)dW^z(s)\right) \nonumber  \\
		&\quad + q\int_0^t z(s)\Bigg \{ \frac{1}{2}\frac{b}{1-b}\eta^2-\frac{1}{2}\frac{b}{1-b}\left(\la\left(B(T-s)\right) + \sigma \rho B(T-s) \right)^2 \nonumber \\
		& \hspace{2.5cm} - \frac{1}{2}b^2\rho^2\left(\pi^{\ast}(s)\right)^2  + \frac{1}{2}\epsilon b^2(1-\rho^2)\left(\pi^{\ast}(s)\right)^2 \nonumber \\
		&\hspace{2.5cm} + b \frac{\rho\kappa}{\sigma} \pi^{\ast}(s) + \frac{b}{1-b}\frac{\rho}{\sigma}\left[\left(\la\right)'\left(B(T-s)\right)+\sigma \rho \right]B'(T-s)\Bigg \}ds \nonumber \\
		&= \underbrace{det_2(t) + q \frac{\kappa}{\sigma^2}\left(z(0)+\int_0^t\kappa \theta ds\right)}_{=: det_3(t)}+ D(t) - \frac{1}{2}\langle D \rangle_t + \underbrace{\frac{q\kappa}{\sigma}\int_0^t\sqrt{z(s)}dW^z(s)}_{=: M(t)} \nonumber \\
		&\quad + q\int_0^t z(s)\Bigg \{ -\frac{\kappa^2}{\sigma^2} + \frac{1}{2}\frac{b}{1-b}\eta^2-\frac{1}{2}\frac{b}{1-b}\left(\la\left(B(T-s)\right) + \sigma \rho B(T-s) \right)^2 \nonumber  \\
		& \hspace{2.5cm} - \frac{1}{2}b^2\rho^2\left(\pi^{\ast}(s)\right)^2  + \frac{1}{2}\epsilon b^2(1-\rho^2)\left(\pi^{\ast}(s)\right)^2 \nonumber \\
		&\hspace{2.5cm} + b \frac{\rho\kappa}{\sigma} \pi^{\ast}(s) + \frac{b}{1-b}\frac{\rho}{\sigma}\left[\left(\la\right)'\left(B(T-s)\right)+\sigma \rho \right]B'(T-s)\Bigg \}ds \nonumber  \\
		&= det_3(t)+ D(t) - \frac{1}{2}\langle D \rangle_t + M(t) - \frac{1}{2}\langle M \rangle_t \nonumber \\
		&\quad + q\int_0^t z(s)\Bigg \{ -\frac{\kappa^2}{\sigma^2}+\frac{1}{2}\underbrace{q}_{=1+\epsilon}\frac{\kappa^2}{\sigma^2} + \frac{1}{2}\frac{b}{1-b}\eta^2-\frac{1}{2}\frac{b}{1-b}\left(\la\left(B(T-s)\right) + \sigma \rho B(T-s) \right)^2 \nonumber \\
		& \hspace{2.5cm} - \frac{1}{2}b^2\rho^2\left(\pi^{\ast}(s)\right)^2  + \frac{1}{2}\epsilon b^2(1-\rho^2)\left(\pi^{\ast}(s)\right)^2 \nonumber \\
		&\hspace{2.5cm} + b \frac{\rho\kappa}{\sigma} \pi^{\ast}(s) + \frac{b}{1-b}\frac{\rho}{\sigma}\left[\left(\la\right)'\left(B(T-s)\right)+\sigma \rho \right]B'(T-s)\Bigg \}ds \nonumber \\
		&=det_3(t)+ D(t) - \frac{1}{2}\langle D \rangle_t + M(t) - \frac{1}{2}\langle M \rangle_t \\
		&\quad + q\int_0^t z(s)\Bigg \{ -\frac{\kappa^2}{2\sigma^2} + \frac{1}{2}\frac{b}{1-b}\eta^2-\frac{1}{2}\frac{b}{1-b}\left(\la\left(B(T-s)\right) + \sigma \rho B(T-s) \right)^2 \nonumber \\
		& \hspace{2.5cm} - \frac{1}{2}b^2\rho^2\left(\pi^{\ast}(s)\right)^2  + \frac{1}{2}\epsilon b^2(1-\rho^2)\left(\pi^{\ast}(s)\right)^2+\frac{1}{2}\epsilon\frac{\kappa^2}{\sigma^2} \nonumber \\
		&\hspace{2.5cm} + b \frac{\rho\kappa}{\sigma} \pi^{\ast}(s) + \frac{b}{1-b}\frac{\rho}{\sigma}\left[\left(\la\right)'\left(B(T-s)\right)+\sigma \rho \right]B'(T-s)\Bigg \}ds \nonumber  \\
		&= det_3(t)+ D(t) - \frac{1}{2}\langle D \rangle_t + M(t) - \frac{1}{2}\langle M \rangle_t +q\int_0^tz(s)\Big \{ \left(  \ast\right) \Big \}ds \label{eq: 1-D-Verification Theorem eq 3} 
	\end{align}
	
	According to Lemma \ref{lem: A2},  the expression $(\ast)$ in (\ref{eq: 1-D-Verification Theorem eq 3}) is negative for all $s\in[0,T]$ if $\epsilon = 0$. Hence, due to the continuity of $(\ast)$ in $\epsilon$ and the boundedness of all deterministic functions in $(\ast)$ w.r.t. $t,$  there exists $\epsilon >0$ such that $(\ast)<0$ for all $s\in [0,T]$. For such a choice of $\epsilon$ we obtain
	\begin{align*}
		q \ln \left(\left | G\left(t,V^{v_0, \pi^{\ast}}(t),z(t)\right) \right |\right)  &\leq det_3(t)+ D(t) - \frac{1}{2}\langle D \rangle_t + M(t) - \frac{1}{2}\langle M \rangle_t \\
		& \leq \sup_{s\in[0,T]}\left(det_3(s)\right)+ D(t) - \frac{1}{2}\langle D \rangle_t + M(t) - \frac{1}{2}\langle M \rangle_t \\
		\Leftrightarrow \ \left |G\left(t,V^{v_0, \pi^{\ast}}(t),z(t)\right)\right |^q  &\leq \exp\left(\sup_{s\in[0,T]}\left(det_3(s)\right)\right)\cdot \underbrace{ \exp\left(D(t) - \frac{1}{2}\langle D \rangle_t\right) \cdot \exp\left(M(t) - \frac{1}{2}\langle M \rangle_t\right)}_{=: (\ast\ast)}.
	\end{align*}
	Since $D$ and $M$ are local martingales with independent diffusions, $ (\ast\ast)$ is a supermartingale. Hence, Doob's optional sampling theorem finally yields
	\begin{align*}
		&\sup_{\tau \in \mathcal{T}}\E\left[\left |G(\tau,V^{v_0, \pi^{\ast}}(\tau),z(\tau))\right |^q\right]  \\
		&\qquad \leq \exp\left(\sup_{s\in[0,T]}\left(det_3(s)\right)\right) \sup_{\tau \in \mathcal{T}}\E\left[\exp\left(D(t) - \frac{1}{2}\langle D \rangle_t\right) \cdot \exp\left(M(t) - \frac{1}{2}\langle M \rangle_t\right)\right] \\
		&\qquad \leq \exp\left(\sup_{s\in[0,T]}\left(det_3(s)\right)\right)\E\left[\exp\left(D(0) - \frac{1}{2}\langle D \rangle_0\right) \cdot \exp\left(M(0) - \frac{1}{2}\langle M \rangle_0\right)\right]\\
		&\qquad=\exp\left(\sup_{s\in[0,T]}\left(det_3(s)\right)\right)< \infty.
	\end{align*}
	Thus, $\left(G(\tau,V^{v_0, \pi^{\ast}}(\tau),z(\tau))\right)_{\tau \in \mathcal{T}}$ is bounded in $L^q$ for some $q>1$ and therefore uniformly integrable. Hence, $\pi^{\ast}$ is optimal for $\mathbf{(P)}$ as per Lemma \ref{lem: optimal portfolio Heston MK23}.
\end{proof}
%%%%%%%%%%%%%%
\vspace{0.25cm}
%%%%%%%%%%%%%%
\begin{proof}[Proof of Corollary \ref{cor: unconstrained Heston}]
	Noting that $\delta_K(\lambda)= \infty$ for all $x\neq 0$ if $K = \R$, the minimum in (\ref{eq: optimal lambda Heston}) is attained by $\la(B) = 0$ for all $B\in \R$ and the ODE (\ref{eq:  ODE B constr. Heston}) simplifies to
	\begin{align*}
		B_u'(\tau) &	= -r_0+r_1B_u(\tau) +\frac{1}{2} r_2 \big( B_u(\tau) \big)^2
	\end{align*}
	with coefficients 
	\begin{align*}
		r_0 = -\frac{b}{2(1-b)}\eta^2, \quad r_1 = \frac{b}{1-b}\eta\sigma\rho-\kappa, \quad r_2 = \sigma^2\left(1+\frac{b}{1-b}\rho^2\right).
	\end{align*}
	The expression for $\pi_u$ can be directly read off from Lemma \ref{lem: optimal portfolio Heston MK23}.
\end{proof}%%%%%%%%%%%%
\vspace{0.25cm}
%%%%%%%%%%%%
\begin{proof}[Proof of Lemma \ref{lem: constrained portfolio equal capped portfolio}]
	Let $\rho = 0.$ Then, according to Remark \ref{rem: optimal constrained pi is capped on ODE solution} and Corollary \ref{cor: unconstrained Heston}, both $\hat{\pi}^{\ast}(t)=\pi_M$ and $\pi_u(t)=\pi_M$ are constant in time $t\in [0,T].$ Thus,
	\begin{align*}
		\pi^{\ast}(t) = \text{Cap}(\hat{\pi}(t),\alpha,\beta) = \text{Cap}(\pi_M,\alpha,\beta)=\text{Cap}(\pi_u(t),\alpha,\beta) \quad \forall t\in [0,T].
	\end{align*}
	Consider now the case $\rho \neq 0$ and $\pi_M \in K = [\alpha,\beta].$ Then, $\alpha(1-b) \leq \eta \leq \beta(1-b)$ and thus $B_- \leq 0 \leq B_+.$ Following Lemma \ref{lem: ODE B for Box constraints} and Remark \ref{rem: optimal constrained pi is capped on ODE solution}, $0 \in Z_{0}$ and the Riccati ODE (\ref{eq:  ODE B constr. Heston}) has coefficients $r_0,$ $r_1,$ and $r_2$ at $\tau = 0.$ Hence, $B(\tau) = B_u(\tau)$ for all $\tau \in [0,T]$ such that $\rho B(\tau) \in Z_{0}.$ Moreover, following Remark \ref{rem: monotonicty of B}, $B_u(\tau)$ and $B(\tau)$ (respectively $\pi_u(t)$ and $\hat{\pi}^{\ast}(t)$) are solutions to autonomous ODEs and thus monotone functions in $\tau$ (respectively $t$). Noting that $\pi_u(t)\in K = [\alpha, \beta] \Leftrightarrow \rho B_u(T-t)\in Z_0,$ we know if $\pi_u(t) \in K,$ then
	$$
	\pi_u(t) = \frac{1}{1-b}\left(\eta + \sigma\rho B_u(T-t)\right) = \frac{1}{1-b}\left(\eta + \sigma\rho B(T-t)\right) = \hat{\pi}(t).
	$$
	Due to the monotonicity of $\pi_u$ and $\hat{\pi}^{\ast},$ we have $\pi_u(t) \notin (\alpha, \beta)$ if and only if $\hat{\pi}(t) \notin (\alpha, \beta).$ Thus, we have in total
	$$
	\pi^{\ast}(t) = \text{Cap}\left(\hat{\pi}^{\ast}(t), \alpha, \beta\right) = \text{Cap}\left(\pi_u(t), \alpha, \beta\right) \quad \forall t\in [0,T].
	$$
	In particular, in both cases $\rho = 0$ and $\pi_M\in K$ we have $\mathcal{P}^H_K = \mathcal{P}^{BS}_K.$
\end{proof}
%%%%%%%%%%%%%%
\vspace{0.25cm}
%%%%%%%%%%%%%%
\begin{proof}[Proof of Lemma \ref{lem: equivalence gap}]
	The implication (ii) $\Rightarrow$ (i) is trivial, because if (ii) holds, then $\hat{\pi}^{\ast}(t)\neq \pi_u(t)$ and either $\hat{\pi}^{\ast}(t) \in (\alpha, \beta)$ or $\pi_u(t) \in (\alpha, \beta).$ Hence, 
	$$
	\pi^{\ast}(t)=\text{Cap}\left(\hat{\pi}^{\ast}(t),\alpha, \beta\right)\neq\text{Cap}\left(\pi_u(t),\alpha, \beta\right),
	$$
	and therefore  $\mathcal{P}^H_K\left(\pi_u,\alpha, \beta\right) \neq \mathcal{P}^{BS}_K\left(\pi_u,\alpha, \beta\right).$ \\
	Assume now that (i) holds. According to Lemma \ref{lem: constrained portfolio equal capped portfolio}, this implies $\rho \neq 0$ and $\pi_M \notin K = [\alpha,\beta].$ Moreover, as $B$ and $B_u$ are continuously differentiable,
	$$
	\hat{\pi}^{\ast}(t) = \frac{1}{1-b}\left(\eta + \sigma \rho B(T-t)\right) \quad \text{and} \quad \pi_u(t) = \frac{1}{1-b}\left(\eta + \sigma \rho B_u(T-t)\right)
	$$
	are continuously differentiable functions with $\hat{\pi}^{\ast}(T) =  \pi_u(T) = \pi_M \notin K.$ Thus, (i) can only be satisfied if there exists $0<t<T$ such that 
	\begin{align}\label{eq: proof equivalence 1}
		\left \{ \hat{\pi}^{\ast}(t), \pi_u(t) \right \} \cap K \neq \emptyset,
	\end{align}
	as otherwise $\pi^{\ast}(t)=  \text{Cap}(\hat{\pi}^{\ast}(t),\alpha,\beta)=  \text{Cap}(\pi_M,\alpha,\beta) =   \text{Cap}(\pi_u(t),\alpha,\beta)$ for all $t\in [0,T].$ Thus, we define the latest such time point as 
	\begin{align*}
		\hat{t} = \sup \left \{ 0 \leq t \leq T \ \Big | \ 	\left \{ \hat{\pi}^{\ast}(t), \pi_u(t) \right \} \cap K \neq \emptyset \right \}.
	\end{align*}
	Since there exists a positive $t>0$ satisfying (\ref{eq: proof equivalence 1}), we know that $\hat{t}>0.$ Moreover, as $\pi_M \notin K$ and $K$ is closed, we also know that $\hat{t}<T.$ If $\hat{\pi}^{\ast}(\hat{t}) \neq \pi_u(\hat{t}),$ then only one of these portfolio processes takes a value in $K$ and (ii) holds for $\hat{t}$. \\
	It thus remains to verify that $\hat{\pi}^{\ast}(\hat{t}) \neq \pi_u(\hat{t})$ to conclude the proof. This will be proven by contradiction. Assume that $\hat{\pi}^{\ast}(\hat{t}) = \pi_u(\hat{t})\in K$ which is equivalent to $\rho B(T-\hat{t})=\rho B_u(T-\hat{t})\in Z_0.$ Further,
	\begin{align*}
		T-\hat{t} = \inf\left \{ 0 \leq \tau \leq T \ \Big | \ 	\left \{ \rho B(\tau), \rho B_u(\tau) \right \} \cap Z_0 \neq \emptyset \right \}.
	\end{align*}
	Moreover, as $\pi_M \notin K \overset{\rho \neq 0}{\Leftrightarrow} \rho B(0)=\rho B_u(0)=0\notin Z_0$, $B$ and $B_u$ are non-constant and thus strictly monotone functions. However, this implies that at $\tau = T-\hat{t},$ $\rho B(\tau)$ and $\rho B_u(\tau)$ have the same value at the boundary of $Z_0$ and satisfy the same ODE while taking values in $Z_0.$ Hence, $B(\tau) = B_u(\tau)$ for all $\tau\in [0,T]$ such that $B(\tau)\in Z_0.$ Noting that
	$$
	\rho B(\tau) \in Z_0 \Leftrightarrow \hat{\pi}^{\ast}(\tau)\in K \quad \text{and} \quad \rho B_u(\tau) \in Z_0 \Leftrightarrow \pi_u(\tau)\in K,
	$$
	finally yields for all $t\in [0,T]$:
	\begin{align*}
		\pi^{\ast}(t) &= \text{Cap}\left(\hat{\pi}^{\ast}(t),\alpha,\beta\right) = \text{Cap}\Big(\frac{1}{1-b}(\eta + \sigma\underbrace{\rho B(T-t)}_{\mathclap{=\rho B_u(T-t) \ \text{while }\hat{\pi}^{\ast}(t)\in K}}),\alpha, \beta \Big) \\
		&= \text{Cap}\left(\frac{1}{1-b}(\eta + \sigma \rho B_u(T-t)),\alpha,\beta\right) =  \text{Cap}\left(\pi_u(t),\alpha,\beta\right),
	\end{align*}
	which is a contradiction to (i). Thus, $\hat{\pi}^{\ast}(\hat{t}) \neq \pi_u(\hat{t}).$
\end{proof}

%%%%%%%%%%%%%%
\vspace{0.25cm}
%%%%%%%%%%%%%%
\begin{proof}[Proof of Corollary \ref{cor: stationary case}]
	\textcolor{white}{123}\\
	\underline{Proof of (i):} 
	Let $\pi_M < \alpha$ and $0<\alpha = 2 \pi_M < \alpha$$\frac{2}{1-b}\eta.$ Then, $\eta >0,$ $\pi_M >0$ and $\eta = (1-b)\pi_M < (1-b)\alpha \Leftrightarrow B_{-}>0.$ Following Lemma \ref{lem: ODE B for Box constraints} and Remark \ref{rem: optimal constrained pi is capped on ODE solution}, $0 \in Z_{-}$ and the Riccati ODE (\ref{eq:  ODE B constr. Heston}) has coefficients $r_0^-,$ $r_1^-,$ and $r_2^-$ at $\tau = 0.$ However, due to the initial condition $B(0)=0$, this implies
	\begin{align*}
		B'(0) = -r^{-}_0+r^{-}_1B(0) + \frac{1}{2}r^{-}_2B(0)^2 =  -r^{-}_0 = -\frac{1}{2}b\alpha\left((1-b)\alpha-2\eta\right) = 0,
	\end{align*}
	i.e., according to Remark \ref{rem: monotonicty of B}, $B$ is constant with $B(\tau) = 0$ for all $\tau \in [0,T].$ In particular, for all $t\in [0,T]$ we have $\hat{\pi}^{\ast}(t) = \pi_M$  and
	$$
	\pi^{\ast}(t) = \text{Cap}(\hat{\pi}^{\ast}(t),\alpha, \beta) = \text{Cap}(\pi_M,\alpha, \beta) = \alpha.
	$$
	On the other hand, the optimal unconstrained portfolio $\pi_u$ does not violate the constraint at time $t^{\ast}$, i.e.
	$$
	\text{Cap}(\pi_u(t^{\ast}),\alpha,\beta) = \pi_u(t^{\ast}) \in (\alpha, \beta),
	$$
	and thus  $\mathcal{P}^H_K\neq\mathcal{P}^{BS}_K.$ \\
	\underline{Proof of (ii):} According to Corollary \ref{cor: unconstrained Heston}, 
		$$
		\frac{\partial}{\partial t}\left(\pi_u(t)\right) = \frac{1}{1-b}\frac{\partial}{\partial t}\left(\eta + \sigma \rho B_u(T-t)\right)
		$$
		Since $B_u$ is the solution to the autonomous ODE (\ref{eq: Unconstr ODE B}), $B_u$ is a monotone function (cf.\ Remark \ref{rem: monotonicty of B}). Therefore,
		\begin{align*}
			\text{sign}\left(\frac{\partial}{\partial t}\left(\pi_u(t)\right)\right) & \hspace{1mm} = \hspace{1mm} \text{sign}\left( \frac{-\sigma \rho}{1-b} B_u'(T-t)\right)  = \text{sign}\left( \frac{-\sigma \rho}{1-b} B_u'(0)\right) \overset{(\ref{eq: Unconstr ODE B})}{=}  \text{sign}\left( \frac{\sigma \rho}{1-b}r_0 \right) \\
			&\overset{(\ref{eq: ODE zones})}{=} \text{sign}\left( \frac{\sigma \rho}{1-b}\left(-\frac{b}{2(1-b)}\eta^2\right) \right) \\
		\end{align*}
		As $\frac{\sigma\eta^2}{2(1-b)^2}$ is positive, this implies $\text{sign}\left(\frac{\partial}{\partial t}\left(\pi_u(t)\right)\right) = -\text{sign}\left(b \rho\right).$ As $\beta >0$ and $\pi_M > \beta \Leftrightarrow 0 = \rho B(0) > B_{+},$ we follow the same line of argument for $\hat{\pi}^{\ast}$ to obtain
		\begin{align*}
			\text{sign}\left(\frac{\partial}{\partial t}\left(\hat{\pi}^{\ast}(t)\right)\right) &= \text{sign}\left(\frac{-\sigma\rho}{1-b}B'(0)\right) = \text{sign}\left(\frac{\sigma\rho}{1-b}r_0^{+}\right) \\
			&= \text{sign}\left(\frac{\sigma\rho}{1-b}\frac{1}{2}b\beta\left([1-b]\beta - 2\eta \right)\right) = \text{sign}\left(\frac{\sigma\rho b \beta}{2}\left(\beta - 2\pi_M\right) \right)
		\end{align*}
		Disregarding the positive factor $\frac{\sigma \beta}{2}$ finally yields
		\begin{align*}
			\text{sign}\left(\frac{\partial}{\partial t}\left(\hat{\pi}^{\ast}(t)\right)\right) = \text{sign}\big(\rho b\underbrace{\left(\beta - 2\pi_M\right)}_{<0} \big) = -\text{sign}\left(b \rho\right) = 	\text{sign}\left(\frac{\partial}{\partial t}\left(\pi_u(t)\right)\right).
		\end{align*}
		If now $\rho<0$ and $b<0,$ then both portfolio allocation $\pi^{\ast}(t) = \text{Cap}\left(\hat{\pi}^{\ast}(t),\alpha,\beta\right)$ and $\text{Cap}\left(\pi_u(t),\alpha,\beta\right)$ are non-increasing in time. In particular, as both allocations are equal to $\beta < \pi_M$ at $t=T,$ they must be equal (to $\beta$) throughout the entire investment horizon. Hence, the projections $\mathcal{P}^{H}_K$ and $\mathcal{P}^{BS}_K$ coincide under these assumptions.
\end{proof}

%%%%%%%%%%%%%%
\vspace{0.25cm}
%%%%%%%%%%%%%%

\begin{proof}[Proof of Theorem \ref{thm: strong verification PCSV}]
	
		We transform the portfolio optimisation problem $\mathbf{(P)}$ into an equivalent optimisation problem by applying the change of control $\pi_A(t):= A'\pi(t)$ for any given $\pi \in \Lambda^{PCSV}.$ Expressing the SDE of the wealth process $V^{v_0,\pi}$ in terms of $\pi_A$ yields
		\begin{align}\label{eq: dynamics in terms of pi_A}
			dV^{v_0,\pi}(t)&= V^{v_0,\pi}(t)V^{v_0,\pi}(t)\left[\left(r+\underbrace{\eta'A}_{=: \eta_A'}\diag(z(t))\underbrace{A'\pi(t)}_{=\pi_A(t)}\right)dt + \underbrace{\pi(t)'A}_{= \pi_A(t)'}\diag(\sqrt{z(t)})dW(t)\right] \nonumber \\
			&= V^{v_0,\pi}(t)\left[\left(r+\eta_A'\diag(z(t))\pi_A(t)\right)dt + \pi_A(t)'\diag(\sqrt{z(t)})dW(t)\right] 
		\end{align}
		and $\pi(t) \in K_{PCSV} \Leftrightarrow \pi_A(t)\in \bigtimes_{i=1}^d [0, \sqrt{\beta_i}].$ 
		%Further, as $\pi_A$ is just a linear transformation of $\pi,$ we know that $\pi \in \Lambda_{PCSV} \Leftrightarrow \pi_A\in \Lambda_{PCSV}.$ 
		In particular, we may equivalently rewrite the portfolio optimisation problem in terms of $\pi_A$ as
		\begin{align*}
			\mathbf{(P_A)}\begin{cases}
				\Phi(v_0) = \underset{\pi \in \Lambda_A}{\sup} \mathbbm{E}\big[U(V^{v_0, \pi}(T)) \big] \\
				\Lambda_A = \big \{ \pi_A \in \Lambda_{PCSV} \ \big | \ \pi_A(t) \in \bigtimes_{i=1}^d [0, \sqrt{\beta_i}] \ \mathcal{L}[0,T]\otimes Q-\text{a.e.} \}
			\end{cases}.
		\end{align*}
		We proceed by solving $\mathbf{(P_A)}$ and then inverting the change of control to obtain a solution for the original optimisation problem $\mathbf{(P)}.$
	As $K_A = \bigtimes_{i=1}^d [0, \sqrt{\beta_i}]$ is a d-dimensional interval, $\mathbf{(P_A)}$ fits into the setting of Section 4.2 in \citet{EKZ2023} with $m=d$ and $\Sigma_i = 1$ for $i=1,..,d.$ According to Lemma 4.5 in \citet{EKZ2023}, if $A:[0,T]\rightarrow \R,$ $B:[0,T]\rightarrow \R^d$ with $A(0) = B_i(0) = 0$ for $i=1,...,d$ satisfy
	\begin{align}
		A'(\tau) &= br +  \sum_{i=1}^d \kappa_i \theta_i \nonumber \\
		B_i'(\tau) &= -\kappa_i B_i(\tau) + \frac{1}{2}\sigma_i\left(B_i(\tau)\right)^2 + \frac{1}{2}\frac{b}{1-b} \inf_{\lambda \in \R}\left\{ 2(1-b)\delta_{[0,\sqrt{\beta_i}]}(\lambda) + \left(\left(\eta_A\right)_i + \lambda_i + \sigma_i \rho_i B_i(\tau)\right)^2\right\},\label{eq: ODEs for B PCSV}
	\end{align}
	then the function $G(t,v,z):= \frac{1}{b}v^b\exp\left(A(T-t)+B(T-t)'z\right)$ is a solution to the HJB PDE associated with $\mathbf{(P_A)}.$ However, for each $i=1,...,d,$ equation (\ref{eq: ODEs for B PCSV}) contains the same ODE and optimisation that were considered in Lemma \ref{lem: ODE B for Box constraints} and Section \ref{sec: Heston's Stochastic Volatility Model}. Since Assumption \ref{ass: condition box constraints} is satisfied, the solution $B_i$ to (\ref{eq: ODEs for B PCSV}) is given by Theorem \ref{thm: solution constrained B in 1-D}. Finally, the candidate optimal portfolio $\pi^{\ast}_A$ for $\mathbf{(P)}$ is given by 
	\begin{align*}
		\left(\pi^{\ast}_A\right)_i(t)= \text{Cap}\left(\frac{1}{b}\left(\left(\eta_A\right)_i + \sigma_i\rho_i B_i(T-t)\right),0,\sqrt{\beta_i}\right).
	\end{align*}
	We still need to formally verify the optimality of $\pi^{\ast}_A$ for $\mathbf{(P_A)}.$ For this purpose, we define the sequence of stopping times  $\tau_{n,t}$ as $\tau_{n,t}=\min(T,\hat{\tau}_{n,t})$, with
	\begin{align*}
		\hat{\tau}_{n,t} = \inf \Big \{t\leq u \leq T \ \Big | \ &\int_t^u \left \Vert b \left(\sqrt{z(s)} \odot \pi_A(s)\right) G(s,V^{v_0, \pi^{\ast}}(s),z(s))\right \Vert^2ds \geq n, \\ 
		&\int_t^u  \left \Vert \left( \sigma \odot \sqrt{z(s)} \odot B(T-s) \right) G(s,V^{v_0, \pi^{\ast}}(s),z(s))\right \Vert ^2ds \geq n \Big \}.
	\end{align*}
	According to Theorem 3.12 in \citet{EKZ2023}, the optimality of $\pi^{\ast}_A$ is verified if
	$$
	\left(G\left(\tau_{n,t}, V^{v_0, \pi^{\ast}}(\tau_{n,t}), z(\tau_{n,t})\right) \right)_{n \in \N}
	$$
	is uniformly integrable for every $t \in [0,T].$ However, as Assumptions \ref{ass: condition box constraints} and \ref{ass: boundedness near maturity} are satisfied for every $i=1,...,d,$ following the same steps as in the proof of Theorem \ref{thm: strong verification in Heston market}, we can show that there exists a constant $q>1$ and a bounded, deterministic function $det(t)$ such that the local martingales
	\begin{align*}
		D_i(t) &= bq\sqrt{1-\rho_i^2}\int_0^t \left(\pi_A\right)_i^{\ast}(s)\sqrt{z_i(s)}d\hat{W}_i(s) \\
		M_i(t)&= \frac{q\kappa_i}{\sigma_i}\int_0^t \sqrt{z_i(s)}dW^{z}_i(s)
	\end{align*}
	can be used to bound $\left |G\left(t,V^{v_0,\pi^{\ast}}(t),z(t)\right) \right |^q$ for every $t\in[0,T]$ through 
	\begin{align*}
		\Big | G&\left(t,V^{v_0,\pi^{\ast}}(t),z(t)\right) \Big |^q  \leq  \exp \left( \sup_{s\in [0,T]} det(s) \right) \underbrace{\exp\left(\sum_{i=1}^d D_i(t) - \frac{1}{2}\langle D_i\rangle_t + M_i(t) - \frac{1}{2}\langle M_i\rangle_t\right)}_{(\ast)}.
	\end{align*}
	The diffusions of the local martingales $D_i,$ $M_i$ are independent and therefore $(\ast)$ is a supermartingale. Since the above bound holds pointwise for every $t\in [0,T],$ Doob's optional sampling theorem yields the $L^q$-boundedness (and thereby uniform integrability) of
	$$
	\Big(G(\tau^{0}_{n,t},V^{v_0, \pi^{\ast}}(\tau^{0}_{n,t}),z(\tau^{0}_{n,t}))\Big)_{n\in \N}.
	$$
	Hence, $\pi_A^{\ast}$ is optimal for $\mathbf{(P_A)}$ and thus $\pi^{\ast}(t) = A\pi^{\ast}_A(t)$ is optimal for $\mathbf{(P)}$ in $\mathcal{M}_{PCSV}.$
\end{proof}
%%%%%%%%%%%%%%
\vspace{0.25cm}
%%%%%%%%%%%%%%
\begin{proof}[Proof of Theorem \ref{thm: optimal portfolio vola-inverse constraints}]
	
		The proof follows similar arguments as the proof of Corollary 5.4 in \citet{Kraft2005}.
		For given $\pi\in \Lambda^{\gamma, \Sigma},$ let $\hat{\pi}_{BS}$ and $\hat{\pi}_H$ be defined through
		\begin{align}\label{eq: change of control mpr}
			\hat{\pi}_{BS}(t) := \Sigma(z(t))\pi(t) \qquad \text{or} \qquad \hat{\pi}_H(t) :=  \frac{\Sigma(z(t))}{\sqrt{z(t)}}\pi(t).
		\end{align}

	\underline{Proof of (i):} By construction, $\pi \in \Lambda_{K(\cdot)}$ if and only if $\hat{\pi}_{BS}(t) \in \hat{K}:=[\alpha,\beta]$ $\mathcal{L}[0,T]\otimes Q$-a.e., i.e. the constraint on $\hat{\pi}_{BS}$ is constant. Under the present assumptions, the wealth process satisfies
	\begin{align*}
		dV^{v_0,\pi}(t) &= V^{v_0,\pi}(t)\left([r + \eta \Sigma(z(t))\pi(t)]dt + \pi(t)\Sigma(z(t))dW(t)\right) \\
		&= V^{v_0,\pi}(t)\left([r + \eta \hat{\pi}_{BS}(t)]dt + \hat{\pi}_{BS}(t)dW(t)\right).
	\end{align*}
	Hence, maximising $\E \left[U\left(V^{v_0,\pi}(T)\right)\right]$ over $\hat{\pi}_{BS}$ subject to the constraint $\hat{\pi}_{BS}(t) \in \hat{K}=[\alpha,\beta]$ $\mathcal{L}[0,T]\otimes Q\text{-a.e.}$ is equivalent to the constrained portfolio optimisation problem in a Black-Scholes market $\mathcal{M}_{BS}$ with market price of risk $\eta$ and volatility $1.$ Therefore, as discussed in Section \ref{subsec: Comparison to Unconstrained Portfolio}, the constant-mix strategy
	\begin{align*}
		\hat{\pi}^{\ast}_{BS} = \text{Cap}\left(\pi_M,\alpha,\beta\right) \Leftrightarrow \pi^{\ast}(t) = \Sigma(z(t))^{-1} \text{Cap}\left(\pi_M,\alpha,\beta\right)
	\end{align*}
	maximises the expected utility over all admissible $\hat{\pi}_{BS}.$ Inverting the change of control through multiplication by $\Sigma(z(t))^{-1}$ yields the claim.
	\\
	\underline{Proof of (ii):} By construction, $\pi \in \Lambda_{K(\cdot)}$ if and only if $\hat{\pi}_H(t) \in \hat{K}:=[\alpha,\beta]$ $\mathcal{L}[0,T]\otimes Q$-a.e., i.e. the constraint on $\hat{\pi}_H$ is constant. Under the present assumptions, the wealth process satisfies
	\begin{align*}
		dV^{v_0,\pi}(t) &= V^{v_0,\pi}(t)\left([r + \eta \sqrt{z(t)} \Sigma(z(t))\pi(t)]dt + \pi(t)\Sigma(z(t))dW(t)\right) \\
		&= V^{v_0,\pi}(t)\left([r + \eta z(t) \hat{\pi}_H(t)]dt + \hat{\pi}_H(t)\sqrt{z(t)}dW(t)\right).
	\end{align*}
	Hence, maximising $\E \left[U\left(V^{v_0,\pi}(T)\right)\right]$ over $\hat{\pi}_H$ subject to the constraint $\hat{\pi}_H(t) \in \hat{K}=[\alpha,\beta]$ $\mathcal{L}[0,T]\otimes Q\text{-a.e.}$ is equivalent to the constrained portfolio optimisation problem in the Heston market $\mathcal{M}_H,$ as considered in Section \ref{sec: Heston's Stochastic Volatility Model}. As all requirements of Theorem \ref{thm: strong verification in Heston market} are satisfied, $\pi^{\ast}$ (as defined in Theorem \ref{thm: strong verification in Heston market}) maximises the expected utility over all admissible $\hat{\pi}_{H}.$ Inverting the change of control through multiplication by $\sqrt{z(t)}\Sigma(z(t))^{-1}$ yields the claim.
\end{proof}
%%%%%%%%%%%%%%
\vspace{0.25cm}
%%%%%%%%%%%%%%

\newpage 
\section{Supplementary Results}\label{sec: App. Supplementary Results}

\begin{lemma}\label{lem: convexity implies Lipschitz continuity}
	\begin{itemize}
		\item [(i)] Let $f:\R\rightarrow \R$ be a convex function. Then, for every $x_0\in \R,$ $\epsilon >0,$ there exists an $L>0$ such that if $|x-x_0|<\epsilon,$ $|y-x_0|<\epsilon,$ then 
		$$
		|f(x)-f(y)|<L|x-y|.
		$$
		In particular, f is locally Lipschitz continuous.
		\item[(ii)] Consider an interval $I\subset \R$ and let $f_1,f_2:I\rightarrow \R$ be Lipschitz continuous on $I.$ Then, \linebreak $f:= min(f_1,f_2)$ is Lipschitz continuous on $I.$
	\end{itemize}
\end{lemma}
\begin{proof}[Proof of Lemma \ref{lem: convexity implies Lipschitz continuity}]
	\textcolor{white}{1}\\
	\underline{Proof of (i):} Since $f$ is convex on $\R,$ f is continuous and thereby bounded on any compact subset of $\R.$ In particular, there exists $m,M\in \R$ such that 
	$$m\leq f(x)\leq M \quad \text{for all } x \ \text{with } |x-x_0|\leq \epsilon.$$
	Hence, the statement follows from Section A, Lemma 3.1.1 in \citet{Hiriart-Urruty2001}. \\
	\underline{Proof of (ii):} Let $L_1, L_2$ denote the Lipschitz constant of $f_1,$ $f_2$ respectively. For any real numbers $x,y \in \R$ recall the identity
	\begin{align}\label{eq: identity minimum}
		\min(x,y) = \frac{1}{2}\left(x+y-|x-y|\right).
	\end{align}
	Therefore, for any $x,y\in I$ we have
	\begin{align*}
		| f(x)-f(y) | &\hspace{1mm}= \hspace{1mm}\Big | \min\left(f_1(x),f_2(x)\right)- \min\left(f_1(y),f_2(y)\right)\Big | \\
		&\overset{(\ref{eq: identity minimum})}{=}\frac{1}{2}\Big |f_1(x)+f_2(x)-|f_1(x)-f_2(x)|- f_1(y)-f_2(y)+|f_1(y)-f_2(y)|\Big | \\
		&\hspace{1mm} \leq \hspace{1mm}\frac{1}{2} \left(|f_1(x)-f_1(y)| + |f_2(x)-f_2(y)| + \Big | |f_1(x)-f_2(x)| - |f_1(y)-f_2(y)| \Big | \right) \\
		&\hspace{1mm} \leq \hspace{1mm} \frac{1}{2} \big(|f_1(x)-f_1(y)| + |f_2(x)-f_2(y)| +  | f_1(x)-f_2(x) - \left(f_1(y)-f_2(y)\right)  | \big) \\
		& \hspace{1mm} = \hspace{1mm} \frac{1}{2} \big(|f_1(x)-f_1(y)| + |f_2(x)-f_2(y)| + \underbrace{\left | f_1(x)-f_2(x) - f_1(y)+f_2(y) \right |}_{\leq |f_1(x)-f_1(y)| + |f_2(x)-f_2(y)|} \big) \\
		&\hspace{1mm} \leq \hspace{1mm} \underbrace{|f_1(x)-f_1(y)|}_{\leq L_1 |x-y|} + \underbrace{|f_2(x)-f_2(y)|}_{\leq L_2|x-y|} \\
		&\hspace{1mm}\leq \hspace{1mm} (L_1+L_2)|x-y|.
	\end{align*}
	Hence, $f$ is Lipschitz continuous on $I$ with Lipschitz constant $L=L_1+L_2.$
\end{proof}

\begin{theorem}\label{thm: Picard-Lindelöf}
	Let $f:\R\rightarrow \R$ be locally Lipschitz continuous on $\R$ and $\tau_0,B_0 \in \R$ be constants. Consider the ordinary differential equation
	\begin{align}\label{eq: autonomous ODE}
		B'(\tau) = f(B(\tau)), \quad B(\tau_0) = B_0. 
	\end{align}
	Then there exists an $\epsilon > 0$ such that the ODE (\ref{eq: autonomous ODE}) has a unique solution $B(\tau)$ for $\tau \in (\tau_0-\epsilon, \tau_0+\epsilon).$
\end{theorem}
\begin{proof}[Proof of \ref{thm: Picard-Lindelöf}]
	This statement follows directly from Chapter \S6, \enquote{VII. Existence and Uniqueness Theorem} in \citet{Walter1998}.
\end{proof}

\begin{lemma}\label{lem: Monotonicity of Autonomous ODEs}
	Consider the setting of Theorem \ref{thm: Picard-Lindelöf} and an open interval $I\subset \R$ with $\tau_0 \in I,$ and let $B:I\rightarrow \R$ satisfy ODE (\ref{eq: autonomous ODE}). Then, $B$ is a monotone function in $\tau.$ If $f(B_0)\neq 0,$ then $B$ is strictly monotone in $\tau.$ Moreover, if $f(B_0) = 0,$ then $B$ is constant in $\tau.$ 
\end{lemma}
\begin{proof}[Proof of Lemma \ref{lem: Monotonicity of Autonomous ODEs}]
	We first consider the case $f(B_0)\neq 0.$ The following argument proceeds by contradiction. Assume that $B$ is not a strictly monotone function of $\tau.$ Since $f$ is locally Lipschitz continuous and therefore continuous, $B$ is continuously differentiable and there must exist $\bar{\tau}\in I$ with $B(\bar{\tau})=\bar{B}$ such that 
	$$
	B'(\bar{\tau}) = f(\bar{B}) = 0 \text{ and } B'(\tau)\neq 0 \text{ for all } \tau \text{ with } |\tau-\tau_0|<|\bar{\tau}-\tau_0|.
	$$ 
	In particular, $B$ is also a solution to the ODE
	\begin{align}\label{eq: stationary autonomous ODE}
		B'(\tau)=f(B(\tau)),\quad B(\bar{\tau})=\bar{B}.
	\end{align}
	On the other hand, as $f(\bar{B}) = f(B(\bar{\tau})) = B'(\bar{\tau})= 0,$ the function $\hat{B}(\tau):= \bar{B}$ is also a solution to the ODE (\ref{eq: stationary autonomous ODE}). However, for all $\tau$ with $|\tau -\tau_0|<|\bar{\tau}-\tau_0|$ we know that $f(B(\tau))=B'(\tau) \neq 0$ and thus $B(\tau)\neq \hat{B}(\tau).$ This is a contradiction to Theorem \ref{thm: Picard-Lindelöf}. Therefore, $B$ must be strictly monotone in $\tau.$ \\
	If $f(B_0)=0,$ by an analogous argument, $B(\tau) = B_0$ is the unique solution to ODE (\ref{eq: autonomous ODE}) and therefore $B$ is constant in $\tau.$
\end{proof}

\begin{lemma}\label{lem: Riccati ODE Filipovic}
	Consider a terminal time point $T>0$, real coefficients $B_0$, $r_0, \ r_1, \ r_2$ and the following Riccati-ODE 
	\begin{align}\label{eq: Riccati Filipovic ODE}
		B'(\tau) = -r_0 + r_1B(\tau) + \frac{1}{2} r_2B(\tau)^2, \quad B(0) = {B_0},
	\end{align}
	where $r_2 \neq 0,$ $r_1^2+2r_0r_2 > 0$ and define $r_3 = \sqrt{r_1^2+2r_0r_2}.$
	\begin{itemize}
		\item[(i)] The function 
		\begin{align}\label{eq: solution Riccati ODE Filipovic}
			B(\tau) = \frac{2r_2r_3B_0+\left(e^{r_3\tau}-1\right)(r_1+r_3)\left(r_1+r_2B_0-r_3\right)}{2r_2r_3-r_2\left(e^{r_3\tau}-1\right)\left(r_1+r_2B_0-r_3\right)}
		\end{align}
		is the unique solution of equation (\ref{eq: Riccati Filipovic ODE}) on its maximal interval  $[0,t_{+}(B_0))$ with life-time $t_{+}(B_0)>0.$ Moreover, for any $T\in [0,t_{+}(B_0))$ we have
		\begin{align}´\label{eq: integrated Riccati Filipovic}
			\int_0^T B(\tau)d\tau = \frac{2}{r_2}\ln \left(\frac{2r_3 e^{\frac{r_3 - B_0}{2}T}}{r_3 \left(e^{r_3 T}+1\right)-r_1\left(e^{r_3 T}-1\right)-r_2\left(e^{r_3 T}-1\right)B_0}\right).
		\end{align}
		\item[(ii)] The life time $t_{+}(B_0)$ of $B$ (in the sense of Lemma 10.1 in \citet{Filipovic2009}) is given as
		\begin{align*}
			t_{+}(B_0) &= \begin{cases}
				\frac{1}{r_3} \ln \left(\frac{r_1 +r_2B_0 + r_3}{r_1 + r_2B_0 -r_3}\right),& \quad \text{if} \ r_1+r_2B_0 -r_3 > 0 \\
				\infty, & \quad \text{if} \ r_1+r_2B_0 -r_3 \leq 0. \\
			\end{cases}
		\end{align*}
	\end{itemize}
\end{lemma}
\begin{proof}[Proof of Lemma \ref{lem: Riccati ODE Filipovic}]
	Statement (i) follows directly from Lemma 10.12 in \citet{Filipovic2009} with $A = \frac{1}{2}r_2,$ $B=r_1,$ $C=r_0,$ $u= B_0$ and
	\begin{align*}
		& -\frac{2r_0\left(e^{r_3 \tau}-1\right)-\left(r_3\left(e^{r_3 \tau}+1\right)+r_1\left(e^{r_3 \tau}-1\right)\right)B_0}{r_3\left(e^{r_3 \tau}+1\right) - r_1\left(e^{r_3 \tau}-1\right)-r_2\left(e^{r_3 \tau}-1\right)B_0} \\
		& \quad = -\frac{-2r_3B_0 + \left(e^{r_3 \tau}-1\right)\left(2r_0 - \left(r_1+r_3\right)B_0\right)}{2r_3 - \left(e^{r_3 \tau}-1\right)\left(r_1+r_2B_0-r_3\right)} \\
		& \quad = \frac{2r_2r_3B_0 - \left(e^{r_3 \tau}-1\right)\big(\overbrace{2r_0r_2}^{=r_3^2-r_1^2} - \left(r_1+r_3\right)r_2B_0\big)}{2r_2r_3 - r_2\left(e^{r_3 \tau}-1\right)\left(r_1+r_2B_0-r_3\right)} \\
		& \quad = \frac{2r_2r_3B_0 - \left(e^{r_3 \tau}-1\right)\left(\left(r_3-r_1\right)\left(r_3+r_1\right) - \left(r_1+r_3\right)r_2B_0\right)}{2r_2r_3 - r_2\left(e^{r_3 \tau}-1\right)\left(r_1+r_2B_0-r_3\right)} \\
		& \quad = \frac{2r_2r_3B_0+\left(e^{r_3\tau}-1\right)(r_1+r_3)\left(r_1+r_2B_0-r_3\right)}{2r_2r_3-r_2\left(e^{r_3\tau}-1\right)\left(r_1+r_2B_0-r_3\right)}.
	\end{align*}
	The function $B$ as in (\ref{eq: solution Riccati ODE Filipovic}) is well-defined and continuously differentiable as long as the denominator in (\ref{eq: solution Riccati ODE Filipovic}) is non-zero. In particular, taking derivatives verifies that $B$ satisfies (\ref{eq: Riccati Filipovic ODE}) as long as its denominator is non-zero. This is true for $\tau = 0,$ when the denominator is equal to $r_2r_3 > 0,$ and the denominator stays positive if $r_1+r_2B_0 -r_3 \leq 0,$ i.e. $t_{+}(B_0) = \infty$. However, if $r_1+r_2B_0 -r_3>0,$ then 
	\begin{alignat*}{3}
		&\quad 0 &&\overset{!}{=} \quad && 2r_2r_3-r_2\left(e^{r_3t_{+}(B_0)}-1\right)\left(r_1+r_2B_0-r_3\right) \\
		\Leftrightarrow & \quad 2r_3 &&= \quad && \left(e^{r_3t_{+}(B_0)}-1\right)\left(r_1+r_2B_0-r_3\right) \\
		\Leftrightarrow & \quad t_{+}(B_0) &&= \quad && \frac{1}{r_3}\ln \left( \frac{r_1+r_2B_0 +r_3}{r_1+r_2B_0-r_3}\right).
	\end{alignat*}
\end{proof}

\begin{corollary}\label{cor: transition times Filipovic}
	Consider the setting of Lemma \ref{lem: Riccati ODE Filipovic}, let $B$ be as in (\ref{eq: solution Riccati ODE Filipovic}) and let $\hat{B} \in \R$ be given. \\
	If 
	\begin{align*}
		\tau_{\hat{B}} := \frac{1}{r_3}\ln{\left(\frac{2r_2r_3\left(\hat{B}-B_0\right) + \left(r_1+r_2B_0-r_3\right)\left(r_1+\hat{B}r_2+r_3\right)}{\left(r_1+r_2B_0-r_3\right)\left(r_1+\hat{B}r_2+r_3\right)} \right)} \leq t_{+}(B_0)
	\end{align*}
	then $B(\tau_{\hat{B}}) =\hat{B}$. 
\end{corollary}
\begin{proof}
	If $\tau_{\hat{B}} < t_{+}(B_0)$, then
	\begin{align*}
		B(\tau_{\hat{B}}) =  \frac{2r_2r_3B_0+\left(e^{r_3\tau_{\hat{B}} }-1\right)(r_1+r_3)\left(r_1+r_2B_0-r_3\right)}{2r_2r_3-r_2\left(e^{r_3\tau_{\hat{B}} }-1\right)\left(r_1+r_2B_0-r_3\right)}.
	\end{align*}
	This in turn implies
	\begin{align*}
		&\hat{B} \overset{!}{=} B(\tau_{\hat{B}})= \frac{2r_2r_3B_0+\left(e^{r_3\tau_{\hat{B}} }-1\right)(r_1+r_3)\left(r_1+r_2B_0-r_3\right)}{2r_2r_3-r_2\left(e^{r_3\tau_{\hat{B}} }-1\right)\left(r_1+r_2B_0-r_3\right)} \\
		\Leftrightarrow \quad & \hat{B}\left(2r_2r_3-r_2\left(e^{r_3\tau_{\hat{B}}}-1\right)\left(r_1+r_2B_0-r_3\right)\right) \\
		& \quad = 2r_2r_3B_0 + r_2\left(e^{r_3\tau_{\hat{B}}}-1\right)\left(r_1+r_3\right)\left(r_1+r_2B_0-r_3\right) \\
		\Leftrightarrow \quad & 2r_2r_3\left(\hat{B}-B_0\right) = \left(e^{r_3\tau_{\hat{B}}}-1\right)\left(r_1+r_2B_0-r_3\right)\left(r_1+\hat{B}r_2+r_3\right) \\
		\Leftrightarrow \quad & \left(e^{r_3\tau_{\hat{B}}}-1\right) =\frac{2r_2r_3\left(\hat{B}-B_0\right)}{\left(r_1+r_2B_0-r_3\right)\left(r_1+\hat{B}r_2+r_3\right)} \\
		\Leftrightarrow \quad & \tau_{\hat{B}} = \frac{1}{r_3}\ln{\left(\frac{2r_2r_3\left(\hat{B}-B_0\right)}{\left(r_1+r_2B_0-r_3\right)\left(r_1+\hat{B}r_2+r_3\right)} +1 \right)} \\
		& \hspace{1.5mm }= \frac{1}{r_3}\ln{\left(\frac{2r_2r_3\left(\hat{B}-B_0\right) + \left(r_1+r_2B_0-r_3\right)\left(r_1+\hat{B}r_2+r_3\right)}{\left(r_1+r_2B_0-r_3\right)\left(r_1+\hat{B}r_2+r_3\right)} \right)}
	\end{align*}
\end{proof}

\begin{lemma}\label{lem: utility sub-optimal strategy}
	Consider a bounded deterministic relative portfolio process $\pi\in \Lambda_K.$ Let $\Ap, \ \Bp:[0,T]\rightarrow \R$ with $\Ap(0) = \Bp(0) = 0$ be solutions to the system of ODEs
	\begin{align}
		\Ap'(T-t) &= rb + \kappa \theta \Bp(T-t) \label{eq: general Riccati A}\\
		\Bp'(T-t) &= -\left[\frac{1}{2}\pi(t)^2b(1-b)-b\eta\pi(t)\right]+\left[\sigma \rho b \pi(t)-\kappa\right]\Bp(T-t) + \frac{1}{2}\sigma^2\Bp(T-t)^2. \label{eq: general Riccati B}
	\end{align}
	If $\Jp$ (as defined in (\ref{eq: expected utility functional})) is the unique solution to the Feynman-Kac-PDE (omitting the argument $(t,v,z)$ for readability)
	\begin{align}\label{eq: Feynman-Kac PDE}
		0 = \Jp_t + \left(r+ \eta \pi(t)z\right)v\Jp_v + \kappa\left(\theta-z\right)\Jp_z + \sigma\rho\pi(t)zv\Jp_{zv} +\frac{1}{2}v^2\pi(t)^2z\Jp_{vv} + \frac{1}{2}\sigma^2z\Jp_{zz},
	\end{align}
	with boundary condition $\Jp(T,v,z)=U(v)=\frac{1}{b}v^b,$ then 
	$$
	\Jp(t,v,z) = \frac{1}{b}v^b\exp\left(\Ap(T-t)+\Bp(T-t)z\right).
	$$
\end{lemma}
\begin{proof}
	We first verify that $G(t,v,z) = \frac{1}{b}v^b\exp(\Ap(T-t)+\Bp(T-t)z)$ is a solution to the Feynman-Kac PDE (\ref{eq: Feynman-Kac PDE}). The partial derivatives of $G$ can be computed as
	\begin{align*}
		G_t(t,v,z)&= -\left(\Ap'(T-t) + \Bp'(T-t)z\right)G(t,v,z), \quad G_v(t,v,z)= \frac{b}{v}G(t,v,z), \\
		G_{vv}(t,v,z)&= \frac{b(b-1)}{v^2}G(t,v,z), \quad 	G_{z}(t,v,z)= \Bp(T-t)G(t,v,z), \\
		G_{zz}(t,v,z)&= \left(\Bp(T-t)\right)^2G(t,v,z), \quad \text{and} \quad G_{zv}(t,v,z)= \frac{b\Bp(T-t)}{v}G(t,v,z).
	\end{align*}
	Substituting these derivatives in (\ref{eq: Feynman-Kac PDE}), while omitting the arguments $(t,v,z),$ yields
	\begin{align*}
		0 &= G_t + \left(r+ \eta \pi(t)z\right)vG_v + \kappa\left(\theta-z\right)G_z + \sigma\rho\pi(t)zvG_{zv} + \frac{1}{2}v^2\pi(t)^2zG_{vv} + \frac{1}{2}\sigma^2zG_{zz} \\
		&= G\Big[-\left(\Ap' + \Bp'z\right)+ \left(r+ \eta \pi(t)z\right)b + \kappa\left(\theta-z\right)\Bp \\
		& \hspace{1.5cm} + \sigma\rho\pi(t)bz\Bp  + \frac{1}{2}\pi(t)^2zb(b-1) + \frac{1}{2}\sigma^2z\Bp^2\Big] \\
		\overset{G\neq 0}{\Leftrightarrow} \quad 0 &= -\left(\Ap' + \Bp'z\right)+ \left(r+ \eta \pi(t)z\right)b + \kappa\left(\theta-z\right)\Bp \\
		& \hspace{1.5cm} + \sigma\rho\pi(t)bz\Bp  + \frac{1}{2}\pi(t)^2zb(b-1) + \frac{1}{2}\sigma^2z\Bp^2 \\
		&= \underbrace{-\Ap' + rb + \kappa \theta \Bp}_{= 0 \ \text{by (\ref{eq: general Riccati A})}}  + z\underbrace{\left[-\Bp' + b\eta \pi(t)- \frac{1}{2}\pi(t)^2b(1-b) + \left(\sigma \rho b \pi(t)-\kappa\right)\Bp + \frac{\sigma^2}{2}\Bp^2\right]}_{= 0 \ \text{by (\ref{eq: general Riccati B})}}.
	\end{align*}
	Moreover, $G$ satisfies the boundary condition $G(T,v,z) =U(v)$ because $\Ap(0) = \Bp(0) = 0.$ Hence, \linebreak
	$G$ is a solution to the Feynman-Kac PDE (\ref{eq: Feynman-Kac PDE}). By assumption, such a solution is unique, and therefore 
	$$\frac{1}{b}v^b\exp\left(\Ap(T-t)+\Bp(T-t)z\right)=G(t,v,z) = \Jp(t,v,z) \qquad \forall (t,v,z)\in [0,T]\times(0,\infty)\times(0,\infty).$$
\end{proof}

\begin{corollary}\label{cor: Wealth-Equivalent Loss Deterministic Strategy}
	Let $\pi\in \Lambda_K$ be a deterministic portfolio process. Let $\Ap,$ $\Bp$ be as in Lemma \ref{lem: utility sub-optimal strategy} and $A,$ $B$ as in Lemma \ref{lem: dual ODEs Heston Model MK23}. Then, the wealth-equivalent loss $L^{\pi}$ is given as 
	\begin{align*}
		L^{\pi}(t,z) = 1-\exp \left(\frac{1}{b}\left(\Ap(T-t)-A(T-t) + \left[\Bp(T-t)-B(T-t)\right]z\right)\right) \qquad \forall (t,z)\in [0,T]\times(0,\infty). 
	\end{align*}
\end{corollary}
\begin{proof}
	$L^{\pi}(t,z)$ is defined as the solution to the equation
	\begin{align*}
		\Phi(t,v(1-L^{\pi}(t,z)),z) &= \E \left[ U\left(V^{v_0,\pi^{\ast}}(T)\right) \ | \ V^{v_0,\pi}(t)=v(1-L^{\pi}(t,z)), \ z(t)=z \right] \\
		&= \E \left[ U\left(V^{v_0,\pi}(T)\right) \ | \ V^{v_0,\pi}(t)=v, \ z(t)=z \right]  \\
		&= 	\Jp(t,v,z)
	\end{align*}
	From Lemma \ref{lem: optimal portfolio Heston MK23} we know that
	\begin{align*}
		\Phi(t,v(1-L^{\pi}(t,z)),z) = \frac{1}{b}\left(v(1-L^{\pi}(t,z))\right)^b\exp(A(T-t) + B(T-t)z).
	\end{align*}
	Similarly, we have by Lemma \ref{lem: utility sub-optimal strategy}
	\begin{align*}
		\Jp(t,v,z) =\frac{1}{b}v^b\exp\left(\Ap(T-t)+\Bp(T-t)z\right).
	\end{align*}
	Thus,
	\begin{align*}
		&\Phi(t,v(1-L^{\pi}(t,z)),z) = \Jp(t,v,z) \\
		\Leftrightarrow \quad & \frac{1}{b}\left(v(1-L^{\pi}(t,z))\right)^b\exp(A(T-t) + B(T-t)z) = \frac{1}{b}v^b\exp\left(\Ap(T-t)+\Bp(T-t)z\right)  \\
		\Leftrightarrow \quad & \left(\frac{v(1-L^{\pi}(t,z))}{v}\right)^b = \exp \left(\Ap(T-t)-A(T-t) + \left[\Bp(T-t)-B(T-t)\right]z\right) \\
		\Leftrightarrow \quad &L^{\pi}(t,z) =  1-\exp \left(\frac{1}{b}\left(\Ap(T-t)-A(T-t) + \left[\Bp(T-t)-B(T-t)\right]z\right)\right).
	\end{align*}
\end{proof}

\end{document}